\documentclass[%
reprint,
onecolumn,
notitlepage,
unsortedaddress,
nofootinbib,
amsmath,amssymb,
aps,
]{revtex4-1}
\pdfoutput=1
\usepackage[utf8]{inputenc}
\usepackage[english]{babel}
\usepackage[T1]{fontenc}
\usepackage{amsmath}
\usepackage{amsthm}
\usepackage{hyperref}

\usepackage{tikz}
\usepackage{lipsum}

\usepackage{dsfont}
\usepackage[normalem]{ulem}
\usepackage{mathtools}
\usepackage[makeroom]{cancel}
\usepackage{bbm}
\usepackage{enumitem}

\usepackage[draft]{changes}
\usepackage{wrapfig}

\usepackage{color}
\usepackage{graphicx}
\usepackage{amsthm}
\usepackage{tikz-cd}
\usepackage{enumitem}

\definecolor{ms}{rgb}{0,.4,1}
\newcommand{\ms}[1]{{\color{black}#1}}
\definecolor{henrik}{rgb}{1,.4,0}

\newcommand{\mb}[1]{\mathbb{#1}}

\newcommand{\Tr}[1]{\mathrm{Tr}\left[ #1\right]} 

\newcommand{\id}{\mathbbm{1}}

\newcommand{\R}{\mb{R}}

\newcommand{\ket}[1]{\left.\left|{#1}\right.\right\rangle}

\newcommand{\bra}[1]{\left.\left\langle{#1}\right.\right|}

\newcommand{\braket}[2]{\left\langle #1 \middle| #2 \right\rangle}

\newcommand{\ketbra}[2]{\ket{#1} \!\! \bra{#2}}

\newcommand{\bigo}[1]{\mathcal{O}\left (#1\right)}

\newcommand{\rom}[1]{\uppercase\expandafter{\romannumeral #1\relax}}
\newcommand{\norbra}[1]{\left( #1\right)}
\newcommand{\sqrbra}[1]{\left[ #1\right]}
\newcommand{\curbra}[1]{\left\{ #1\right\}}
\newcommand{\lind}{\mathcal{L}}
\newcommand{\J}{\mathbb{J}}

\usepackage{amssymb}             
\newcommand{\dt}{{\rm d}t\,}
\newcommand{\de}{{\rm d}}

\newcommand{\RR}{\mathbb{R}}
\newcommand{\LL}{\mathbb{L}}
\newcommand{\TrR}[2]{\mathrm{Tr}_{#1}\left[ #2\right]}
\newcommand{\idO}{\mathds{I}}
\newcommand{\norbraB}[1]{\boldsymbol{(}#1\boldsymbol{)}}

\newtheorem{lemma}{Lemma}
\newtheorem{theorem}{Theorem}
\newtheorem{corollary}{Corollary}[theorem]
\theoremstyle{definition}
\newtheorem{definition}{Definition}
\usepackage{algorithm}
\usepackage{algpseudocode}
\usepackage{comment}

\newcounter{counterN}
\newcommand{\ord}[1]{{\mathcal{O}\left(#1\right)}}
\newcommand{\cJ}{\mathcal{J}}
\newcommand{\DD}{\mathcal{D}}
\newtheorem*{theorem*}{Theorem}

\usepackage[normalem]{ulem}

\usepackage{tcolorbox}
\tcbuselibrary{breakable}
\tcbuselibrary{skins}
\definecolor{my-green}{RGB}{0,144,81}
\definecolor{my-red}{RGB}{255,113,91}
\tcbset{enhanced jigsaw, colback=my-green!10!white,colframe=my-green!85!black,fonttitle=\bfseries}

\newenvironment{additional-info}[1]{ 
	\begin{tcolorbox}[breakable,title=\textbf{#1},width=\textwidth]}{\end{tcolorbox}}
\newenvironment{additional-proof}[1]{ 
	\begin{tcolorbox}[breakable,colback=my-red!15!white,colframe=my-red!85!black,title=\textbf{#1},width=\textwidth]}{\end{tcolorbox}}

\begin{document}
	
	\title{Quantum Fisher Information and its dynamical nature}
	
	\author{Matteo Scandi}
        \affiliation{Institute for Cross-Disciplinary Physics and Complex Systems (IFISC) UIB-CSIC, Campus Universitat Illes Balears,E-07122 Palma de Mallorca, Spain}
	\affiliation{ICFO - Institut de Ciencies Fotoniques, The Barcelona Institute of\\
		Science and Technology, Castelldefels (Barcelona), 08860, Spain}
	
	\author{Paolo Abiuso}
	\affiliation{Institute for Quantum Optics and Quantum Information - IQOQI Vienna,
		Austrian Academy of Sciences, Boltzmanngasse 3, A-1090 Vienna, Austria}
	\affiliation{ICFO - Institut de Ciencies Fotoniques, The Barcelona Institute of\\
		Science and Technology, Castelldefels (Barcelona), 08860, Spain}
	
	\author{Jacopo Surace}
	\affiliation{Perimeter Institute for Theoretical Physics,
		31 Caroline Street North, Waterloo, ON N2L 2Y5, Canada}
	\affiliation{ICFO - Institut de Ciencies Fotoniques, The Barcelona Institute of\\
		Science and Technology, Castelldefels (Barcelona), 08860, Spain}

	\author{Dario De Santis}
	\affiliation{Scuola Normale Superiore, I-56126 Pisa, Italy}
	\affiliation{ICFO - Institut de Ciencies Fotoniques, The Barcelona Institute of\\
		Science and Technology, Castelldefels (Barcelona), 08860, Spain}

	\begin{abstract}
		The importance of the Fisher information metrics and its quantum generalisations is testified by the number of applications that this has in very different fields, ranging from hypothesis testing to  metrology, passing through thermodynamics. Still, from the rich range of possible quantum Fisher informations, only a handful are typically used and studied. This review aims at collecting a number of results scattered in the literature and provide a cohesive treatment to people who begin the study of Fisher information and to those who are already working on it to have a more organic understanding of the topic. Moreover, we complement the review with new results about the relation between Fisher information and physical evolutions. 
		Extending the analysis of~\cite{abiuso2022characterizing,scandi2023physicality},
		we show that dynamical properties such as (complete) positivity, Markovianity, detailed balance, retrodictive power of evolution maps can be caracterized in terms of their relation with respect to the Fisher information metrics. These results show a fact that was partially overseen in the literature, namely the inherently dynamical nature of Fisher information.
	\end{abstract}
	\maketitle
	
	
	
	
	\newpage
	\begin{tcolorbox}[breakable,colback=ms!15!white,colframe=ms!85!black,title=\textbf{Index},width=\textwidth]\label{index}
		\hyperref[introduction]{\textbf{I. Introduction}}\\
		~\\
		\hyperref[contrastFunctionsToFisher]{\textbf{II. From quantum contrast functions to quantum Fisher information}}\\
		\emph{Generalised quantum divergences (contrast functions) are introduced, and their general properties  discussed. Thm.~\ref{cf:thm:Ruskai} gives a characterisation of their local behaviour in terms of a family of functionals. The same objects also appear in the characterisation of monotone metrics on the space of states (Thm.~\ref{cf:thm:Petz}). These are the two ways of defining the family of quantum Fisher information metrics. The informative boxes of this section provide more details about the relation between contrast functions and their $g$ function (\hyperref[box:ContrastFunctions]{Box 1}),  a first characterisation of standard monotones (\hyperref[box:standardMonotones]{Box 2}) and  other possible quantifiers of statistical difference (\hyperref[box:statisticalQuantifiers]{Box 3}).}\\
		~\\
		\hyperref[FisherDynamics]{\textbf{III. Dynamical properties of Fisher information}}\\
		\emph{The equivalence between the positivity of linear maps and the contraction of Fisher information is proved (Thm.~\ref{cf:theo:CPiffContracts}). This implies that one can define the physicality of an evolution by looking at whether it contracts everywhere the Fisher information on the space of states tensor an ancilla of the same dimension (Corollary~\ref{cf:cor:CPiffContracts}). In \hyperref[box:cpIffContracts]{Box 4} the same results are derived by considering contrast functions.}
		\begin{enumerate}[label=\textbf{\Alph*.}]
			\item \hyperref[markovianity]{\textbf{Markovianity as monotonous contraction of Fisher information}}\\
			\emph{A novel expression of the Fisher information currents is presented (Thm.~\ref{cf:thm:FisherFlow}). The relation between Fisher information and Markovianity is completely characterised: it is shown that the monotonous contraction of Fisher information on the space of states (adjoined with an ancilla of the same dimension) is equivalent to the Markovianity of the corresponding evolution (Thm.~\ref{cf:theo:markov_iff_contract}). Despite this deep connection between the two notions, it is impossible to  witness non-Markovianity through the expansion of Fisher information for all evolutions only by the use of extra ancillas and copies of the channel (Thm.~\ref{cf:theo:no-go}). Still, if one allows for a post-processing at the end of the dynamics, it is also shown that one can actually provide an operational witness (Thm.~\ref{cf:theo:witness}).}
			\item \hyperref[retrodiction]{\textbf{Retrodiction and Fisher information}}\\
			\emph{ A generalised version of the Petz recovery map is introduced, and it is shown how this allows to canonically map evolved states (close to a prior)  back into their initial condition. Markovianity is then related to the behaviour of the error committed during this procedure (Thm.~\ref{cf:theo:recovery_fisher}). In~\hyperref[box:contractionCoefficients]{Box 9} we discuss the topic of universal recovery, and we give some preliminary results using $\chi^2_f$-divergences (Thm.~\ref{theo:chi2recoverybound} and corollaries thereof). Moreover, we find that the traditional Petz recovery map can be characterised in two different ways: either as the unique universal recovery map that is a quantum channel (Corollary~\ref{cor:eqChi2divergences} and~\ref{cor:generalisedUniversalRecovery}); or as the one whose spectrum dominates all other maps in the family (Thm.~\ref{thm:petzSup}).}
			\item \hyperref[detailedBalance]{\textbf{Fisher information and detailed balance}}\\
			\emph{The geometrical structure of quantum and classical dynamics satisfying detailed balance is discussed. 
				In Thm.~\ref{cf:thm:classicalDB} we show how this condition can naturally be formulated for classical systems in terms of the self-adjointness of the generator with respect to the Fisher scalar product. A similar result is also proved for quantum dynamics (Thm.~\ref{cf:thm:eqDB}) where we find a slightly more general form of the Lindbladian compared to the most used choice in the literature. }
		\end{enumerate}
		\hyperref[propertiesFisher]{\textbf{IV. Mathematical properties of  quantum Fisher information}}
		\begin{enumerate}[label=\textbf{\Alph*.}]
			\item \hyperref[sec:standardMonotonesProperties]{\textbf{The set of standard monotone functions}}\\
			\emph{The set of standard monotone functions is characterised, focusing on an important symmetry of the set and its convex structure, generated by a continuous family of vertices.}
			\item \hyperref[sec:FisherOperatorProperties]{\textbf{Properties of quantum Fisher operators}}\\
			\emph{The way in which the properties of the defining functions are mirrored on their corresponding Fisher operators is presented, giving a first general characterisation.}
			\item \hyperref[sec:CPfisher]{\textbf{Complete positivity of the Fisher information operators}}\\
			\emph{A full characterisation of the complete positive Fisher information operator is provided. In particular, Thm.~\ref{cf:thm:cpCharacterisation} gives the most general expression for such maps.}
		\end{enumerate}
		\hyperref[gardenFisher]{\textbf{V. A garden of quantum Fisher information}}\\
		\emph{We present the most notable examples of quantum Fisher operators, summarised in table~\ref{cf:fig:figstandardmonotones}, and in Fig.~\ref{cf:fig:figstandardmonotones2}.}\\
		
		\hyperref[conclusions]{\textbf{VI. Conclusions and open questions}}
	\end{tcolorbox}
	
	\section{Introduction}\label{introduction}
	
	The goal of physics is to find fixed laws in an evolving world. The discoveries of the last century, especially in statistical mechanics, but even more fundamentally in quantum theory,
	have shown that relaxing the concept of a  law from the determinism that characterised Newtonian mechanics to the more general notion of probabilistic relation not only allows for a richer expressibility, but ultimately for a more faithful representation of our experience of the world. The same change in perspective is also useful in the description of time evolutions, justifying the replacement of deterministic transitions with noisy transformations acting on probability distributions (or density matrices, in the case of quantum theory). For this reason it is no surprise that statistical considerations play such a central role in modern physics. 
	For example, when assessing the quality of a theory, one is confronted with the problem of quantifying the proximity between some experimental data and the predictions of a theoretical model.
	More fundamentally, quantifying the discrepancy between (classical and quantum) statistical mixtures is crucial in all areas of physics, in order to characterize the informational content of one's particular description of a physical phenomenon.
	
	Without further constraints, there is no unique procedure that can capture the many different ways in which probability distributions can differ. Still, if one restricts their attention to 
	regular metric structures,
	one can single out a unique family of statistical distances, namely the Fisher information one, only by imposing a simple request: the fact that the distinguishability between two probabilities should decrease under noisy transformations. This condition reflects the intuitive idea that if noise affects the dynamics of an experiment, the corresponding capability to distinguish different physical states will diminish as well. The fact that statistical distances can be solely identified from a requirement about their dynamics, which goes under the name of Chentsov-Petz Theorem~\cite{cencovStatisticalDecisionRules2000,campbell1986extended,petzMonotoneMetricsMatrix1996} (see Thm.~\ref{cf:thm:Petz}), is quite remarkable, as \emph{a priori} it is not obvious why there should be such a strong relation with time evolutions.
	
	This peculiarity was long time overlooked, eclipsed by the abundance of operational uses of the Fisher information in disparate fields:
	most famously, in the field of metrology the
	Cramér-Rao bound \cite{Cramer46,Rao92} quantifies the ultimate limits to parameter estimation via the Fisher information, both in the classical and quantum scenario \cite{paris2009quantum,sidhu2020geometric};
	in a related area, the Chernoff bounds characterise the asymptotics of multi-copy hypothesis testing~\cite{Chernoff52,audenaert2007discriminating}; quantum Fisher information has been used as well in the witnessing and characterization of entanglement~\cite{pezze2009entanglement,li2013entanglement} and critical behaviour of systems close to phase transitions~\cite{camposvenuti2007quantum,zanardi2007information,ma2009fisher}; applications of it have been studied in the fields of 
	computational complexity~\cite{rissanen1996fisher}, chemistry~\cite{nalewajski2008use}, and even evolution theory~\cite{frank2009natural}.
	
	Still, few works have appeared exploring the relation between Fisher metrics and the rate at which general dynamics degrade information (see for example~\cite{temmeH2divergenceMixingTimes2010, hiai2016contraction}, or more recently~\cite{gao2023sufficient}, which expands the results of~\cite{jungeUniversalRecoveryMaps2018a}  connecting the the amount of information loss 
	to the redtrodicting power of the Petz recovery map~\cite{petz1986sufficient}). 
	Here, following the investigations started in~\cite{abiuso2022characterizing,scandi2023physicality}, we go to the root of the interconnection between Fisher information metrics and physical evolutions: in Sec.~\ref{FisherDynamics} a dual of the Chentsov-Petz Theorem is proven, showing that (complete) positivity of a linear map can be defined solely in terms of the induced decrease of Fisher information (cf. also~\cite{scandi2023physicality}); this same tool can then be used to characterize Markovian evolutions (Sec.~\ref{markovianity}) and witnesses of non-Markovian effects; we then show how the quantum Fisher information operator formalism induces a natural notion of retrodiction for evolution maps, both in terms of retrodiction maps as well as in the quantification of information degradation in time (Sec.~\ref{retrodiction}); finally we study quantum detailed balance under the lenses of the Fisher information operators, analysing the geometry and structural characterisation of detailed-balanced Lindblad operators (Sec.~\ref{detailedBalance}). These identifications hint to some deep relation between statistical inference and time evolutions.
	
	Then, the goal of this work is twofold: on the one hand, we aim to motivate the reader to the study of the Fisher information by proving its intimately dynamical nature, showing how properties of physical evolutions can naturally be cast in this language (Sec.~\ref{FisherDynamics}); on the other, conscious of the mathematical technicalities that could scare away new practitioners, we provide a comprehensive guide to the  properties of quantum Fisher information while trying to keep the exposition as accessible as possible. To this end, we first introduce in Sec.~\ref{contrastFunctionsToFisher} the bare minimum formalism to discuss the dynamical properties of Fisher information, postponing the more mathematical discussions to the second part of the work, including the analysis of the most relevant formal properties of the quantum Fisher information operators in Sec.~\ref{propertiesFisher}. We point out in this context that Sec.~\ref{gardenFisher}, which contains a list of different expressions that the Fisher information can take, was designed to favour the sporadic consultation, and should be considered as a field guide to the many different forms that this quantity can take. Moreover, in order to keep the exposition as pedagogical as possible, we complement the main text with informative boxes where specific subjects are explored in greater detail. When not interested, readers can skip these boxes without compromising the understanding of the main text.

	\section{From quantum contrast functions to quantum Fisher information}\label{contrastFunctionsToFisher}
	
	There are two paths to arrive to the definition of quantum Fisher information metrics: one statistical and one dynamical. The first defines them as the metrics that arise from the local expansion of quantifiers of statistical difference, as it is explained in Thm.~\ref{cf:thm:Ruskai}. The other approach singles out the Fisher information family from all the possible metrics on the space of states, as the unique family that contracts under all physical evolutions (Thm.~\ref{cf:thm:Petz}). It is a truly remarkable fact that these two definitions actually define the same family, leaving us with the choice of which path to undertake. We preferred here to start with the statistical definition, in order to insert local results in a more global context. Moreover, this allows us to introduce the bare minimum formalism needed to understand the dynamical nature of Fisher information (Sec.~\ref{FisherDynamics}).
	
	We then have to ask ourselves how to assess the similarity of different classical statistical distributions. This problem does not have a straightforward answer: different methods yield different quantifications, and there is no clear argument that would single out a unique strategy over all the others. As a matter of fact, this difficulty actually reflects the many different behaviours that a probability distribution shows depending on the regime one is focusing on: two distributions could be very similar in the asymptotic regime, but show substantial differences when one restricts their attention to single-shot experiments. 
	For this reason, rather than trying to reduce the richness of the phenomenology that one can focus on,
	Csiszár generalised the usual relative entropy with an axiomatic construction based on the ansatz~\cite{csiszarInformationTheoryStatistics2004}: 
	\begin{align}
		H_g(\rho||\sigma) :=\Tr{\rho\,g(\sigma/\rho )}\,,\label{cf:eq:CsizarDivergence}
	\end{align}
	where $\rho$ and $\sigma$ are classical probability vectors\footnote{Here and in the rest of the work, we identify probability vectors with diagonal density matrices, so the expression in Eq.~\eqref{cf:eq:CsizarDivergence} should be understood as: $H_g(\rho||\sigma) :=\sum_i \;\rho_i\,g(\sigma_i/\rho_i)$.} $g$ is a convex function with domain over the positive reals, and finally we abuse the notation slightly by denoting with $\sigma/\rho$ the vector with components $\{\sigma_i/\rho_i\}$. The quantities in Eq.~\eqref{cf:eq:CsizarDivergence} take the name of \emph{statistical divergences} or \emph{contrast functions}~\cite{csiszarInformationTheoryStatistics2004,lesniewskiMonotoneRiemannianMetrics1999,petz2002covariance}.
	
	In~\cite{petz1986quasi, lesniewskiMonotoneRiemannianMetrics1999} the same axiomatic construction was presented for non-commuting states. In particular, the axioms chosen are:
	\begin{enumerate}
		\item \emph{positivity: } $H(\rho||\sigma)\geq0$, with equality \emph{iff} $\rho\equiv\sigma$;\label{cf:item:positivity}
		\item \emph{homogeneity:} $H(\lambda\rho||\lambda\sigma) = \lambda\,H(\rho||\sigma) $, for $\lambda>0$;\label{cf:item:homogeneity}
		\item \emph{joint convexity:} namely the condition that 
		$
		H(\lambda \rho_1 +(1-\lambda )\rho_2||\lambda \sigma_1 +(1-\lambda )\sigma_2)\leq \lambda H(\rho_1 ||\sigma_1 ) + (1-\lambda) H(\rho_2 ||\sigma_2 )
		$
		for $0\leq\lambda\leq1$; \label{cf:item:jointConvexity}
		\item \emph{monotonicity:} for any Completely Positive Trace Preserving (CPTP) map $\Phi$, it should hold that $H(\Phi(\rho)||\Phi(\sigma))\leq H(\rho||\sigma)$ (this condition is implied directly from the previous two, see~\hyperref[box:standardMonotones]{Box 1});\label{cf:item:monotonicity}
		\item \emph{differentiability:} the function $h_{\rho,\sigma}(x,y) := H(\rho + x A || \sigma + y B)$, for $A$ and $B$ Hermitian operators, is $C^\infty$.\label{cf:item:differentiability}
	\end{enumerate}
	Whereas the axioms are a straightforward generalisation of the ones chosen by Csiszár, for quantum states it is not clear how to generalise the original ansatz, especially due to the ratio in the classical expression in Eq.~\eqref{cf:eq:CsizarDivergence}. One standard approach, coming from constructions in ${\rm C}^*$-algebras, is to introduce the two superoperators:
	\begin{align}
		\LL_\rho[A] = \rho\,A\qquad\qquad\qquad
		\RR_\rho[A] = A\,\rho
	\end{align}
	that generalise the multiplication by a scalar to the non-commuting case. It is straightforward to verify that the inverse of the left/right multiplication operator is simply given by $(\LL_\rho)^{-1}/(\RR_\rho)^{-1} = \LL_{\rho^{-1}}/\RR_{\rho^{-1}}$. Moreover, it should also be noticed that 
	$[\LL_\rho,\RR_\sigma] =0$ for any two states $\rho$ and $\sigma$. Then, the ansatz proposed in~\cite{petz1986quasi} reads:
	\begin{align}\label{cf:eq:HGexpressionApp}
		H_g (\rho||\sigma) := \Tr{ g\boldsymbol{(}\LL_\sigma\RR_\rho^{-1}\boldsymbol{)}\sqrbra{\rho}}\,,
	\end{align}
	where now the role of the ratio in Eq.~\eqref{cf:eq:CsizarDivergence} is taken by $\LL_\sigma\RR_\rho^{-1}$, called the relative modular operator.
	In order for the axioms above to hold, one needs to impose that $g(x)$ is an operator convex function\footnote{An operator convex function $g(x)$ is defined by the property $g( t A + (1-t) B)\leq t\,g(A) + (1-t)\, g(B)$ for any two Hermitian $A$ and $B$, and $t\in[0,1]$.} $g(x):(0,\infty)\rightarrow\RR$ satisfying $g(1) = 0$~\cite{lesniewskiMonotoneRiemannianMetrics1999}. In the \hyperref[box:ContrastFunctions]{Box 1} at the end of the section we show how the properties of $g$ are connected to the axioms above. It should also be noticed that if one chooses $g(x) = -\log x$, then Eq.~\eqref{cf:eq:HGexpressionApp} gives the usual relative entropy. 
	
	In order to highlight the similarities and the differences with the classical contrast functions, it is insightful to look at the coordinate expression of Eq.~\eqref{cf:eq:HGexpressionApp}. First, define the eigensystems of $\rho$ and $\sigma$ as:
	\begin{align}
		\rho = \sum_i \;\rho_i \ketbra{\rho_i}{\rho_i}\,, \qquad\qquad\sigma = \sum_j \;\sigma_j \ketbra{\sigma_j}{\sigma_j}\,.\label{eq:rhoSigmaCoordinateExpression}
	\end{align}
	Since $\LL_\sigma\RR_\rho^{-1} [\ketbra{\sigma_j}{\rho_i}] = \frac{\sigma_j}{\rho_i}\ketbra{\sigma_j}{\rho_i}$, the relative modular operator is diagonal in the basis given by $\{\ketbra{\sigma_j}{\rho_i}\}$. Hence, Eq.~\eqref{cf:eq:HGexpressionApp} can be expanded as:
	\begin{align}
		H_g (\rho||\sigma) :&= \sum_{i,j}\; \Tr{g\boldsymbol{(}\LL_\sigma\RR_\rho^{-1}\boldsymbol{)}\sqrbra{\ketbra{\sigma_j}{\sigma_j}\rho\ketbra{\rho_i}{\rho_i} }} = \sum_{i,j}\;\rho_i\,g\norbra{\frac{\sigma_j}{\rho_i}} |\!\braket{\sigma_j}{\rho_i}\!|^2\,,\label{cf:eq:HGCoordinateExpression}
	\end{align}
	where we inserted the two resolutions of the identity $\sum_i \ketbra{\rho_i}{\rho_i}  = \sum_j \ketbra{\sigma_j}{\sigma_j}  =\id$ to rewrite the state $\rho$ in the diagonal basis of the relative modular operator.
	In particular, it is clear that if $\rho$ and $\sigma$ are diagonal in the same basis, then $|\!\braket{\sigma_j}{\rho_i}\!|^2 = \delta_{i,j}$, and Eq.~\eqref{cf:eq:HGexpressionApp} reduces to the classical expression in Eq.~\eqref{cf:eq:CsizarDivergence}. Moreover, it should be noticed that the requirement $g(1) = 0$ implies that $H_g (\rho||\rho) =0$ for any $\rho$.
	
	Interestingly, if one restricts their attention to normalised states, any linear term in $g(x)$ does not contribute to the contrast function, as it follows from the identity $H_{(x-1)} (\rho||\sigma)= \Tr{\rho-\sigma} = 0$ (which can be verified in coordinates or directly from Eq.~\eqref{cf:eq:HGexpressionApp} (see \hyperref[box:ContrastFunctions]{Box 1}). For this reason, we consider two functions to be equivalent if $g_1(x)-g_2(x)\propto(x-1)$. Then, ignoring the linear contributions, any operator convex function with $g(1)=0$ has the following integral expression:
	\begin{align}\label{cf:eq:gIntegralExpression}
		g(x) =\int_0^1\de\nu_g(s)\; \frac{(x-1)^2}{x+s} + \int_0^1\frac{\de\nu_g(s^{-1})}{s}\; \frac{(x-1)^2}{1+sx}\,,
	\end{align}
	where $\nu_g$ is a positive measure with finite mass (see App.~\ref{app:contrastFunctions}). Eq.~\eqref{cf:eq:gIntegralExpression} is particularly useful to give a unified expression for all contrast functions. Indeed, it can be proven that 
	contrast functions in Eq.~\eqref{cf:eq:HGexpressionApp} can be expressed as:
	\begin{align}
		H_g (\rho||\sigma) &= \Tr{(\rho-\sigma)\, \RR_\rho^{-1}h\boldsymbol{(}\LL_\sigma\RR_\rho^{-1}\boldsymbol{)}[(\rho-\sigma)]} \label{cf:eq:HGAlternative},
	\end{align}
	where we defined the function $h(x) := g(x)/(x-1)^{2}$ (here and in the following we refer to App.~\ref{app:contrastFunctions} for the explicit calculations). Putting together Eq.~\eqref{cf:eq:gIntegralExpression} and Eq.~\eqref{cf:eq:HGAlternative}, it is a matter of simple algebra to give the general integral form of quantum contrast functions:
	\begin{align}
		H_g( \rho || \sigma)  
		= \int_0^1 \de\nu_g(s)\;\Tr{(\rho-\sigma)(\LL_\sigma + s \RR_\rho)^{-1}[(\rho-\sigma)]}+\int_0^1\frac{\de\nu_g(s^{-1})}{s}\;\Tr{(\rho-\sigma)(\LL_\rho + s \RR_\sigma)^{-1}[(\rho-\sigma)]}\,.\label{cf:eq:integralRepresentationG}
	\end{align}
	This expression also shows that a contrast function is symmetric if and only if $\de\nu_g(s) =\de\nu_g(s^{-1})/s $, since exchanging $\rho\leftrightarrow\sigma$ exchanges the two integrals above (the other direction follows from the fact that the quantities in the integral of Eq.~\eqref{cf:eq:integralRepresentationG} are extreme points in the set of symmetric contrast functions, see \hyperref[box:contrastAdditionalProperties]{Box 10}). At the level of the defining function this corresponds to the requirement that $g(x)=x \,g(x^{-1})$. In order to put contrast functions in relation with their symmetric versions, it is useful to introduce the measure $\de N_g(s):=(\de\nu_g(s)+ \de\nu_g(s^{-1})/s)$. Then, it follows from Eq.~\eqref{cf:eq:integralRepresentationG} that:
	\begin{align}
		H_g( \rho || \sigma) +H_g( \sigma|| \rho ) = \int_0^1 \de N_g(s)\;\Tr{(\rho-\sigma)\norbra{(\LL_\sigma + s \RR_\rho)^{-1}+(\LL_\rho + s \RR_\sigma)^{-1}}[(\rho-\sigma)]}\,.\label{cf:eq:AppIntegralExpression}
	\end{align}
	
	We are now ready to study the local behaviour of contrast functions. Indeed, thanks to Eq.~\eqref{cf:eq:HGAlternative} and Eq.~\eqref{cf:eq:integralRepresentationG}, it is particularly simple to show that the following theorem holds:
	\begin{theorem}[Lesniewski, Ruskai~\cite{lesniewskiMonotoneRiemannianMetrics1999}]\label{cf:thm:Ruskai}
		For each $g$ satisfying the required properties to define a contrast function, one  can locally approximate $	H_g(\pi + \varepsilon A || \pi+ \varepsilon B) $ up to corrections of order $\bigo{\varepsilon^3}$ as:
		\begin{align}
			H_g(\pi + \varepsilon A || \pi+ \varepsilon B)  &= 	\varepsilon^2\,\int_0^1\de N_g(s)\;\Tr{(A-B)(\LL_\pi + s \RR_\pi)^{-1}[(A-B)]} + \bigo{\varepsilon^3}=	\label{cf:eq:thm2HGexpIntegral}\\
			&=\frac{\varepsilon^2}{2}\, \Tr{(A-B)\,\J_f^{-1}\big|_{\pi} [(A-B)]} + \bigo{\varepsilon^3}\,,\label{cf:eq:thm2HGexp}
		\end{align}
		where $A$ and $B$ are traceless, Hermitian perturbations, and the superoperator $\J_f\big|_\pi $ is defined as:
		\begin{align}
			\J_f\big|_\pi := \RR_\pi \,f\norbraB{\LL_\pi\RR_\pi^{-1}}\,.\label{cf:eq:65}
		\end{align}
		Moreover, the two functions $g$ and $f$ are connected by the equation:
		\begin{align}\label{cf:eq:correspondenceFG}
			f(x) = \frac{(x-1)^2}{g(x) + x \, g(x^{-1})}\,.
		\end{align}
	\end{theorem}
	Eq.~\eqref{cf:eq:thm2HGexpIntegral} can be directly derived from Eq.~\eqref{cf:eq:integralRepresentationG}. Indeed, thanks to the quadratic dependence of the contrast function on $(\rho-\sigma)$, up to correction of order $\bigo{\varepsilon^3}$, one can substitute $(\LL_{\pi + \varepsilon B} + s \RR_{\pi + \varepsilon A})^{-1} =  (\LL_{\pi + \varepsilon A} + s \RR_{\pi + \varepsilon B})^{-1} = (\LL_{\pi} + s \RR_{\pi })^{-1}$.   Then by using Eq.~\eqref{cf:eq:HGAlternative} one can see that:
	\begin{align}
		H_g (\pi +  \varepsilon A||\pi +  \varepsilon B) 
		&=\frac{\varepsilon^2}{2} \,\norbra{\Tr{(A-B)\, \norbra{\RR_{\pi +  \varepsilon A}^{-1} h\boldsymbol{(}\LL_{\pi +  \varepsilon B}\RR_{\pi +  \varepsilon A}^{-1}\boldsymbol{)}}[(A-B)]} + \Tr{A\leftrightarrow B} }+ \bigo{\varepsilon^3}=\label{cf:eq:A19} \\
		&= \frac{\varepsilon^2}{2} \,\Tr{(A-B)\, \RR_{\pi }^{-1}\, \frac{1}{f}\boldsymbol{(}\LL_{\pi }\RR_{\pi }^{-1}\boldsymbol{)}\,[(A-B)]} + \bigo{\varepsilon^3}\,,\label{cf:eq:A20}
	\end{align}
	where in Eq.~\eqref{cf:eq:A19} we denote by $\Tr{A\leftrightarrow B}$ the first trace of the line but with $A$ and $B$ exchanged, and we used the fact that exchanging the arguments of $H_g(\rho||\sigma)$ does not affect the result at this order of approximation. Then, in the last line one can read the explicit expression of $\J_f^{-1}\big|_{\pi}$. In fact, due to the multiplicative behaviour of $\LL_\pi/\RR_\pi$ the inverse of $\J_f\big|_{\pi}$ is given by:
	\begin{align}
		\J_f^{-1}\big|_{\pi} = \RR_{\pi }^{-1}\, \frac{1}{f}\boldsymbol{(}\LL_{\pi }\RR_{\pi }^{-1}\boldsymbol{)}\,.\label{cf:eq:65Inverse}
	\end{align}
	
	It should be noticed that each $f$ is in one to one relation with a unique symmetric contrast function, i.e., that for which $g(x) = x g(x^{-1})$ holds, from which it follows that $f(x) = (x-1)^2/(2\,g^{\rm symm}(x))$. Moreover, $f$ satisfies the following three conditions: 
	\begin{enumerate}
		\item since $g$ is matrix convex, $f$ is matrix monotone~\cite{lesniewskiMonotoneRiemannianMetrics1999,hiaiIntroductionMatrixAnalysis2014}. Indeed, combining Eq.~\eqref{cf:eq:gIntegralExpression} with Eq.~\eqref{cf:eq:correspondenceFG}, we get that $1/f$ can be expanded as:
		\begin{align}
			\frac{1}{f(x)} = \int_0^1\de N_g(s)\;\norbra{\frac{1}{x+s}+\frac{1}{1+s \,x}}\,.\label{eq:21nuova}
		\end{align}
		Since the functions in the integrand are matrix monotone decreasing, the same holds for $1/f$. This directly implies that $f$ is matrix monotone;
		\item $f(x)$ satisfies the symmetry $f(x) = x f(x^{-1})$, as it can be verified by a straightforward calculation;
		\item without loss of generality, we require the normalisation condition $f(1) = 1$, corresponding to fixing the value of $g''(1) = 1$, which can be imposed by the simple rescaling $\tilde{g}(x) = g(x)/g''(1)$.\label{item:3standardMonotoneFirst}
	\end{enumerate}
	The functions satisfying these three conditions are called standard monotone. It can be shown (see \hyperref[box:standardMonotones]{Box 2} at the end of the section) that they are all pointwise bounded as:
	\begin{align}
		\frac{2x}{x+1} \,\leq f(x)\leq \,\frac{x+1}{2}\,.\label{cf:eq:rangeStandardMonotones}
	\end{align}
	Moreover, thanks to normalisation condition $f(1) = 1$, in the case of commuting observables (i.e., $[\pi,A]=[\pi,B]=[A,B]=0$), it follows that:
	\begin{align}
		H_g (\pi +  \varepsilon A||\pi +  \varepsilon B) 
		&=\frac{\varepsilon^2}{2} \,\Tr{(A-B)^2 \,\pi^{-1}} + \bigo{\varepsilon^3}= \frac{\varepsilon^2}{2}\sum_i \;\frac{(A_i-B_i)^2}{\pi_i} + \bigo{\varepsilon^3}\,,\label{eq:classicalFisher}
	\end{align}
	where we used the fact that on commuting observables $\LL_{\pi }\RR_{\pi }^{-1}[A] = \idO[A] = A$, and we expanded $A,\,B$ and $\pi$ on a common eigenbasis. This shows that from the point of view of classical probability, up to a normalisation, all contrast functions locally behave in the same way. This is a well known result in statistics, and it is one of the ways of defining the Fisher information metric. Indeed, consider the case in which $[\pi, \delta\rho] = 0$, where $\delta\rho$ is a vector in the tangent space of $\pi$ (i.e., a traceless, Hermitian perturbation). Then, from Eq.~\eqref{eq:classicalFisher} we can see that:
	\begin{align}
		H_g (\pi ||\pi +  \varepsilon \,\delta\rho) &=\frac{\varepsilon^2}{2}\sum_i \;\frac{(\delta\rho_i)^2}{\pi_i} + \bigo{\varepsilon^3} = \frac{\varepsilon^2}{2}\sum_{i,j} \;\delta\rho_i \, \eta_{i,j} \, \delta\rho_j + \bigo{\varepsilon^3}\,,\label{eq:classicalFisher2}
	\end{align}
	where we introduced the matrix $\eta_{i,j} = \delta_{i,j}/\pi_i$. This is  symmetric, positive, and it smoothly depends on the base-point (assuming that $\pi$ is a full rank state, i.e., we are on the interior of the space of states). These are the defining properties of a metric, so that $\eta$ endows the tangent space of diagonal density matrices with a metric structure. In this context, $\eta$ is exactly what takes the name of classical Fisher information metric.
	
	It should be noticed, though, that the uniqueness of $\eta$ is lost when moving to non-commuting observables. This might leave open the doubt whether the family of operators $\J_f^{-1}\big|_\pi $ is the right generalisation of the Fisher information metric to quantum states, or if one should rather introduce some further constraints to single out a unique quantity. In order to resolve this question one can use another independent characterisation of the classical Fisher information: the Chentsov theorem. This says that the unique metric on the space of probability distributions contracting under all stochastic maps is exactly the Fisher information~\cite{cencovStatisticalDecisionRules2000,campbell1986extended}. 
	The generalisation of this result to quantum systems was provided by Petz in the study of monotone metrics. These are  scalar products  $K_\pi(A,B)$  on the tangent space of the manifold of states that satisfy the two additional properties:
	\begin{enumerate}
		\item \emph{smoothness:} $K_\pi(A,B)$ depends smoothly on the base-point $\pi$;
		\item \emph{monotonicity:} for every CPTP $\Phi$ the metric is contractive: $$K_{\Phi(\pi)}(\Phi(A),\Phi(A))\leq K_\pi(A,A) \,.$$
	\end{enumerate}
	As we mentioned, Chentsov theorem identifies the Fisher information as the unique monotone metric on the space of classical probability distributions. For quantum states, on the other hand, we have:
	\begin{theorem}[Petz~\cite{petzMonotoneMetricsMatrix1996}]\label{cf:thm:Petz}
		The monotone metrics on quantum states are all given in the form:
		\begin{align}
			K_{f,\pi}(A,B)	 := \Tr{A\, \,\J_f^{-1}\big|_\pi [B]},\label{cf:eq:monotoneMetrics}
		\end{align}
		where $f:\RR^+\rightarrow\RR^+$ is an operator monotone function. Moreover, requiring that $K_{f,\pi}(A,B)$ is real and that it reduces to the classical Fisher information for commuting variables constrains $f$ to be a standard monotone function.
	\end{theorem}
	This last result corroborates the interpretation of $\J_f^{-1}\big|_\pi$ as the natural extension of the classical Fisher information metric to quantum mechanical systems. Moreover, it also shows that the definition of contrast functions in Eq.~\eqref{cf:eq:HGexpressionApp} is well justified, as their local behaviour correctly reduces to the quantum Fisher information.
	
	This concludes the introduction to the quantum Fisher information. A more in-depth treatment will be provided in Sec.~\ref{propertiesFisher}, where we present the main properties of the superopertator $\J_f\big|_\pi$ and of the standard monotone functions,  and in Sec.~\ref{gardenFisher}, where we provide a list of the most frequently used quantum Fisher information.
	
	Before moving on, though, it is useful to clarify the nomenclature. We will use the expression quantum Fisher information scalar product to refer to $K_{f,\pi}(A,B)$, whereas the quantity $\mathcal{F}_{f,\pi}(\delta\rho) := K_{f,\pi}(\delta\rho,\delta\rho)$ is traditionally referred to as quantum Fisher information. Finally, we call the two maps $\J_f\big|_\pi$ and $\J_f^{-1}\big|_\pi$ the quantum Fisher operators (and its inverse). Since our results apply both to the quantum and to the classical case, in the following we sometimes drop the reference to which scenario we are considering, as it should also be clear from the context.

	\vspace{0.2cm}
	
	\begin{additional-info}{Box 1. Conditions on the defining function $g$~\hyperref[box:standardMonotones]{$\rightarrow$}}\label{box:ContrastFunctions}
		In this section we explore which properties one has to impose on the function $g$  in Eq.~\eqref{cf:eq:HGexpressionApp} in order for conditions \ref{cf:item:positivity}-\ref{cf:item:differentiability} to be satisfied. Condition~\ref{cf:item:positivity} can be decomposed in two parts: first, that the contrast function is positive, and moreover, that it is zero if and only if $\rho\equiv\sigma$ (\emph{faithfulness}).  As one can verify from the coordinate expression in Eq.~\eqref{cf:eq:HGCoordinateExpression}, this second part of condition~\ref{cf:item:positivity} corresponds to imposing that $g(1)=0$ is the only zero of the function. Moreover, as we mentioned in the main text, one can verify that $g(x) = (x-1)$ is identically zero for normalised states, as this corresponds to:
		\begin{align}
			H_{(x-1)} (\rho||\sigma)= \Tr{ \boldsymbol{(}\LL_\sigma\RR_\rho^{-1} - \idO\boldsymbol{)}\sqrbra{\rho}}= \Tr{\sigma-\rho} = 0\,.
		\end{align}
		Then, in order to ensure the positivity of $H_g(\rho||\sigma)$ it is sufficient to require that $g(x) + a(x-1)\geq 0$, for some arbitrary  constant $a$. On the other hand, thanks to the way in which the ansatz is formulated, condition~\ref{cf:item:homogeneity} is automatically satisfied, while one needs to require that $g(x)$ is matrix convex at $x = 1$ for the joint convexity to hold (condition~\ref{cf:item:jointConvexity}). This directly implies  the monotonicity of the contrast functions, condition~\ref{cf:item:monotonicity}, fact that can be proved as follows: first, it should be noticed that $H_g(\rho||\sigma)$ are unitary invariant. Indeed, one has:
		\begin{align}
			H_g (U\,\rho\,U^\dagger||U\,\sigma\,U^\dagger) :&= \Tr{U\,\rho^{1/2}\,U^\dagger\, g\boldsymbol{(}\LL_{U\,\sigma\,U^\dagger}\RR_{U\,\rho^{-1}\,U^\dagger}\boldsymbol{)}\sqrbra{U\,\rho^{1/2}\,U^\dagger}} =\label{cf:eq:127n}\\
			&=\Tr{U\,\rho^{1/2}\,U^\dagger\,U\, g\boldsymbol{(}\LL_{\sigma}\RR^{-1}_{\rho^{-1}}\boldsymbol{)}\sqrbra{\rho^{1/2}}\,U^\dagger} = H_g (\rho||\sigma) \,,\label{cf:eq:128n}
		\end{align}
		where the step from Eq.~\eqref{cf:eq:127n} to the Eq.~\eqref{cf:eq:128n} can be verified either in coordinates (see Eq.~\eqref{cf:eq:HGCoordinateExpression}), or by expanding $g(x)$ in Laurent series, while in Eq.~\eqref{cf:eq:128n} we exploited the unitarity of $U$ and the cyclicity of the trace. Secondly, also notice that for generic contrast functions it holds that:
		\begin{align}
			H_g (\rho\otimes\tau||\sigma\otimes\tau) &= \Tr{g\boldsymbol{(}\LL_{\sigma\otimes\tau}\RR_{\rho\otimes\tau}^{-1}\boldsymbol{)}\sqrbra{\rho\otimes\tau}} = \Tr{ g\boldsymbol{(}\LL_{\sigma}\RR_{\rho}^{-1}\otimes\idO\boldsymbol{)}\sqrbra{\rho\otimes\tau}}=\label{cf:eq:130B}\\
			&= \Tr{    g\boldsymbol{(}\LL_{\sigma}\RR_{\rho}^{-1}\boldsymbol{)}\sqrbra{\rho}}\Tr{\tau}= H_g (\rho||\sigma)\label{cf:eq:131n}\,,
		\end{align}
		where in Eq.~\eqref{cf:eq:130B} we used the fact that $\LL_\tau\RR_\tau^{-1}$ coincides with the identity operator on the commutant of $\tau$ (again one can see this either from the Laurent series or directly from the coordinate expression in Eq.~\eqref{cf:eq:HGCoordinateExpression}). These two facts allow to deduce the monotonicity of the contrast functions in Eq.~\eqref{cf:eq:HGexpressionApp} from their joint convexity. Indeed, given a CPTP map $\Phi$, one can express it in terms of its Stinespring's dilation:
		\begin{align}\label{cf:eq:stinespring}
			\Phi(\rho) = \TrR{E}{U\, (\rho \otimes \ketbra{\psi}{\psi})\,U^\dagger}\,,
		\end{align}
		where $U$ is a unitary operator and $\ketbra{\psi}{\psi}$ is an environmental pure state of dimension $d_E$. Take a unitary basis $\{V_i\}$ for the space of bounded operators of dimension $d_E\times d_E$. It is a well-known result that $\sum_i (V_i(\rho_E)V_i^\dagger)/d_E^2 = \id_E/d_E$ for any $\rho_E$~\cite{nielsen2002quantum}. We denote this superoperator by $\Delta_1(\rho)$. This identity, together with Eq.~\eqref{cf:eq:stinespring} allows to rewrite the action of the channel as:
		\begin{align}\label{cf:eq:130n}
			\Phi(\rho)\otimes \frac{\id_{d_E}}{d_E} = \frac{1}{d_E^2} \,\sum_{i=1}^{d_E^2} \; (\id\otimes V_i)\,U\, (\rho \otimes \ketbra{\psi}{\psi})\,U^\dagger(\id\otimes V_i^\dagger)\,.
		\end{align}
		Putting together this expression with Eq.~\eqref{cf:eq:131n} we can finally prove monotonicity. Indeed, one has:
		\begin{align}
			H_g(\Phi(\rho)||\Phi(\sigma)) &= H_g\norbra{\Phi(\rho)\otimes \frac{\id_{d_E}}{d_E} \bigg|\bigg|\Phi(\sigma)\otimes \frac{\id_{d_E}}{d_E} } = \\
			&\!= H_g\norbra{(\idO\otimes\Delta_1)\norbra{U\, (\rho \otimes \ketbra{\psi}{\psi})\,U^\dagger}||(\idO\otimes\Delta_1)\norbra{U\, (\sigma \otimes \ketbra{\psi}{\psi})\,U^\dagger}} \leq\\
			&\leq \frac{1}{d_E^2} \,\sum_{i=1}^{d_E^2} \;H_g(\rho \otimes \ketbra{\psi}{\psi}||\sigma \otimes \ketbra{\psi}{\psi}) =H_g(\rho ||\sigma ) \,,\label{cf:eq:36}
		\end{align}
		where in the first step we used Eq.~\eqref{cf:eq:131n}, then we applied the decomposition in Eq.~\eqref{cf:eq:130n} (notice that we denote by $\idO$ the identity superoperator), and the inequality comes from the joint convexity of $H_g(\rho||\sigma)$, together with the unitary invariance of contrast functions. Finally, we applied Eq.~\eqref{cf:eq:131n} once more. This proves condition~\ref{cf:item:monotonicity}.
		
		Finally, condition~\ref{cf:item:differentiability} follows from requiring $g(x)$ to be smooth. 
		
		Hence, in order for $H_g(\rho||\sigma)$ to be a proper contrast function, $g(x)$ has to be a matrix convex function in $C^\infty(\RR^+)$, such that $g(1)= 0$ is the unique zero and $g(x) + a(x-1)\geq 0$, for some arbitrary  constant $a$.
	\end{additional-info}
	
	\begin{additional-info}{Box 2. Bounding the set of standard monotone functions~\hyperref[box:ContrastFunctions]{$\leftarrow$},\hyperref[box:statisticalQuantifiers]{$\rightarrow$}}\label{box:standardMonotones}
		The characterisation in Eq.~\eqref{cf:eq:rangeStandardMonotones} of the set of standard monotone functions follows from the  Lemma:
		\begin{lemma}[Theorem 4.43 from~\cite{hiaiIntroductionMatrixAnalysis2014}]\label{cf:lemma:TF}
			Consider a function $f:(0,\infty)\rightarrow(0,\infty)$. The following conditions are equivalent:
			\begin{enumerate}
				\item $f(x)$ is matrix monotone;
				\item $[Tf](x) := x/f(x)$ is  matrix monotone;\label{item:tTransform}
				\item $f(x)$ is matrix concave.
			\end{enumerate}
		\end{lemma}
		Then, we can show that standard monotones are all contained in the interval:
		\begin{align}
			\frac{2x}{x+1} \,\leq f(x)\,\leq \frac{x+1}{2}\,.\label{cf:eq:rangeStandardMonotones2}
		\end{align}
		This can be proved as follows: thanks to the condition $f(x) = x f(x^{-1})$ one only needs to characterise the properties of $f(x)$ in the interval $[0,1]$. Moreover, the same condition also implies that $f'(1) = \frac{1}{2}$. In fact, this can be easily verified from the equation
		\begin{align}
			f'(x) =  f(x^{-1}) - \frac{1}{x} f'(x^{-1})\,,
		\end{align} 
		by setting $x$ to $1$ and using the normalisation $f(1)=1$. Then, from concavity it follows that $f(x) \leq f(1) + f'(1)(x-1) = (x+1)/2$. The upper bound satisfies all the necessary constraints, so it can be identified as the largest standard monotone function, which we indicate by $f_B$. Finally, notice that the transformation $f\rightarrow Tf$ (defined in~\ref{item:tTransform}) inverts the inequality and maps standard monotones into standard monotones. For this reason, any standard monotone function can be pointwise bounded as in Eq.~\eqref{cf:eq:rangeStandardMonotones2}.
	\end{additional-info}
	
	\begin{additional-info}{Box 3. Other possible quantifiers of statistical difference~\hyperref[box:standardMonotones]{$\leftarrow$},\hyperref[FisherDynamics]{$\rightarrow$}}\label{box:statisticalQuantifiers}
		It is important to point out that the contrast functions in Eq.~\eqref{cf:eq:HGexpressionApp} are not the only quantifiers of statistical difference that satisfy conditions~\ref{cf:item:positivity}-\ref{cf:item:differentiability} and locally expand to the Fisher information metric. 
		
		For example, it is straightforward to verify that the family of $\chi^2_f$-divergences introduced in~\cite{temmeH2divergenceMixingTimes2010}:
		\begin{align}
			\chi^2_f(\rho||\sigma) = \Tr{(\rho-\sigma)\, \,\J_f^{-1}\big|_\rho [(\rho-\sigma)]}\,\label{eq:chi2Divergences}
		\end{align}
		are positive, homogeneous, jointly convex, monotonous under CPTP maps and differentiable. Moreover, they trivially reduce to the corresponding Fisher information for close-by states. This shows that the one in Eq.~\eqref{cf:eq:HGexpressionApp} is just an ansatz, and that other possible functions may be chosen to satisfy the same properties.
		
		Indeed, another possibility is given by the geodesic distance ${\rm D}_f(\rho,\sigma)$ associated to each metric $\J_f^{-1}$. In this case we have the additional two properties that ${\rm D}_f(\rho,\sigma)$ is symmetric (as it is a distance) and that it satisfies the triangle inequality:
		\begin{align}
			{\rm D}_f(\rho,\sigma) \leq {\rm D}_f(\rho,\tau) + {\rm D}_f(\tau,\sigma)\,.
		\end{align}
		For classical probability distributions the geodesic distance has a closed expression, namely ${\rm D}_f(\rho,\sigma) = 2\,\arccos \Tr{\sqrt{\rho\sigma}}$, whereas for the quantum case the solution is known only in a couple of cases (see Sec.~\ref{Bures} and Sec.~\ref{wignerYanase}).
		
		Despite the formal difference between the expression for contrast functions and $\chi^2_f$-divergences, one could wonder whether there are some special functions $g$ and $f$ for which the two coincide. For this reason, we provide here a criterion to verify whether an arbitrary quantifier satisfying conditions~\ref{cf:item:positivity}-\ref{cf:item:differentiability} can actually be expressed as a contrast function. 
		
		First, it should be noticed that given a contrast function $H_g(\rho||\sigma)$, the defining function is completely characterised as $g(x) := H_g\norbra{\frac{\id}{d}\big|\big|\,x\,\frac{\id}{d}}$, where we implicitly extended the definition of $H_g$ also to unnormalised states. 
		
		Moreover, starting from two arbitrary states $\rho$ and $\sigma$ one can define the two probability distributions:
		\begin{align}
			(\rho:\sigma)^1_{i,j} = \rho_i |\!\braket{\sigma_j}{\rho_i}\!|^2\qquad\qquad (\rho:\sigma)^2_{i,j} = \sigma_j |\!\braket{\sigma_j}{\rho_i}\!|^2\,,
		\end{align}
		where we used the same notation as in Eq.~\eqref{eq:rhoSigmaCoordinateExpression}. Then, it follows from the coordinate expression in Eq.~\eqref{cf:eq:HGCoordinateExpression} that:
		\begin{align}
			H_g(\rho||\sigma) = H_g((\rho:\sigma)^1||(\rho:\sigma)^2)\,,
		\end{align}
		where on the left hand side one has the contrast function between two quantum states, and on the right hand side between two classical probability distributions. Hence, if two contrast functions $H_{g_1}(\rho||\sigma)$ and $H_{g_2}(\rho||\sigma)$ coincide for arbitrary classical distributions, then they will coincide for quantum states as well.
		
		Consider now the case of $\chi^2_f$-divergences. It is easy to verify that $\chi^2_f\norbra{\frac{\id}{d}\big|\big|x\,\frac{\id}{d}}= (x-1)^2$, so that if $\chi^2_f$ were to be expressable in terms of contrast functions, this should happen for $g(x) = (x-1)^2$. The corresponding contrast function is studied in Sec.~\ref{harmonic}, where its explicit expression is provided. Interestingly, it is proved in there that for $f_H(x) := \frac{2x}{x+1}$, one can rewrite:
		\begin{align}
			H_{\frac{(x-1)^2}{2}}(\rho||\sigma) = \frac{1}{2}\:\chi^2_{f_H}(\rho||\sigma)\,,
		\end{align}
		proving the representability of the $\chi^2_{f_H}$-divergence. On the other hand, for all other  $\chi^2_{f}$-divergences this also shows that they cannot be expressed in terms of contrast functions.
		
		The method just presented is taken from~\cite{hiai2011quantum}, where it was used to prove that the fidelity, the Chernoff distance and the Hoeffding distances are all not representable as contrast functions.
	\end{additional-info}
	
	\section{Dynamical properties of Fisher information}\label{FisherDynamics}
	
	In the {standard} treatment of  Fisher information, most of the focus goes into its significance as a distinguishability measure (see Thm.~\ref{cf:thm:Ruskai}), and in the many different results connecting it to estimation theory and information theory, as for example the Cramér-Rao bound~\cite{Cramer46,Rao92} or the Chernoff bound~\cite{Chernoff52} (cf. Eq.~\eqref{cf:eq:cramerRao130} and Eq.~\eqref{cf:eq:chernoffBound}). Nonetheless, Thm~\ref{cf:thm:Petz} hints at a more dynamical character of the Fisher information metrics: {these are the only metrics} that monotonically decrease under the action of any CPTP {- i.e., physical -} map. This shows that Fisher information metrics can be understood in two very distinct, almost independent, ways: one statistical (Thm.~\ref{cf:thm:Ruskai}) and one dynamical (Thm.~\ref{cf:thm:Petz}). Indeed, it is not even obvious \emph{a priori} that the two should be connected, partially explaining why the interest in the statistical properties of the Fisher information lead to so less insights about the dynamical features.
	
	This section is devoted to justify and legitimate the dynamical point of view, in order to demonstrate the intimately dynamical nature of the Fisher metrics. 
	Indeed, we find that not only the Fisher information can be defined in terms of its contractivity under physical evolutions, but also many properties of the dynamical maps can be given solely in terms of their  behaviour with respect to the Fisher information. A central result in this direction is given by the following Theorem and its Corollary~\cite{scandi2023physicality}:
	\begin{theorem}\label{cf:theo:CPiffContracts}
		Consider a trace-preserving, Hermitian preserving, linear map $\Phi: \mathcal{M}_d(\mathbb{C})\rightarrow\mathcal{M}_d(\mathbb{C})$, where $\mathcal{M}_d(\mathbb{C})$ is the space of $d \times d$ complex matrices. Define $  \mathcal{S}_d\subset\mathcal{M}_d(\mathbb{C})$ to be the set of positive semidefinite, trace-one matrices. If $\Phi$ satisfies the following three properties: 
		\begin{enumerate}\label{cf:thm:contractIf_P}
			\item $\Phi$ is invertible;\label{cf:it:1}
			\item $\Phi$ maps at least one point from the interior of $ \mathcal{S}_d$ into $ \mathcal{S}_d$;\label{cf:it:2}
			\item for any state $\rho\in \mathcal{S}_d$ such that  $\Phi(\rho)\in \mathcal{S}_d$ and tangent vector $\delta \rho$, one has:\label{cf:it:3}
			\begin{align}
				K_{f,\rho}(\delta\rho,\delta\rho) \geq
				K_{f,\Phi(\rho)}(\Phi(\delta\rho),\Phi(\delta\rho)) \,;\label{eq:contractionContrast}
			\end{align}
		\end{enumerate}
		then the image of $\Phi$ is completely contained in $\mathcal{S}_d$ ($\Phi$ is a Positive  (P) map).
	\end{theorem}
	
	From this result it trivially follows that:
	
	\begin{corollary}\label{cf:cor:CPiffContracts}
		Consider  the extended channel $\Phi_{\rm anc}:=\Phi \otimes\idO_d $. Under the same assumptions of Thm~\ref{cf:thm:contractIf_P}, where the contractivity requirement for $\Phi_{\rm anc}$ takes the form:
		\begin{itemize}
			\item[3'.] for any state $\rho\in \mathcal{S}_d\otimes\mathcal{S}_d$ such that  $\Phi_{\rm anc}(\rho)\in \mathcal{S}_d\otimes\mathcal{S}_d$ and tangent vector $\delta \rho$, it holds that: 
			\begin{align}
				K_{f,\rho}(\delta\rho,\delta\rho) \geq
				K_{f,\Phi_{\rm anc}(\rho)}(\Phi_{\rm anc}(\delta\rho),\Phi_{\rm anc}(\delta\rho)) \;;
			\end{align}
		\end{itemize}
		then, $\Phi$ is Completely-Positive (CP).
	\end{corollary}
	
	The theorem and its corollary show that the property of being completely positive can be defined exclusively in statistical terms. CP maps are all the physically realisable dynamics, so the results above gives a characterisation of which physical evolutions are possible without any reference to an actual theory of the world. Indeed, this characterisation could be regarded as the dual of the Chentsov--Petz theorem: not only the Fisher information is the unique family of metrics that contracts under arbitrary CP maps, but also CP maps are the only set of maps that contract the Fisher information.
	
	\begin{proof}
		We prove Theorem~\ref{cf:thm:contractIf_P}
		by contradiction: suppose there exists a map $\Phi$ that is not P, but that satisfies Eq.~\eqref{eq:contractionContrast}. Thanks to condition~\ref{cf:it:2} there exists at least one point $\pi$ in the interior of $ \mathcal{S}_d$ such that $\Phi(\pi)\in  \mathcal{S}_d$ is also a state. Moreover, from the assumption that $\Phi$ is not P there is also at least one state $\sigma$ such that $\Phi(\sigma)\notin \mathcal{S}_d$. Without loss of generality one can choose $\sigma$ to be in the interior of $ \mathcal{S}_d$: if this is not the case, one can take a ball around $\Phi(\sigma)$ whose image still lays outside of the state space, and by inspecting the intersection between its preimage and $ \mathcal{S}_d$ one can find a point satisfying the assumption. Hence, the line $\rho_\lambda := (1-\lambda)\,\pi + \lambda \,\sigma$ also lays in the interior of the state space for $\lambda\in[0,1]$. From this, it follows that the following superior is finite:
		\begin{align}
			\sup_{\lambda, \Tr{\delta\rho^2} = 1}\,K_{f,\rho_\lambda}(\delta\rho,\delta\rho)  = \sup_{\lambda, \Tr{\delta\rho^2} = 1} \, \Tr{\delta\rho\,\J_f^{-1}\big|_{\rho_\lambda} [\delta\rho]} < \infty\,,
		\end{align}
		where we used the fact that the Fisher information is a bounded operator when one is restricted to a closed set completely inside the interior of the state space (see Eq.~\eqref{cf:eq:182x}). 
		Since by varying $\lambda$, $\Phi(\rho_{\lambda})$ interpolates linearly between the positive definite matrix $\Phi(\pi)$ and one with at least one negative eigenvalue, namely $\Phi(\sigma)$, there exists a $\lambda^*$ such that $\Phi(\rho_{\lambda^*})$ is a state with at least one zero-eigenvalue. We are now ready to prove the claim. Set the state $\rho_\eta$ so that the smallest eigenvalue of $\Phi(\rho_\eta)$ is of order $\eta$, where $\eta\ll 1$. Moreover, consider a perturbation $\delta \rho_\eta$ such that $[\Phi(\rho_\eta), \Phi(\delta\rho_\eta)] = 0$, and having  a positive finite contribution along the eigenvectors corresponding to $\eta$-eigenvalues (the positivity condition ensures that in the limit $\eta\rightarrow0$ the perturbed state is still in the interior of $ \mathcal{S}_d$ for any finite $\eta$). Since $\Phi$ is full rank one can always find such $\delta\rho_\eta$. Then, choosing a common eigenbasis for $\rho_\eta$ and $\delta\rho_\eta$,  the evolved Fisher information reads: 
		\begin{align}
			K_{f,\Phi(\rho_\eta)}(\Phi(\delta\rho_\eta),\Phi(\delta\rho_\eta)) = \sum_{i}\; \frac{\Phi(\delta\rho_\eta)_i^2}{\Phi(\rho_\eta)_i}\,,\label{eq:46}
		\end{align}
		where we have used the expression of the Fisher information for commuting observables (see Eq.~\eqref{eq:classicalFisher2}). The quantity in \eqref{eq:46} scales as $\eta^{-1}$ as $\eta\rightarrow0$. We can thus always find a $\eta$ small enough such that:
		\begin{align}
			K_{f,\Phi(\rho_\eta)}(\Phi(\delta\rho_\eta),\Phi(\delta\rho_\eta))
			> \sup_{\lambda, \Tr{\delta\rho^2} = 1}\,K_{f,\rho_\lambda}(\delta\rho,\delta\rho)
			\geq \;
			K_{f,\rho_\eta}(\delta\rho_\eta,\delta\rho_\eta) \;,
		\end{align}
		contradicting the assumption that $\Phi$ contracts the Fisher metric for any two points in the interior of the space of states (condition~\ref{cf:it:3}). This proves the claim.
	\end{proof}
	
	It should be noticed that since the counterexample we find only comprises commutative observables, there is no need to specify which Fisher information we are actually using in Thm.~\ref{cf:theo:CPiffContracts}: indeed, if this condition is valid for one $f$, then it is automatically valid for all others. \ms{Moreover, as it was proven in~\cite{scandi2023physicality}, one can drop the condition~\ref{cf:it:1} at the cost of a slightly more technical proof (we refer to~\cite{scandi2023physicality} for the details on this matter, together with some slight generalisation of the result to include trace non-preserving maps).}
	
	The results just proven demonstrate a deep connection between the Fisher information and the dynamics of quantum systems. Indeed, this also reflects in the fact that a number of dynamical properties can be naturally formulated in terms of the Fisher metrics. 
	In the following, we are going to focus on three dynamical aspects in particular:
	\begin{itemize}
		\item[{\hyperref[markovianity]{A}}.] Markovianity of {an evolution} can be related to the monotonous contraction of Fisher information, using the same principle that allows to express the CP-ness of a map in terms of its contractivity property with respect to the Fisher information;
		\item[{\hyperref[retrodiction]{B}}.] The contractivity of Fisher information strictly relates also to the ability of recovering the original states of the dynamics, via a generalisation of Bayesian retrodiction;
		\item[{\hyperref[detailedBalance]{C}}.] Detailed balanced dynamics can be characterized in terms of the adjointness of the dynamical generators of the evolution, with respect to the scalar product induced by the Fisher information.
	\end{itemize}  
	
	These topics are the subject of the next subsections and corroborate the interpretation of the Fisher information as an intimately dynamical quantity.
	
	\begin{additional-info}{Box 4. Corollaries for finite divergences~\hyperref[markovianity]{$\rightarrow$}}\label{box:cpIffContracts}
		Thm.~\ref{cf:theo:CPiffContracts} gives a characterisation of P-maps in terms of their contractivity with respect to the Fisher information. Still, it should be noticed that since the proof proceeds by providing a local counterexample, the same result could in principle be formulated in terms of finite distinguishability measures. In particular, it is easy to prove the following:
		\begin{corollary}
			Consider a channel satisfying the premises of Thm.~\ref{cf:theo:CPiffContracts}, and condition~\ref{cf:it:1} and~\ref{cf:it:2} therein. If one also requires that:
			\begin{itemize}
				\item[ 3*.] For any two states $\rho\in\mathcal{S}_d$ and $\sigma\in\mathcal{S}_d$ such that $\Phi(\rho)\in\mathcal{S}_d$ and $\Phi(\sigma)\in\mathcal{S}_d$ it holds that:
				\begin{align}
					H_g (\rho||\sigma)\geq H_g (\Phi(\rho)||\Phi(\sigma))\,;
				\end{align}
			\end{itemize}
			then, the image of $\Phi$ is completely contained in $\mathcal{S}_d$, i.e., $\Phi$ is a P-map.
		\end{corollary}
		\begin{proof}
			The proof of the corollary directly follows from the one of Thm.~\ref{cf:theo:CPiffContracts}. Indeed, by contradiction, assume there exists a state $\sigma$ that  gets mapped outside of $\mathcal{S}_d$, and consider a state $\pi\in \mathcal{S}_d$ such that $\Phi(\pi)\in\mathcal{S}_d$. Analogously to the proof above, define $\rho_\lambda := (1-\lambda)\,\pi + \lambda \,\sigma$. For $\varepsilon$ small enough we can also apply Thm.~\ref{cf:thm:Ruskai} to obtain:
			\begin{align}
				\sup_{\lambda, \Tr{\delta\rho^2} = 1}\,H_g(\rho_\lambda ||\rho_\lambda+\varepsilon\,\delta\rho) = \sup_{\lambda, \Tr{\delta\rho^2} = 1}\,\frac{\varepsilon^2}{2}\,K_{f,\rho_\lambda}(\delta\rho,\delta\rho)  +\bigo{\varepsilon^3}< \infty\,.
			\end{align}
			Then, following the steps presented for Thm.~\ref{cf:theo:CPiffContracts}, we define a $\rho_\eta$ such that $\Phi(\rho_\eta)$ is of order $\eta$, where $\eta\ll 1$, and a $\delta\rho_\eta$ such that $[\Phi(\rho_\eta), \Phi(\delta\rho_\eta)] = 0$, and having  a positive finite contribution along the eigenvectors corresponding to $\eta$-eigenvalues. The positivity condition ensures that in the limit $\eta\rightarrow0$ the perturbed state is still in the interior of $ \mathcal{S}_d$ for any finite $\eta$, and $\varepsilon$ small enough. Then, $H_g(\Phi(\rho_\eta) ||\Phi(\rho_\eta+\varepsilon\,\delta\rho_\eta))$ scales as $\varepsilon^2/\eta$. Hence, one can find $\eta$ small enough such that:
			\begin{align}
				H_g(\Phi(\rho_\eta) ||\Phi(\rho_\eta+\varepsilon\,\delta\rho_\eta)) > \sup_{\lambda, \Tr{\delta\rho^2} = 1}\,H_g(\rho_\lambda ||\rho_\lambda+\varepsilon\,\delta\rho) \geq \;H_g(\rho_\eta ||\rho_\eta+\varepsilon\,\delta\rho_\eta) \,.
			\end{align}
			This gives the desired contradiction, concluding the proof.
		\end{proof}
		The same kind of argument could be made for $\chi^2_f$-divergences and geodesic distances. Moreover, even in this case one can restrict their attention to a single contrast function (or any other distinguishability measure that locally expands to the Fisher information), as the counterexample we use is commutative. 
	\end{additional-info}
	
	\subsection{Markovianity as monotonous contraction of Fisher information}\label{markovianity}
	
	In order to investigate how the Fisher information relates to Markovianity we need to introduce the concept of  CP-divisible evolutions. To this end, consider a family of CP-maps $\Phi_t$ depending smoothly on $t$, which represents the time-parameter associated to the evolution. We assume that for any two times $t$ and $s$ ($t\geq s$) one can define an intermediate map $\Phi_{t,s}$ satisfying the property $\Phi_t = \Phi_{t,s}\circ\Phi_s$. This kind of one-parameter families, or evolutions, are called divisible\footnote{{Note that if $\Phi_t$ is invertible $\forall \,t$, then it is trivially divisible as one can set $\Phi_{t,s} \equiv \Phi_t \circ \Phi_s^{-1}$}.}. In particular, if all the intermediate maps are positive preserving, the global dynamics is called P-divisible; in the same way, if the intermediate maps are all CP, then the corresponding dynamics is called CP-divisible. In the following we will identify CP-divisible evolutions with Markovian ones, according to the most canonical notion of quantum Markovianity~\cite{rivas2014quantum}. It should be noticed that for CP-divisible evolutions all the contrast functions are monotonically non-increasing, as it can be verified from their derivative:
	\begin{align}
		\frac{\de}{\de t} H_g( \Phi_t(\rho) || \Phi_t(\sigma)) =\lim_{\varepsilon\rightarrow0} \frac{H_g( \Phi_{t+\varepsilon,t}\Phi_t(\rho) || \Phi_{t+\varepsilon,t}\Phi_t(\sigma))-H_g( \Phi_t(\rho) || \Phi_t(\sigma))}{\varepsilon}\leq 0\;,
		\label{cf:eq:183B}
	\end{align}
	where we used the contractivity of the contrast functions under the action of quantum channels. Since $\rho$ and $\sigma$ are arbitrary, the same result also holds for quantum Fisher metrics, and this behaviour is referred to as monotonic degradation of information, which also justifies the identification of CP-divisibility with Markovian dynamics.
	
	
	Since $\Phi_{t,s}$ is defined for any $t\geq s$, one can translate the study of the semigroup to the one of the generators, defined by the limit:
	\begin{align}
		\lind_t := \lim_{\varepsilon\rightarrow0}\, \frac{\Phi_{t+\varepsilon,t}-\idO}{\varepsilon}\,.\label{cf:eq:generatorDef}
	\end{align}
	
	\begin{figure}[h!]
		\centering
		\includegraphics[width=1\linewidth]{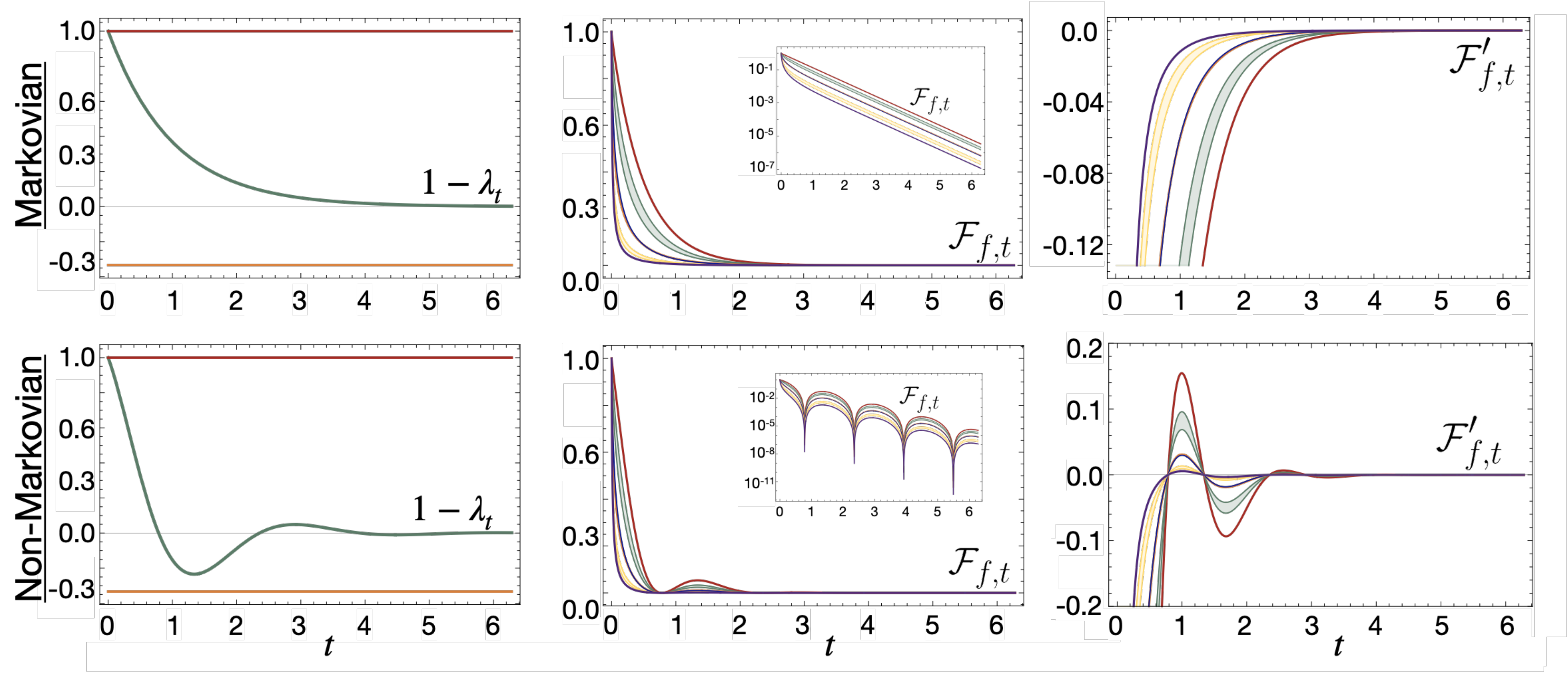}
		\caption{Evolution of the quantum Fisher informations under the action of the depolarising channel $\Delta_{\lambda_t} (\rho) = (1-\lambda_t)\rho +\lambda_t\,\frac{\id}{2}$ on a qubit. The time dependent contraction coefficient are $\lambda_t^M = 1-e^{-t}$ and $\lambda_t^{NM} = 1-e^{-t}\cos (2t)$, respectively in the Markovian and non-Markovian case. Notice that non-Markovianity is associated to a local increase in the value of $1-\lambda_t$. In the first panel we show the evolution of $\lambda_t$, in the second of $\mathcal{F}_{f,t}$ (the inset is in logarithmic scale) and in the third the behaviour of $\mathcal{F}'_{f,t}$. Non monotonicity in $\lambda_t$ are mirrored in the change of sign of $\mathcal{F}'_{f,t}$. The colour scheme is the same as in Fig.~\ref{cf:fig:figstandardmonotones2}.}
		\label{cf:fig:nonmarkovflux}
	\end{figure}
	
	It can be shown that an invertible evolution is divisible if and only if $\lind_t$ can be written in the time-dependent GKLS form~\cite{hall2014canonical}, namely:
	\begin{align}
		\lind_t[\rho] = -i[H(t),\rho] + \sum_{\alpha=1}^{d^2-1} \;\lambda_\alpha(t)\,\norbra{A_\alpha(t) \,\rho\, A_\alpha^\dagger(t) - \frac{1}{2}\{A_\alpha^\dagger(t) \,A_\alpha(t), \rho\}}\,,\label{cf:eq:lindDiagonal}
	\end{align}
	where $H(t)$ is a Hermitian matrix, $\lambda_\alpha(t)$ are real scalars called rates, while the operators ${A_\alpha}(t)$, which go under the name of jump operators, are traceless matrices which are orthonormal with respect to the Hilbert-Schmidt scalar product (i.e., $\Tr{A_\alpha^\dagger A_\beta}=\delta_{\alpha\beta}$). In this context, it can be proven that CP-divisibility is equivalent to the requirement that the rates are positive at all times ($\lambda_\alpha(t)\geq0$).
	
	Thanks to the linearity of the derivative, one can express the  evolution of the Fisher information metric under divisible dynamics as a sum of independent fluxes, each associated to a single rate $\lambda_\alpha(t)$. In particular, our main object of interest is given by the following derivative:
	\begin{align}
		\frac{\de}{\dt}H_{g}(\Phi_t(\pi)||\Phi_t(\pi+\varepsilon \,\delta\rho))&=\frac{\varepsilon^2}{2} \frac{\de}{\dt} \Tr{\Phi_t(\delta\rho)\, \J_f^{-1}\big|_{\Phi_{t}(\pi)} [\Phi_t(\delta\rho) ]} +\bigo{\varepsilon^3} =\label{cf:eq:201B}\frac{\varepsilon^2}{2} \,\mathcal{F}'_{f,t} +\bigo{\varepsilon^3}\,,
	\end{align}
	where we used the shorthand notation for the Fisher information $\mathcal{F}_{f,t} := \mathcal{F}_{f,\Phi_t(\pi)}(\Phi_t(\delta\rho))$.
	In the following we also use the notation $\pi_t := \Phi_t(\pi)$ and $\delta\rho_t := \Phi_t(\delta\rho)$. Then, we have that:

	\begin{theorem}\label{cf:thm:FisherFlow}
		For any divisible dynamics, let $\{A_\alpha(t)\}$ and $\{\lambda_\alpha(t)\}$ be respectively the time dependent jump operators and time dependent rates, defined according to Eq.~\eqref{cf:eq:lindDiagonal}. Then, the derivative of the Fisher information takes the form:
		\begin{align}
			\mathcal{F}'_{f,t} = \sum_\alpha\; \lambda_\alpha(t)\, \,\mathcal{I}^f_\alpha(t)\,,\label{cf:eq:fisherRates}
		\end{align}
		where the current $\mathcal{I}^f_\alpha(t)$ is given by:
		\begin{align}
			\mathcal{I}^f_\alpha(t) = - 2\,\int_0^1 \de N_g (s)\;&\bigg(\Tr{\pi_t\, [A_\alpha(t),B_{s}(t)^\dagger]^\dagger\,[A_\alpha(t),B_{s}(t)^\dagger] }+s\,\Tr{\pi_t\, [A_\alpha(t),B_{s}(t)]^\dagger\,[A_\alpha(t),B_{s}(t)] }\bigg)\,,\label{cf:eq:flowFisherInformation}
		\end{align}
		and the measure $\de N_g(s)$ is the one used in Eq.~\eqref{cf:eq:thm2HGexpIntegral}, while the operators $B_{s}(t)$ are defined as:
		\begin{align}
			B_{s}(t) := (\LL_{\pi_t}+s\,\RR_{\pi_t})^{-1}[\delta\rho_t]\,.\label{cf:eq:bOperators}
		\end{align}
	\end{theorem}
	
	The proof of this theorem is deferred to App.~\ref{app:sec:fisherFlow}. It should be noticed that the two traces in Eq.~\eqref{cf:eq:flowFisherInformation} are positive definite, as they take the form $\Tr{\pi_t\,X^\dagger X}$ for some operator $X$. This implies that the currents $\mathcal{I}^f_\alpha(t)$ are always negative, showing that the summands in Eq.~\eqref{cf:eq:fisherRates} can become positive only if the corresponding rate $\lambda_\alpha(t)$ becomes negative, i.e., in the presence of non-Markovianity. In the same way, we see that $\mathcal{F}'_{f,t}$ will always be negative for Markovian dynamics, signalling the expected monotonic contraction of the Fisher information. The effects of this decomposition are exemplified in Fig.~\ref{cf:fig:nonmarkovflux}, where we plotted the Fisher information and its derivative for a depolarising channel, both in the Markovian and non-Markovian regime. As it can be seen, in this case the oscillations in $\mathcal{F}'_{f,t}$ mirror the onset of non-Markovianity.
	
	\begin{additional-info}{Box 5. Examples of Fisher information currents}
		In order to give a more practical feeling of the expression in Eq.~\eqref{cf:eq:flowFisherInformation}, we present here some specific examples for which the currents $\mathcal{I}^f_\alpha(t)$ take a particularly simple form. The first case that is interesting to study is the one of classical evolutions, corresponding to jump operators all of the form $A_{i\leftarrow j} = \ketbra{i}{j}$, together with the requirement that all the observables commute with $\pi_t$. Thanks to this fact, $B_s(t)$ is simply given by $B_s(t) = B_s(t)^\dagger = \delta\rho_t /((1+s)\pi_t)$, where we use a slight abuse of notation to indicate the component-wise division. This allows us to rewrite Eq.~\eqref{cf:eq:flowFisherInformation} as:
		\begin{align}
			\mathcal{I}^{{c}}_{i\leftarrow j}(t) &= - 2\,\int_0^1 \de N_g (s)\,(1+s)\:\Tr{\pi_t \sqrbra{\ketbra{i}{j},\frac{\delta \rho_t}{(1+s) \pi_t}}^\dagger\sqrbra{\ketbra{i}{j},\frac{\delta \rho_t}{(1+s) \pi_t}} } =\\
			&=- \,\int_0^1 \de N_g (s)\;\frac{2}{(1+s)} \norbra{\frac{(\delta \rho_t)_{j}}{( \pi_t)_j} - \frac{(\delta \rho_t)_{i}}{( \pi_t)_i}}^2\, ( \pi_t)_j =\\\label{cf:eq:classicalFisherFlow}
			&= - \norbra{\frac{(\delta \rho_t)_{j}}{( \pi_t)_j} - \frac{(\delta \rho_t)_{i}}{( \pi_t)_i}}^2\, ( \pi_t)_j\,,
		\end{align}
		where in the last line we used the normalisation condition  on the measure $\de N_g(s)$ (see Eq.~\eqref{eq:normalisationFisherMeasure}):
		\begin{align}
			\int_0^1 \de N_g(s)\; \frac{2}{1+s} = 1\,,\label{cf:eq:normalisationDN}
		\end{align}
		This result indeed coincides with the one obtained for classical stochastic dynamics~\cite{abiuso2022characterizing}.
		
		Another case of particular interest is given by the flux of Bures metric (see Sec.~\ref{Bures}). This corresponds to a measure of the form $\de N(s) = \delta(s-1)\,\de s$. Then, thanks to the self-adjointness relation $B_{1}(t) =B_{1}(t)^\dagger$, by carrying out the integration one obtains:
		\begin{align}
			\mathcal{I}^B_\alpha(t) = -4 \, \Tr{\pi_t\, [A_\alpha(t),B_{1}(t)]^\dagger\,[A_\alpha(t),B_{1}(t)] }\,.
		\end{align}
		Interestingly, in this case the current is directly connected to the symmetric logarithmic derivative of the state. In fact, by inverting Eq.~\eqref{cf:eq:bOperators}, one obtains:
		\begin{align}
			B_{1}(t)\,\pi_t + \pi_t\, B_{1}(t) = \delta \rho_t\,.
		\end{align}
		Comparing this expression with the one in Eq.~\eqref{cf:eq:103}, it is apparent that $2 B_1$ corresponds to the symmetric logarithmic derivative $L$ in $\pi_t$ and in the direction of $\delta\rho_t$.  It should be pointed out that the expression of the Bures flux in terms of $L$ was already found in~\cite{luQuantumFisherInformation2010}.
		
		At the other extreme, the smallest among the contrast functions corresponds to the measure $\de N(s) =  \frac{\delta(s)}{2}\,\de s$ (see Sec.~\ref{harmonic}). Then, the current can be expressed as:
		\begin{align}
			\mathcal{I}^H_\alpha(t)  = - \, \Tr{\pi_t\, [A_\alpha(t),\delta\rho_t\,\pi_t^{-1}]^\dagger\,[A_\alpha(t),\delta\rho_t\,\pi_t^{-1}] }\,.
		\end{align}
		This shows how the formula in Thm.~\ref{cf:thm:FisherFlow} generalises the computations presented in~\cite{luQuantumFisherInformation2010} to the whole family of quantum Fisher information.
	\end{additional-info}
	
	We are now ready to give a complete characterisation of the relation between Markovianity and Fisher information metrics (see Fig.~\ref{cf:fig:fig1b}). As it was mentioned above, all Fisher information metrics monotonically contract under Markovian evolutions. With the hindsight of Thm.~\ref{cf:theo:CPiffContracts} it is then not surprising that the reverse also holds: 
	\begin{theorem}\label{cf:theo:markov_iff_contract}
		A divisible evolution $\Phi_{t}$ acting on a $d$-dimensional state space is P-divisible if it induces a monotonic decrease in the Fisher information at all times and on all states in $\mathcal{S}_d^{\mathrm{o}}$, the interior of the space of states. In formulae, P-divisibility of the evolution $\Phi_{t}$ is implied by the condition:
		\begin{align}
			\frac{\de}{\de t'}\,\Tr{\Phi_{t',t}(\delta\rho)\,\J_f^{-1}\big|_{\Phi_{t',t}(\rho)}[\Phi_{t',t}(\delta\rho)]}\bigg|_{t'=t}\leq 0\quad\forall \,t\,,\rho\,, \delta\rho\,,\label{eq:monotonicContraction}
		\end{align}
		where $\rho$ is an arbitrary point in $\mathcal{S}_d^{\mathrm{o}}$, and $\delta\rho$ an arbitrary perturbation in the tangent space. Moreover, the same condition for the dynamics $\Phi_t\otimes\idO_d$ is equivalent to the fact that $\Phi_t$ is CP-divisible, i.e., Markovian.
	\end{theorem}
	
	This theorem, together with Thm.~\ref{cf:theo:no-go} and~\ref{cf:theo:witness}, was already presented in~\cite{abiuso2022characterizing}, to which we refer for the proofs. Still, it should be noticed that Thm.~\ref{cf:theo:markov_iff_contract} is a direct corollary of Thm.~\ref{cf:theo:CPiffContracts}. Indeed, we can see that for invertible divisible dynamics, given $\varepsilon$ small enough the intermediate map $\Phi_{t+\varepsilon,t}$ satisfies the first two conditions of Thm.~\ref{cf:theo:CPiffContracts}, and in the limit $\varepsilon\rightarrow0$,  $\Phi^{-1}_{t+\varepsilon,t}\norbraB{ \mathcal{S}_d\cap\Phi_{t+\varepsilon,t}( \mathcal{S}_d)} =\mathcal{S}_d^{\mathrm{o}}$ (as $\Phi^{-1}_{t+\varepsilon,t}$ is $\varepsilon$-close to the identity superoperator). Moreover, it is easy to see that in the limit $\varepsilon\rightarrow0$,  Eq.~\eqref{eq:contractionContrast} and Eq.~\eqref{eq:monotonicContraction} are equivalent. Then, it directly follows from Thm.~\ref{cf:theo:CPiffContracts} that  the assumptions of Thm.~\ref{cf:theo:markov_iff_contract} force $\Phi_t$ to be P-divisible. Interestingly, for classical evolutions this directly implies Markovianity, as in this context there is no difference between P-divisibility and CP-divisibility.
	
	\begin{figure}
		\centering
		\includegraphics[width=.9\linewidth]{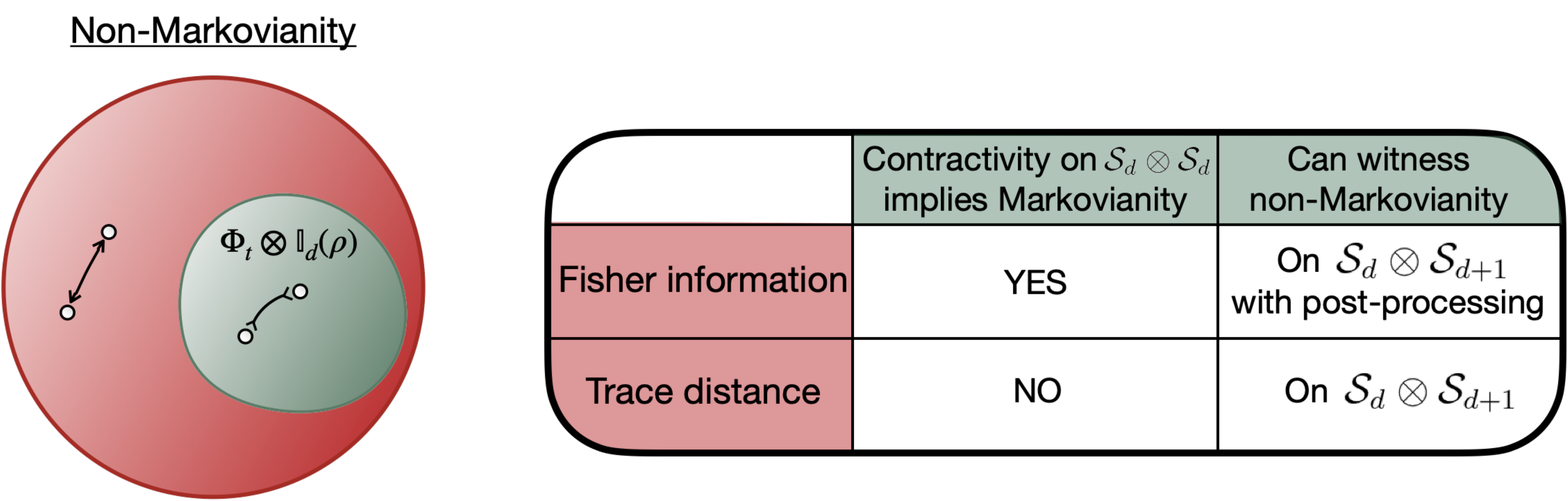}
		\caption{Illustrative summary of the results presented in Sec.~\ref{markovianity}. On the left we give a pictorial representation of Thm.~\ref{cf:theo:markov_iff_contract}: in red we depict the whole state space $  \mathcal{S}_d\otimes  \mathcal{S}_d$, while in green the image of the evolution, i.e., $\Phi_t\otimes \idO_d ( \mathcal{S}_d\otimes \mathcal{S}_d)$. Thm.~\ref{cf:theo:markov_iff_contract} tells us that a map is non-Markovian if and only if there exists at least two points in $\mathcal{S}_d\otimes \mathcal{S}_d$ (not necessarily in the image of the map) for which the Fisher distance increases. On the right, we compare the Fisher information with the most canonical quantity to witness non-Markovianity, namely the trace distance. While the monotone contractivity of the Fisher information implies the Markovianity of the dynamics, the same does not hold for the trace distance (see Thm.~\ref{cf:theo:markov_iff_contract} and \hyperref[box:fisherAndTrace]{Box 7} on the trace distance). On the other hand, supplying ancillas to the system allows for the detection of non-Markovianity through the latter quantity, while for the Fisher information one additionally needs some post-processing of the state (Thm.~\ref{cf:theo:no-go} and~\ref{cf:theo:witness}).}
		\label{cf:fig:fig1b}
	\end{figure}
	
	It should be noticed that the equivalence one can find for the contractivity of Fisher information is quite peculiar to this quantity. Indeed, one can contrast the result just obtained with the one for the trace distance, the most canonical quantity studied in the context of non-Markovianity (we provide a small review on this topic in~\hyperref[box:fisherAndTrace]{Box 7} at the end of the section): in this case one can explicitly construct classical non-Markovian dynamics for which the trace distance between any two points contracts. For quantum evolutions, on the other hand, the use of a $d$-dimensional ancilla is necessary to separate positive preserving maps from completely positive maps, as complete positivity of $\Phi_t$ is equivalent to the fact that $\Phi_t\otimes\idO_d$ is P. Despite the addition of a $d$-dimensional ancilla, examining the trace distance is not enough to obtain a result along the lines of Thm.~\ref{cf:theo:markov_iff_contract}. Indeed, it can be shown that one needs the ancillary space to be at least of dimension $d+1$ for the trace distance to expand in the presence of any non-Markovianity~\cite{bylicka2017constructive}. 
	
	The difference with the trace distance is actually even sharper. Indeed, in order for the trace distance to operationally witness the non-Markovianity of the evolution it is sufficient to use ancillas of dimension high enough (thanks to the construction in~\cite{bylicka2017constructive}, $d+1$ is enough). This means that an increase of trace distance can always be obtained on the image of $\Phi_t$ when a sufficient number of extra degrees of freedom are provided. Thm.~\ref{cf:theo:markov_iff_contract} on the other hand, ensures that an expansion in the Fisher information metrics always happens in the presence of non-Markovianity, but it does not say anything about whether the states needed to actually verify it can be physically prepared. If the violation of the monotonicity happens close to the boundary of the state space, for example, and the initial part of the evolution is particularly contracting, there is no way to actually detect the non-Markovian behaviour by looking at the Fisher information alone. Still, one could hope that by using a sufficient number of ancillas the Fisher information could provide witnesses for non-Markovianity. Somehow surprisingly, one can prove that this cannot be the case:
	\begin{theorem}
		\label{cf:theo:no-go}
		Given a divisible evolution $\Phi_t$, no ancillary degree of freedom of finite dimension or finite number of copies of the dynamics are sufficient to witness all possible non-Markovian evolutions via revivals of the Fisher distance between two initially prepared states.
	\end{theorem}
	The difference in behaviour of the trace distance and the Fisher information metrics arises from the linearity of the map $\Phi_t$ and the translational invariance of the trace distance~\cite{abiuso2022characterizing,bylicka2017constructive}: thanks to this property, if a witness exists anywhere on the state space, then it can always be translated into the image of $\Phi_t$. The Fisher information, on the other hand, has a strong dependence on the base-point, so that the same kind of argument cannot be applied.  We refer to~\cite{abiuso2022characterizing} for the proof of Thm.~\ref{cf:theo:no-go}.
	
	Despite the negative result of Thm.~\ref{cf:theo:no-go}, one can still define a witness based on Fisher information by using post-processing:
	\begin{theorem}
		\label{cf:theo:witness}
		Given an evolution $\Phi_t$, for any state $\rho$ and perturbation $\delta\rho$ defined on the system space and sufficiently many ancillary degrees of freedom, 
		it is possible to implement a class of CP-maps $F^{(t)}_{\delta\rho}$ depending on $\Phi_t$ and $\delta\rho$ that can be used to witness non-Markovianity at time~$t$ through the expansion of Fisher information. Formally, if the intermediate evolution $\Phi_{t+\de t,t}$ is Positive (CP in the quantum case), then
		\begin{align}
			\frac{\de}{\de s}\,\Tr{F^{(t)}_{\delta\rho}\circ\Phi_{s}(\delta\rho)\,\J_f^{-1}\big|_{F^{(t)}_{\delta\rho}\circ\Phi_{s}(\rho)}[F^{(t)}_{\delta\rho}\circ\Phi_{s}(\delta\rho)]}\bigg|_{s=t}\leq 0\quad\forall \,\rho\,, \delta\rho\,,
		\end{align}
		whereas in the presence of non-Markovianity there always exists at least one $\delta\rho$ for which the left-hand side is strictly positive.
		
		The minimal dimension of the ancilla for classical systems is $d_A=2$, while for quantum maps one needs $d_A = d+1$.
	\end{theorem}
	
	There is a shortcoming to this construction, though: in the definition of the post-processing $F^{(t)}_{\delta\rho}$ one needs to assume complete knowledge about the dynamics $\Phi_t$ until the onset of Markovianity. In this way, one either needs to try all the possible $\delta\rho$, or has to know in advance the structure of the dynamics in order to provide an explicit construction. Still, this example serves more as a proof of principle showing the possibility of designing post-processing filters to exploit the Fisher information for the detection of non-Markovianity. Again, the proof of this fact is contained in~\cite{abiuso2022characterizing}, where the explicit expression of $F^{(t)}_{\delta\rho}$ was also presented.

	Theorems~\ref{cf:theo:no-go} and~\ref{cf:theo:witness}
	complete the characterisation of the relation between Markovianity and contractivity of the Fisher information, both on the image of $\Phi_t$ and on the rest of the state space, as summarised in Fig.~\ref{cf:fig:fig1b}. We point out once again the importance of Thm.~\ref{cf:theo:markov_iff_contract}: both Chentsov theorem and its quantum generalisation by Petz identify the defining property of the Fisher information metric to be its contractivity under stochastic maps or quantum channels. Thm.~\ref{cf:theo:markov_iff_contract}, on the other hand, could be read off as saying that the defining property of Markovianity is that it contracts the Fisher information monotonically at all times. This second implication shows how natural the concept of Fisher metric is in the context of open system dynamics.
	
	\begin{additional-info}{Box 6. Markovianity for classical evolutions~\hyperref[box:fisherAndTrace]{$\rightarrow$}}\label{box:markovianEvolutionsClassical}
		Classical dissipative evolutions are described by stochastic maps, i.e., matrices $\Phi$ satisfying the two conditions:
		\begin{align}
			&\sum_{i}\;(\Phi)_{i,j} = 1\label{appM:eq:T_require}\,;\\
			&(\Phi)_{i,j} \geq 0\qquad \forall\;\; \{i,j\}\,,\label{appM:eq:T_requireP}
		\end{align}
		where the first condition ensures the conservation of total probability, while the second is needed to make sure that states are mapped into states.
		In complete analogy to the quantum case, a family of stochastic maps $\Phi_t$ depending smoothly on $t$ is called divisible if for any two times $t$ and $s$ ($t\geq s$) one can define an intermediate map $\Phi_{t,s}$ satisfying the relation $\Phi_{t} = \Phi_{t,s}\circ\Phi_s$. A divisible stochastic dynamics is Markovian if all the intermediate maps $\Phi_{t,s}$ are stochastic.
		
		The smoothness in $t$ allows to define the rate matrix $R_t$ through the equation
		\begin{align}
			R_t:=\lim_{\varepsilon\rightarrow 0}\,\frac{\Phi_{t+\varepsilon,t} - \idO}{\varepsilon}\,.\label{eq:rateMatrixClassical}
		\end{align}
		Then, since the composition of two stochastic maps is again stochastic, Markovianity holds if and only if $R_t$ generates a stochastic evolution for any time $t$. For this reason, it is useful to characterise this kind of rate matrices. To this end, consider the matrix $\Phi_{t+\varepsilon,t}\simeq \idO + \varepsilon\,R_t$. For the global evolution to be Markovian, this matrix should satisfy the two conditions in Eq.~\eqref{appM:eq:T_require} and Eq.~\eqref{appM:eq:T_requireP}, namely:
		\begin{align}
			&\sum_i \;\norbra{\delta_{ij} +\varepsilon\, (R_t)_{i,j}} = 1+\varepsilon\, \sum_i \; (R_t)_{i,j} = 1\,,\label{cf:eq:Rate_conditions}\\
			&\delta_{ij} + \varepsilon\, (R_t)_{i,j}\geq 0 \qquad\forall\;\; \{i,j\}\,,\label{cf:eq:Rate_conditions2}
		\end{align}
		where $\delta_{ij}$ denotes the Kronecker delta. From the first condition we can deduce that $\sum_i  (R_t)_{i,j} = 0$. In particular, highlighting the diagonal terms, one obtains $R_{j,j}^{(t)} = -\sum_{i\neq j}  R_{i,j}^{(t)} $. Matrices satisfying this constraint can be decomposed as:
		\begin{align}
			R_t=\sum_{i\neq j} \; a^{(t)}_{i\leftarrow j}\left( \ketbra{i}{j}-\ketbra{j}{j} \right)\label{appM:rateMatrixTracePreserving}
		\end{align}
		for $a^{(t)}_{i\leftarrow j}$ some real coefficients. We assume that this condition is always satisfied, also for non-Markovian evolutions, as it corresponds to the requirement that the dynamics preserves the normalisation. In fact, since non-Markovian evolutions are trace preserving on their domain, one can argue by linearity that this condition can be extended to the whole space of states.
		
		Moreover, the condition in Eq.~\eqref{cf:eq:Rate_conditions2} implies that $ (R_t)_{i,j}\geq 0$ whenever $i\neq j$. In the parametrisation above this means that the rates satisfy $a^{(t)}_{i\leftarrow j}\geq 0$. Hence, Markovianity corresponds to the request of having positive rates $a^{(t)}_{i\leftarrow j}$ for all times. 
	\end{additional-info}
	
	\begin{additional-info}{Box 7. Use of the trace distance in non-Markovianity~\hyperref[box:markovianEvolutionsClassical]{$\leftarrow$},\hyperref[retrodiction]{$\rightarrow$}}\label{box:fisherAndTrace}
		The study of non-Markovianity is mainly carried out in terms of distinguishability distances. Indeed, since Markovian dynamics leads to a monotonic decrease in these distances, any increase thereof signals the appearance of non-Markovianity. In this context, the most used quantity is given by the trace distance:
		\begin{align}
			D_{\rm Tr}(\rho,\sigma)=\Tr{|\rho-\sigma|}\,,
		\end{align}
		which can be connected to the maximal probability $p_d$ of distinguishing $\rho$ from $\sigma$ in a single shot measurement, thanks to the relation $p_d (\rho,\sigma)= (1+D_{\rm Tr}(\rho,\sigma))/2$~\cite{nielsen2002quantum}. Moreover, since it is translational invariant, it is particularly appealing for analytical calculations. In particular, suppose $\rho$ and $\sigma$ are two probability vectors, and define $\delta\rho := \sigma-\rho$. Then, the trace distance is given by:
		\begin{align}
			D_{\rm Tr}(\rho,\sigma) =D_{\rm Tr}(\rho,\rho + \delta\rho) = \Tr{|\delta\rho|}\,.
		\end{align}
		For classical systems, if the evolution is described by the rate matrix $R_t$, it is straightforward to see that:
		\begin{align}
			\frac{\de}{\dt}\, D_{\rm Tr}(\rho,\sigma) &= \frac{\de}{\dt}\,\Tr{|\delta\rho|}=\sum_i   \;\frac{\de}{\dt} |{\delta\rho_i}| = \sum_i  \;{\rm sign}(\delta\rho_i) \delta\dot{\rho_i}=\\
			&=\sum_i  {\rm sign}(\delta\rho_i)  \sum_{j} \,(R_t)_{i,j}\, d_j =\sum_{i\neq j} \;{\rm sign}(\delta\rho_i)\left(a_{i\leftarrow j}^{(t)}\delta\rho_j - a_{j\leftarrow i}^{(t)}\delta\rho_i\right)=\\
			&=\sum_{i\neq j} \;a_{j\leftarrow i}^{(t)}\,\delta\rho_i \left({\rm sign}(\delta\rho_j)- {\rm sign}(\delta\rho_i)\right)\,.\label{trace_deriv}
		\end{align}
		where we used the parametrisation of the rate matrix in Eq.~\eqref{appM:rateMatrixTracePreserving}, and in the last line we swapped the indexes of the first term in order to put the coefficient $a_{j\leftarrow i}^{(t)}$ in evidence. It should be noticed that if all the rates are positive, then the sum will be negative: in fact, either ${\rm sign}(\delta\rho_j)= {\rm sign}(\delta\rho_i)$, in which case the term inside the parenthesis is zero, or ${\rm sign}(\delta\rho_j)= -{\rm sign}(\delta\rho_i)$, so that $\delta\rho_i({\rm sign}(\delta\rho_j)- {\rm sign}(\delta\rho_i))=-2\,\delta\rho_i\,{\rm sign}(\delta\rho_i)=-2|\delta\rho_i|$. 
		This calculation shows explicitly how the trace distance decreases under Markovian maps. 
		
		Interestingly, it is a well known fact that an evolution is Markovian if and only if the trace norm of any vector $\boldsymbol{v}$ decreases monotonically~\cite{rivas2014quantum}. For this reason, it would be natural to expect the same to hold for the trace distance as well. It is easy to see, though, that this is false: one simple counterexample can be given in dimension $d=2$, with $a_{1\leftarrow 2}<0$ and $a_{2\leftarrow 1}>0$, and the additional condition that $|a_{2\leftarrow 1}|>|a_{1\leftarrow 2}|$. With this choice of rates in Eq.~\eqref{trace_deriv}, since in dimension $2$ tracelessness of $\delta\rho$ implies that $\delta\rho_1=-\delta\rho_2$, it is easy to check that the derivative of the trace distance stays negative. Still, there is no contradiction between the two results: indeed, the trace distance can only access vectors of the form $\delta \rho= \sigma-\rho$, i.e., that are traceless. This reduces the dimension of the vectors tested by one. In fact, by choosing the traceful vector $v_i=\delta_{i2}$, one is able to witness non-Markovianity in the counterexample just presented.
		
		Suppose now to have access to extra ancillary degrees of freedom on which the dynamics acts trivially, i.e., the global evolution is given by $\Phi_t\otimes\idO_{d_A}$, where $d_A$ is the dimension of the ancilla. Then, one can always find two states $\rho$ and $\sigma$ on the composite space such that tracing out the system gives $\TrR{S}{\rho-\sigma} \neq 0$, while $\TrR{A}{\rho-\sigma} = 0$, so that the total trace is zero. In this way, ancillary degrees of freedom give access to traceful vectors, and so to the possibility of witnessing non-Markovianity. A similar argument was also presented in~\cite{bylicka2017constructive} for the case of quantum dynamics. These results lead to the following:
		\begin{theorem}
			\label{lem:ancilla}
			Given a divisible dynamics $\Phi_t$, there always exists an ancilla of finite dimension $d_A$ on which the dynamics acts trivially (i.e., the total evolution is given by $\Phi_t\otimes\idO_{d_A}$) such that one can witness any non-Markovianity in $\Phi_t$ via revivals in the trace distance between initially prepared states.
		\end{theorem}
		
		For quantum states the minimal dimension of the ancilla is $d_A = d+1$~\cite{bylicka2017constructive}. It should be noticed, though, that in order to actually speak about complete positivity one always needs an ancilla of dimension at least $d$ (as complete positivity coincides with the positivity of the map $\Phi\otimes\idO_d$). In this way, the trace distance needs one extra dimension more than the minimum in order to witness non-Markovianity. A similar result actually holds also for classical systems. In this case, one only needs to enlarge the state space by one extra dimension. Then, for any vector $\boldsymbol{v}$ on the original space, one can always construct a traceless state on this extended space as:
		\begin{align}
			\delta\rho = \begin{cases}
				\delta\rho_i = v_i & {\rm if }\;i\in\curbra{1,\dots,d}\,;\\
				\delta\rho_{i}=-\sum_{j=1}^{d}\delta\rho_j & {\rm if }\;i\equiv d+1\,.
			\end{cases}
		\end{align}
		This state satisfies $\Tr{|\delta\dot{\rho}|} = \Tr{|\dot{\boldsymbol{v}}|}$, which proves the claim. Still, in this case it should be noticed that this construction cannot be carried out by adjoining an ancilla, as in that case the minimum dimension of the state space is $d_A=2$. Nonetheless, this example is useful as it shows that two dimensional ancillas are enough, and that in principle it would be sufficient to embed the dynamics $\Phi_t$ in a space only one dimension larger than the original space in order to witness classical non-Markovianity.
	\end{additional-info}
	
	\subsection{Retrodiction and Fisher information}\label{retrodiction}
	
	In the literature about non-Markovianity, one of the most used expression is \emph{backflow of information}. Still, by looking at the usual quantifiers defined to assess it, one might be surprised to discover that for the most part they are actually distinguishability measures accounting for the statistical difference between states \emph{at time $t$} - such as the trace distance and Fisher information, on which we focused in the previous section.
	Whereas one could in principle justify the interpretation of non-monotonicity for these quantities as actual backflow of information, here we take a different approach: we prove, in fact, that the contractivity of the Fisher information is in one-to-one correspondence with the monotonic degradation in the ability of an agent to retrodict the \emph{initial state} of the system (Thm.~\ref{cf:theo:recovery_fisher}). 
	
	Before doing so, we need some formalism. First of all, we present a way to simulate the Fisher scalar product $\J_f^{-1}\big|_{\Phi_{t}(\pi)}$ at time $t$ by the scalar product at time $\J_{f'}^{-1}\big|_{\pi}$ at time $t=0$. This is obtained by the following rearranging:
	\begin{align}
		\nonumber
		K_{f,\Phi_t(\pi)}(\Phi_t(A),\Phi_t(B))&= \Tr{\Phi_t(A)\, \J_f^{-1}\big|_{\Phi_{t}(\pi)} [\Phi_t(B) ]} =\Tr{A\, \J_{f'}^{-1}\big|_{\pi}\sqrbra{\norbra{ \J_{f'}\big|_{\pi}\,\Phi_t^\dagger \;\J_f^{-1}\big|_{\Phi_{t}(\pi)}}\Phi_t(B) }} =\\ \label{cf:eq:generalisedPetzDef}
		&=\Tr{A\, \J_{f'}^{-1}\big|_{\pi}[\widetilde{\Phi}_t\big|^{(f',f)}_\pi\,\Phi_t(B)] } = K_{f',\pi}(A,\widetilde{\Phi}_t\big|^{(f',f)}_\pi\Phi_t(B))\,,
	\end{align}
	where in the last line we have implicitly defined the \emph{generalised Petz recovery map} $\widetilde{\Phi}_t\big|^{(f',f)}_\pi$, generalizing the famous recovery introduced by Petz~\cite{ohya2004quantum}, which is obtained for $f(x)=f'(x)=\sqrt{x}$. In this way, the evolution of any Fisher scalar products can be modelled without the need to actually evolve the state, at the cost of introducing a time dependent vector $\widetilde{\Phi}_t\big|^{(f',f)}_\pi\Phi_t(B)$.  
	
	We provide here a list of properties that $\widetilde{\Phi}_t\big|^{(f',f)}_\pi$ satisfies and we defer the proofs of these facts to  \hyperref[box:proofGeneralisedPetz]{Box 8} at the end of the section:
	\begin{enumerate}
		\item $\widetilde{\Phi}_t\big|^{(f',f)}_\pi$ reduces to the Bayes rule for classical dynamics, i.e., if $\Phi_t$ is a stochastic map, and $\pi$ a diagonal state, for any $f$ and $f'$ one has that:\label{item:reducesToBayes}
		\begin{align}
			\widetilde{\Phi}_t\big|^{(f',f)}_\pi(\boldsymbol{\cdot}) = \sum_{i,j}\; \frac{(\Phi_t)_{j,i} \, \pi_i}{(\Phi_t(\pi))_j}\, \ketbra{i}{j} \boldsymbol{\cdot}\ketbra{j}{i}\,;\label{eq:classicalBayes}
		\end{align}
		\item if $\pi$ is full-rank, $\widetilde{\Phi}_t\big|^{(f',f)}_\pi$ is trace preserving;\label{item:tracePreserving}
		\item the operator $\widetilde{\Phi}_t\big|^{(f',f)}_\pi\Phi_t$ is positive, and self-adjoint with respect to $\J_{f'}^{-1}\big|_{\pi}$, meaning that:\label{item:positiveSelfAdjoint}
		\begin{align}
			K_{f',\pi}(A,\widetilde{\Phi}_t\big|^{(f',f)}_\pi\Phi_t(B)) = K_{f',\pi}(\widetilde{\Phi}_t\big|^{(f',f)}_\pi\Phi_t(A),B)\,;\label{eq:selfAdjointRetrodiction}
		\end{align}
		\item the state $\pi$, called the \emph{prior state} due to the analogy with Bayes rule in Eq.~\eqref{eq:classicalBayes}, is perfectly retrieved by  $\widetilde{\Phi}_t\big|^{(f',f)}_\pi\Phi_t$, i.e.:\label{item:retrievesThePrior}
		\begin{align}
			\widetilde{\Phi}_t\big|^{(f',f)}_\pi\Phi_t(\pi) = \pi\,;\label{eq:retrievesThePrior}
		\end{align}
		\item If $f'(x)\leq f(x)$ for every $x\in \RR^+$, then the spectrum of $\widetilde{\Phi}_t\big|^{(f',f)}_\pi\Phi_t$ is completely contained in $[0,1]$. Moreover, as it can be seen from Eq.~\eqref{eq:retrievesThePrior}, 1 is always part of the spectrum;\label{item:spectrumInZeroOne}
		\item the transformation associating to a map its generalised Petz recovery, which we denote by $P_{(f',f),\pi}(\Phi_t):=\widetilde{\Phi}_t\big|^{(f',f)}_\pi$, can be reversed as:\label{item:involutivity}
		\begin{align}
			P_{(f,f'),\Phi_t(\pi)}(P_{(f',f),\pi}(\Phi_t)) = \Phi_t\,;
		\end{align}
		\item \label{item:markovianity} In the case of divisible evolutions $\Phi_t  = \Phi_{t,s}\circ\Phi_s$, one can express $P_{(f',f),\pi}(\Phi_t)$ in terms of $P_{(f',f''),\pi}(\Phi_{s})$ and $P_{(f'',f),\Phi_s(\pi)}(\Phi_{t,s})$, for any $f''$, as:
		\begin{align}
			P_{(f',f),\pi}(\Phi_{t})  = P_{(f',f''),\pi}(\Phi_{s})\circ P_{(f'',f),\Phi_s(\pi)}(\Phi_{t,s})\,.
		\end{align}
		Similarly with what happens for the adjoint, the generalised Petz recovery map composes from left to right.
	\end{enumerate}
	
	\begin{figure}
		\centering
		\includegraphics[width=.9\linewidth]{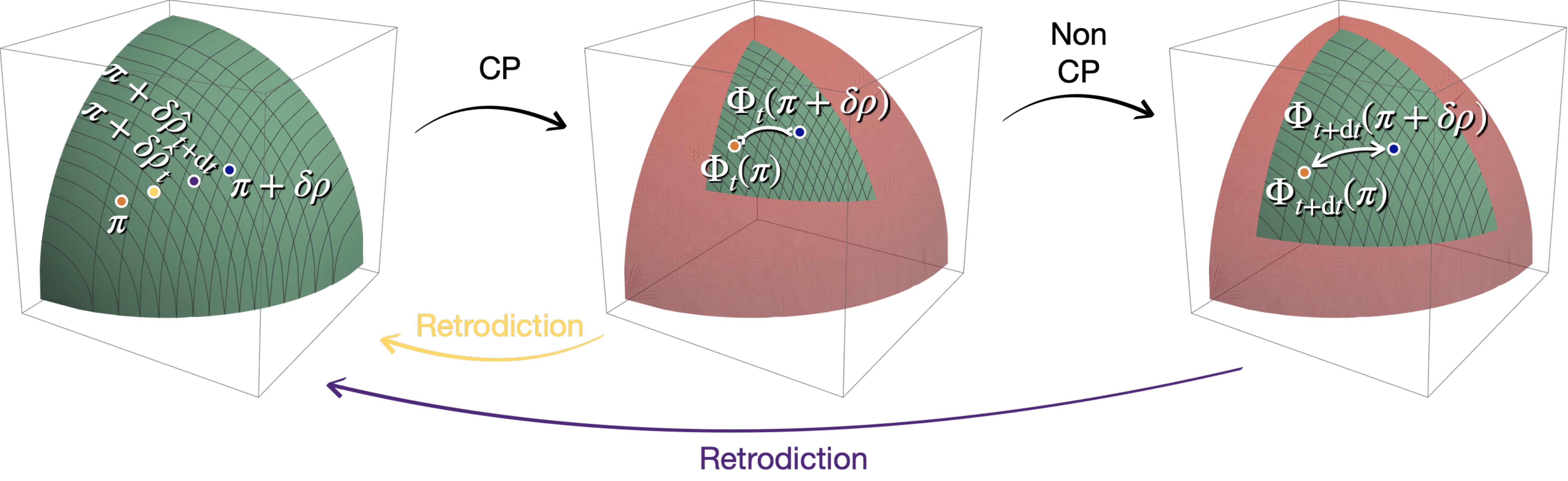}
		\caption{Depiction of the content of Thm.~\ref{cf:theo:recovery_fisher}: consider a dynamics $\Phi_t$ which is Markovian till some time $t$. This implies that the Fisher distance between any two points on the state space contracts. If between time $t$ and $t+\dt$ we observe non-Markovianity through the Fisher distance, i.e., two points in the image of $\Phi_t$ gets farther away, then retrodicting at time $t+\dt$ gives a better result than doing so at time $t$ (we use the notation $\delta\hat\rho_t := \widetilde{\Phi}_t\big|^{(f',f)}_\pi\Phi_t(\delta\rho)$).}
		\label{fig:retrodictionmarkov}
	\end{figure}
	
	Thanks to Eq.~\eqref{eq:classicalBayes}, one can consider $\widetilde{\Phi}_t\big|^{(f',f)}_\pi$ as some kind of generalisation of Bayes rule to quantum systems. This interpretation can be corroborated in the following context:
	consider a state $\pi+\delta\rho$, with $\delta\rho$ a small perturbation, and evolve it according to $\Phi_t$; at this point, $\widetilde{\Phi}_t\big|^{(f',f)}_\pi$ is applied to recover as much information about the initial state as possible. 
	Thanks to the definition in Eq.~\eqref{cf:eq:generalisedPetzDef}, there is a one-to-one correspondence between the Fisher information at time $t$, and the  Fisher scalar product at the initial state with the retrodicted vector. In formulae, this means that:
		\begin{align}
			K_{f,\Phi_t(\pi)}(\Phi_t(\pi+\delta\rho),\Phi_t(\pi+\delta\rho))=
			K_{f',\pi}(\pi+\delta\rho,\widetilde{\Phi}_t\big|^{(f',f)}_\pi\Phi_t(\pi+\delta\rho))\;. 
		\end{align}
		Using the property in Eq.~\eqref{eq:retrievesThePrior},  we can further simplify the above equality to:
		\begin{align}
			1 + K_{f,\Phi_t(\pi)}(\Phi_t(\delta\rho),\Phi_t(\delta\rho))=
			1 + K_{f',\pi}(\delta\rho,\widetilde{\Phi}_t\big|^{(f',f)}_\pi\Phi_t(\delta\rho))\;, 
		\end{align}
		which makes explicit the appearance of the Fisher information. As the above equality shows, increases/decreases of the Fisher information at time $t$ correspond to increases/decreases of the overlap between $\pi+\delta\rho$ and its retrodiction, as measured by the Fisher scalar product.
	
	At the level of distinguishability measures, the quality of the retrieval can be quantified by the contrast function:
	\begin{align}\nonumber
		&H_{g(f')}(\pi+\delta\rho||\widetilde{\Phi}_t\big|^{(f',f)}_\pi\,\Phi_t(\pi+\delta\rho)) = H_g(\pi+\delta\rho||\pi+\widetilde{\Phi}_t\big|^{(f',f)}_\pi\,\Phi_t(\delta\rho))=\\ \label{cf:eq:a1214}
		&\qquad=\frac{1}{2}\,\Tr{(\delta\rho-\widetilde{\Phi}_t\big|^{(f',f)}_\pi\,\Phi_t(\delta\rho)) \, \J_{f'}^{-1}\big|_{\pi}[(\delta\rho-\widetilde{\Phi}_t\big|^{(f',f)}_\pi\,\Phi_t(\delta\rho))]}+\ord{|\delta\rho|^3}\,,
	\end{align}
	where we introduced the notation $g(f')$ to indicate the contrast function that locally expands to $\J_{f'}^{-1}\big|_{\pi}$ (i.e., $g(f')$ and $f'$ are related by Eq.~\eqref{cf:eq:correspondenceFG}), and in the first line we used Eq.~\eqref{eq:retrievesThePrior} to apply the operators to $\delta\rho$ only. Then, the interpretation of  $\widetilde{\Phi}_t\big|^{(f',f)}_\pi$  as a recovery map might suggest that if $\Phi_t$ is non-Markovian at time $t$, then the divergence above will decrease at that time: thanks to the backflow of information, the ability of an agent to retrodict the initial state improves. Still, from a mathematical point of view, in principle there is no intuitive reason for this to be the case: in fact, the behaviour of $H_g(\rho||\sigma)$ is characterised only if the same channel is applied on both states, whereas in the equation above $\widetilde{\Phi}_t\big|^{(f',f)}_\pi\,\Phi_t$ is applied only on the right.  Moreover, as we discuss at the end of the section, $\widetilde{\Phi}_t\big|^{(f',f)}_\pi\,\Phi_t$ might not even be a channel in general, further eroding any mathematical justification to the intuition discussed. Still, the following theorem bridges the gap between the intuitive interpretation of Eq.~\eqref{cf:eq:a1214} and its actual mathematical form, showing that indeed $H_{g(f')}(\pi+\delta\rho||\widetilde{\Phi}_t\big|^{(f',f)}_\pi\,\Phi_t(\pi+\delta\rho))$ decreases for non-Markovian dynamics:
	\begin{theorem}
		\label{cf:theo:recovery_fisher}
		Given a divisible dynamics $\Phi_t$ and a prior state $\pi$, define the generalised Petz map $\widetilde{\Phi}_t\big|^{(f',f)}_\pi$ according to Eq.~\eqref{cf:eq:generalisedPetzDef}. If $f'(x)\leq f(x)$ for every $x\in \RR^+$, we have the implication:
		\begin{align}
			\exists\, \delta \rho\,,\;\; \frac{\de}{\dt}\Tr{\Phi_t(\delta\rho) \,\J_f^{-1}\big|_{\Phi_t(\pi)} [\Phi_t(\delta\rho)]}>0\;\;\implies\;\;\;\;\frac{\de}{\dt}\,H_{g(f')}(\pi+\delta\rho||\widetilde{\Phi}_t\big|^{(f',f)}_\pi\,\Phi_t(\pi+\delta\rho))<0  \,.\label{eq:90}
		\end{align}
		Since the Fisher information is monotonically contractive under Markovian dynamics, the condition on the left is only attainable in the non-Markovian regime.
	\end{theorem}

	To the best of our knowledge, this result is the first that explicitly considers backflow of information in a single state at time $t=0$: indeed, whereas here we compare the initial condition with our best guess about it, the usual approach to non-Markovianity is to compare the behaviour of two different states at time $t$. There are two messages one should get from Thm.~\ref{cf:theo:recovery_fisher}: first, it legitimates the intuitive interpretation of $\widetilde{\Phi}_t\big|^{(f',f)}_\pi$ as a state retrieval map; second, it directly connects the contractivity properties of the Fisher information at time $t$ (which is a distinguishability measure), with the improvement in the ability of an agent to retrodict the initial state of the system.
	
	\begin{proof}
		Starting from the expansion in Eq.~\eqref{cf:eq:a1214}, and ignoring corrections of the order $\ord{|\delta\rho|^3}$,  we can reassemble the terms as:
		\begin{align}
			&\frac{\de}{\dt}H_{g(f')}(\pi+\delta\rho||\widetilde{\Phi}_t\big|^{(f',f)}_\pi\,\Phi_t(\pi+\delta\rho))
			=\frac{1}{2}\frac{\de}{\dt}\,\Tr{((\idO-\widetilde{\Phi}_t\big|^{(f',f)}_\pi\,\Phi_t)(\delta\rho)) \, \J_{f'}^{-1}\big|_{\pi}[((\idO-\widetilde{\Phi}_t\big|^{(f',f)}_\pi\,\Phi_t)(\delta\rho))]}  =\\
			&= \frac{1}{2}\frac{\de}{\dt}\,\Tr{\delta\rho\; \J_{f'}^{-1}\big|_{\pi}[((\idO-\widetilde{\Phi}_t\big|^{(f',f)}_\pi\,\Phi_t)^2(\delta\rho))]} =-\Tr{\delta\rho\; \J_{f'}^{-1}\big|_{\pi}[((\idO-\widetilde{\Phi}_t\big|^{(f',f)}_\pi\,\Phi_t)\circ\norbra{\frac{\de}{\dt}\,\widetilde{\Phi}_t\big|^{(f',f)}_\pi\,\Phi_t}(\delta\rho))]}\,,
		\end{align}
		where in the first equality is just a rewriting of Eq.~\eqref{cf:eq:a1214}, then we used the self-adjointness property~\ref{item:positiveSelfAdjoint} in Eq.~\eqref{eq:selfAdjointRetrodiction}  to group the superoperator $(\idO-\widetilde{\Phi}_t\big|^{(f',f)}_\pi\,\Phi_t)$ on the right, and finally we carried out the time derivative.
		
		Now, it should be noticed that thanks to the property~\ref{item:spectrumInZeroOne} the spectrum of  $(\idO-\widetilde{\Phi}_t\big|^{(f',f)}_\pi\,\Phi_t)$ is also contained in $[0,1]$. On the other hand, $-\frac{\de}{\dt}\,\widetilde{\Phi}_t\big|^{(f',f)}_\pi\,\Phi_t$ has positive eigenvalues if and only if the Fisher information corresponding to $f$ is contractive, as one can see from the standard rewriting:
		\begin{align}
			-\,\Tr{A\; \J_{f'}^{-1}\big|_{\pi}[\norbra{\frac{\de}{\dt}\,\widetilde{\Phi}_t\big|^{(f',f)}_\pi\,\Phi_t)}(A)]} = -\frac{\de}{\dt}\Tr{\Phi_t(A)\; \J_f^{-1}\big|_{\Phi_t(\pi)}[\Phi_t(A))]}\,. 
		\end{align}
		Then, suppose that the Fisher information expands at some point. This means that there exists at least one negative eigenvalue of $-\frac{\de}{\dt}\,\widetilde{\Phi}_t\big|^{(f',f)}_\pi\,\Phi_t$, which we denote by $\lambda$. Also, denote by $\delta\tilde{\rho}$ the corresponding eigenoperator. Then, the right hand side of Eq.~\eqref{eq:90} reads:
		\begin{align}
			K_{f',\pi}\norbra{\delta\tilde{\rho},(\idO-\widetilde{\Phi}_t\big|^{(f',f)}_\pi\,\Phi_t)\circ\norbra{-\frac{\de}{\dt}\,\widetilde{\Phi}_t\big|^{(f',f)}_\pi\,\Phi_t}(\delta\tilde{\rho})} = \lambda\, K_{f',\pi}\norbra{\delta\tilde{\rho},(\idO-\widetilde{\Phi}_t\big|^{(f',f)}_\pi\,\Phi_t)(\delta\tilde{\rho})} < 0\,,
		\end{align}
		where the last step follows from the fact that $\lambda<0$ and $(\idO-\widetilde{\Phi}_t\big|^{(f',f)}_\pi\,\Phi_t)$ is a positive operator. This concludes the proof.
	\end{proof}
	
	As we mentioned, $\widetilde{\Phi}_t\big|^{(f',f)}_\pi$ might not even be a CP map for general evolutions: for example, by choosing $\Phi =\idO$, the requirement that $\widetilde{\idO}_{(f',f), \pi} = \J_{f'}\big|_{\pi}\J_f^{-1}\big|_{\pi}$ is CP is equivalent to constraining both $\J_{f'}\big|_{\pi}$ and $\J_f^{-1}\big|_{\pi}$ to be CP. The conditions that are needed to ensure this are deferred to Sec.~\ref{sec:CPfisher}, but we can already use Thm.~\ref{cf:thm:cpCharacterisation} to give the most general expression of a CP Petz recovery $\widetilde{\Phi}_t\big|^{(f',f)}_\pi$:
	\begin{align}
		\widetilde{\Phi}_t\big|^{(f',f)}_\pi = \int_{-\infty}^\infty \de\nu^+_{f'}(t) \int_{-\infty}^\infty\de \nu^-_f(s)\;\mathcal{V}_{\pi} \norbra{\tfrac{1}{2}-it}\circ\Phi_t^\dagger\circ\mathcal{V}_{\Phi(\pi)} \norbra{is-\tfrac{1}{2}}\,.\label{eq:95}
	\end{align}
	where $\de\nu^+_{f'}$ and $\de \nu^-_f$ are two symmetric probability distribution on the real line, and we introduced the CP map $\mathcal{V}_{\pi} (z)[A] := \pi^z\,A\, (\pi^z)^\dagger$. This expression is somehow reminiscent of the one found in~\cite{jungeUniversalRecoveryMaps2018a} for the universal recovery map. Then, it directly follows from Thm.~\ref{cf:theo:recovery_fisher} that:
	\begin{corollary}\label{cf:cor:cpRetrodiction}
		Suppose $f$ and $f'$ are chosen such that $\J_f^{-1}\big|_{\pi}$ and $\J_{f'}\big|_{\pi}$ are both CP. Then, there exists a channel $\widetilde{\Phi}_t\big|^{(f',f)}_\pi$ of the form in Eq.~\eqref{eq:95}  that perfectly retrieves the prior state (i.e., $\widetilde{\Phi}_t\big|^{(f',f)}_\pi\Phi_t(\pi) = \pi$), reduces to the Bayes' rule for classical dynamics, and such that Eq.~\eqref{eq:90} is satisfied.
	\end{corollary}
	The fact that $\widetilde{\Phi}_t\big|^{(f',f)}_\pi$ can be made CP means that there exists a physical procedure that an agent can carry out to retrodict the initial state of the system after the noisy operation $\Phi$ is applied. Moreover, Corollary~\ref{cf:cor:cpRetrodiction} gives a recipe to exploit non-Markovianity for error correction, showing that the backflow of information is indeed beneficial.
	
	\begin{additional-info}{Box 8. Proofs of the properties of the generalised Petz recovery~\hyperref[box:contractionCoefficients]{$\rightarrow$}}\label{box:proofGeneralisedPetz}
		We present here the derivations of the properties that generalised Petz recovery maps satisfy.
		
		First, regarding condition~\ref{item:reducesToBayes}, i.e., the fact that $\widetilde{\Phi}_t\big|^{(f',f)}_\pi$ reduces to Bayes rule for classical dynamics this can be seen as follows: a channel implementing a classical evolution takes the form $\Phi_t(\cdot) = \sum_{i,j}\, (\Phi_t)_{i,j}\,\ketbra{i}{j} \boldsymbol{\cdot}\ketbra{j}{i}$, where $\{\ket{i}\}$ are part of an orthonormal basis. Moreover, if we restrict to diagonal states, $\J_f^{-1}\big|_{\Phi_{t}(\pi)}(\cdot) =  \sum_{i}\,(\Phi_{t}(\pi)_i)^{-1} \ketbra{i}{i} \boldsymbol{\cdot}\ketbra{i}{i}$ and $\J_{f'}\big|_{\pi}(\cdot) =  \sum_{i}\,\pi_i \ketbra{i}{i} \boldsymbol{\cdot}\ketbra{i}{i}$ irrespective of $f$ and $f'$. Then, putting this elements together one obtains Eq.~\eqref{eq:classicalBayes}. This can be interpreted as the Bayes rule by thinking of $\pi_i$ as the probability of obtaining the microstate $i$, $(\Phi_t(\pi))_j$ as the probability to be in the microstate $j$ after the evolution $\Phi_t$ and, finally, $(\Phi_t)_{j,i}$ as the conditional probability of the transition $i\rightarrow j$.
		
		Condition~\ref{item:tracePreserving} tells us that $\widetilde{\Phi}_t\big|^{(f',f)}_\pi$ is trace preserving if $\pi$ is full-rank (which we always assume to be the case in this work). This can be readily verified by applying its adjoint on the identity matrix:
		\begin{align}
			(\widetilde{\Phi}_t\big|^{(f',f)}_\pi)^\dagger [\id] = \J_f^{-1}\big|_{\Phi_{t}(\pi)}\,\Phi_t\;\J_{f'}\big|_{\pi}[\id] =  \J_f^{-1}\big|_{\Phi_{t}(\pi)}  [\Phi_t(\pi)] = \id\,,
		\end{align}
		where we used the fact that for every map $\Phi$ the condition of being trace preserving can be written as $\Phi^\dagger(\id) = \id$.
		
		Let us now pass to condition~\ref{item:positiveSelfAdjoint}. This gives a first characterisation of the spectral properties of $\widetilde{\Phi}_t\big|^{(f',f)}_\pi\Phi_t$. First, it should be noticed that this map is self-adjoint with respect to $\J_{f'}^{-1}\big|_{\pi}$, as it can be readily verified: 
		\begin{align}
			&\Tr{A \,\J_{f'}^{-1}\big|_{\pi} [\widetilde{\Phi}_t\big|^{(f',f)}_\pi\,\Phi_t(B)]} = \Tr{\Phi_t(A) \,\J_f^{-1}\big|_{\Phi_t(\pi)} [\Phi_t(B)]}=\\ &\qquad\qquad=\Tr{\norbra{\J_{f'}\big|_{\pi} \Phi^\dagger\J_f^{-1}\big|_{\Phi_t(\pi)}}\Phi_t(A) \,\J_{f'}^{-1}\big|_{\pi} [B]} 
			=\Tr{\widetilde{\Phi}_t\big|^{(f',f)}_\pi\,\Phi_t(A) \,\J_{f'}^{-1}\big|_{\pi} [B]}\,,\label{cf:eq:selfAdjFB}
		\end{align}
		where we repeatedly used the self-adjointness of the Fisher superoperators to move them from right to left. This shows that $\widetilde{\Phi}_t\big|^{(f',f)}_\pi\,\Phi_t$ is self-adjoint, so its spectrum is real. Moreover, it is also positive since it can be rewritten as:
		\begin{align}
			\widetilde{\Phi}_t\big|^{(f',f)}_\pi\,\Phi_t = \J_{f'}^{\frac{1}{2}}\big|_{\pi}\norbra{ (\J_f^{-\frac{1}{2}}\big|_{\Phi_t(\pi)}\Phi_t\,\J_{f'}^{\frac{1}{2}}\big|_{\pi})^\dagger(\J_f^{-\frac{1}{2}}\big|_{\Phi_t(\pi)}\Phi_t\,\J_{f'}^{\frac{1}{2}}\big|_{\pi}}\J_{f'}^{-\frac{1}{2}}\big|_{\pi}\,.
		\end{align}
		This proves the claim, as similarity transformations preserve the spectrum.
		
		It is also easy to see that $\widetilde{\Phi}_t\big|^{(f',f)}_\pi\,\Phi_t$ retrieves the prior (condition~\ref{item:retrievesThePrior}):
		\begin{align}\label{cf:eq:priorRecover}
			\widetilde{\Phi}_t\big|^{(f',f)}_\pi\,\Phi_t (\pi) &=  \J_{f'}\big|_{\pi}\,\Phi_t^\dagger\; \J_f^{-1}\big|_{\Phi_{t}(\pi)}[\Phi_t (\pi)]  = \J_{f'}\big|_{\pi}[\Phi_t^\dagger(\id) ]= \J_{f'}\big|_{\pi}[\id] = \pi\,.
		\end{align}
		This equation shows that the evolution of the state $\pi$ can be completely undone by applying $\widetilde{\Phi}_t\big|^{(f',f)}_\pi$. Moreover, the spectrum of $\widetilde{\Phi}_t\big|^{(f',f)}_\pi$ contains $1$, and the associated eigenoperator is $\pi$.
		
		We can now show that, with some restrictions on $f$ and $f'$, condition~\ref{item:spectrumInZeroOne} is satisfied, namely the fact that the spectrum is actually contained in the interval $[0,1]$. Indeed, this follows from the chain of inequalities:
		\begin{align}\label{cf:eq:1209}
			\frac{\Tr{A \, \J_{f'}^{-1}\big|_{\pi}[\widetilde{\Phi}_t\big|^{(f',f)}_\pi\,\Phi_t(A)]}}{\Tr{A \, \J_{f'}^{-1}\big|_{\pi}[A]}} = \frac{\Tr{\Phi_t(A) \,\J_f^{-1}\big|_{\Phi_t(\pi)} [\Phi_t(A)]}}{\Tr{A \, \J_{f'}^{-1}\big|_{\pi}[A]}}\, \leq\, \frac{\Tr{A \,\J_f^{-1}\big|_{\pi} [A]}}{\Tr{A \, \J_{f'}^{-1}\big|_{\pi}[A]}}\, \leq\, \frac{\Tr{A \,\J_{f'}^{-1}\big|_{\pi} [A]}}{\Tr{A \, \J_{f'}^{-1}\big|_{\pi}[A]}}=1\,,
		\end{align}
		where the first inequality follows from the contractivity of the Fisher information, while the second only holds if $\J_{f}^{-1}\big|_{\pi}\leq\J_{f'}^{-1}\big|_{\pi}$, corresponding to the case in which $f'(x)\leq f(x)$ for every $x\in\RR^+$ (see Eq.~\eqref{cf:eq:orderJ}).
		
		Define now the map $P_{(f',f),\pi}(\Phi_t):=\widetilde{\Phi}_t\big|^{(f',f)}_\pi$. Condition~\ref{item:involutivity} tells us that $P_{(f,f'),\Phi_t(\pi)}(P_{(f',f),\pi}(\Phi_t)) = \Phi_t$. This can be verified by a direct computation:
		\begin{align}
			P_{(f,f'),\Phi_t(\pi)}(P_{(f',f),\pi}(\Phi_t)) &= \J_f\big|_{\Phi_t(\pi)}(\widetilde{\Phi}_t\big|^{(f',f)}_\pi)^\dagger \J_{f'}^{-1}\big|_{\widetilde{\Phi}_t\big|^{(f',f)}_\pi\Phi_{t}(\pi)} =\\
			&=\J_f\big|_{\Phi_t(\pi)}\norbra{\J_f^{-1}\big|_{\Phi_{t}(\pi)}\,\Phi_t\;\J_{f'}\big|_{\pi}} \J_{f'}^{-1}\big|_{\pi} =\Phi_t\,.
		\end{align}
		Interestingly, this make the transformation $P_{(f,f),\pi}$ involutive, since it can be reversed by just changing the base-point, i.e., through $P_{(f,f),\Phi_t(\pi)}$. This is one of the key properties of the classical Bayes retrieval (as it was discussed in~\cite{surace2022state}). Moreover, if one also requires $P_{(f,f),\pi}(\Phi_t)$ to be CP for general maps, this constrains the defining function to $f_{SQ}(x)=\sqrt{x}$, as it is the only case in which both $\J_{f_{SQ}}\big|_{\pi}$ and $\J_{f_{SQ}}^{-1}\big|_{\pi}$ are CP for any state. This gives a way to single out the usual Petz recovery map~\cite{petz1986sufficient}.
		
		Finally, condition~\ref{item:markovianity} explains the relation between the generalised Petz recovery and Markovian evolutions. In particular, for divisible dynamics one has:
		\begin{align}
			P_{(f',f),\pi}(\Phi_{t})  &= \J_{f'}\big|_{\pi}\,(\Phi_{t,s}\circ\Phi_s)^\dagger \;\J_f^{-1}\big|_{\Phi_{t}(\pi)} =\norbra{\J_{f'}\big|_{\pi} \Phi_s^\dagger\, \J_{f''}^{-1}\big|_{\Phi_{s}(\pi)}}\norbra{\J_{f''}\big|_{\Phi_{s}(\pi)}\Phi_{t,s}^\dagger\,\J_f^{-1}\big|_{\Phi_{t,s}\Phi_{s}(\pi)}}=\\
			&=P_{(f',f''),\pi}(\Phi_{s})\circ P_{(f'',f),\Phi_s(\pi)}(\Phi_{t,s})\,,
		\end{align}
		proving the claim. It should be noticed that in the case of $f\equiv f'$ the composition becomes $P_{(f,f''),\pi}(\Phi_{s})\circ P_{(f'',f),\Phi_s(\pi)}(\Phi_{t,s})$ for any $f''$, and in  particular if one choses $f'' = f$ it shows that $P_{(f,f),\pi}$ is compatible with the structure of divisible semigroups in the sense that $P_{(f,f),\pi}(\Phi_{t}) =P_{(f,f),\pi}(\Phi_{s})\circ P_{(f,f),\Phi_s(\pi)}(\Phi_{t,s})$.
	\end{additional-info}
	
	\begin{additional-info}{Box 9. Generalised universal recovery maps ~\hyperref[box:proofGeneralisedPetz]{$\leftarrow$},\hyperref[detailedBalance]{$\rightarrow$}}\label{box:contractionCoefficients}
		We discuss here a possible generalisation of the result in~\cite{jungeUniversalRecoveryMaps2018a}, namely the fact that the contraction of the relative entropy (defined in Eq.~\eqref{cf:eq:relativeEntropy}) can be bounded as:
		\begin{align}
			S(\rho||\sigma)-S(\Phi(\rho)||\Phi(\sigma)) \geq-\, \int_{-\infty}^\infty\dt\;\beta(t) \log F(\rho,\widetilde{\Phi}_{P,\sigma,t} \Phi (\rho))\,,\label{eq:106}
		\end{align}
		where $\beta(t)$ is a symmetric probability distribution (we refer to~\cite{jungeUniversalRecoveryMaps2018a} for its particular expression), $\widetilde{\Phi}_{P,\sigma,t} := \mathcal{V}_{\sigma} \norbra{\frac{1}{2}-it}\circ\Phi^\dagger\circ\mathcal{V}_{\Phi(\sigma)} \norbra{it-\frac{1}{2}}$ is called the rotated Petz recovery map,  and we introduced the fidelity $F(\rho,\sigma) := \Tr{\sqrt{\sqrt{\rho}\,\sigma\sqrt{\rho}}}$. The interest in this bound is twofold: on the one hand it quantifies the minimum amount of information lost when one applies the channel $\Phi$; on the other, it implies that if $S(\rho||\sigma)=S(\Phi(\rho)||\Phi(\sigma))$, then $\widetilde{\Phi}_{P,\sigma,t}\Phi(\rho)=\rho$, i.e., $\widetilde{\Phi}_{P,\sigma,t}$ perfectly recovers every state for which the dynamics does not decrease the relative entropy (with respect to $\sigma$). This condition is what one refers to as universality in this context. This can be proved as follows: first, it should be noticed that for any two states the fidelity $F(\rho,\sigma)$ is in $[0,1]$, so its logarithm is negative definite. Hence, in this case we have the chain of inequalities:
		\begin{align}
			0 = S(\rho||\sigma)-S(\Phi(\rho)||\Phi(\sigma)) \geq-\, \int_{-\infty}^\infty\dt\;\beta(t) \log F(\rho,\widetilde{\Phi}_{P,\sigma,t} \Phi (\rho))\geq 0\,,\label{eq:107}
		\end{align}
		which implies that $F(\rho,\widetilde{\Phi}_{P,\sigma,t} \Phi (\rho)) = 1$  for all $t$. This is only possible if $\widetilde{\Phi}_{P,\sigma,t} \Phi (\rho) = \rho$, proving the claim.
		
		One possibility of generalising Eq.~\eqref{eq:106} is to try to find a similar result holding for the general family of contrast functions defined in Eq.~\eqref{cf:eq:HGexpressionApp}. Unfortunately, to the best of our knowledge this problem is still open. Nevertheless, if one focuses on $\chi^2_f$-divergences alone (see the definition in Eq.~\eqref{eq:chi2Divergences}), we can prove that:
		\begin{theorem}\label{theo:chi2recoverybound}
			Given a CPTP map $\Phi$ and any two standard monotone functions $f$ and $f'$ such that $f'(x)\leq f(x)$ for every $x\in\RR^+$, the following bound holds: 
			\begin{align}
				\chi^2_{f'}(\rho||\sigma) -\chi^2_f(\Phi(\rho)||\Phi(\sigma) ) &= \Tr{\Delta\,\J_{f'}^{-1}\big|_\rho\sqrbra{(\idO-\widetilde{\Phi}_{(f',f),\rho}\,\Phi)(\Delta)}}\geq\label{eq:108}\\
				&\geq \mathcal{F}_{f',\rho}\norbra{(\idO-\widetilde{\Phi}_{(f',f),\rho}\,\Phi)(\Delta)} \geq \left| \left|(\idO-\widetilde{\Phi}_{(f',f),\rho}\,\Phi)(\Delta)\right| \right|_1^2\,,\label{eq:109}
			\end{align}
			where we used the abbreviation $\Delta = (\rho-\sigma)$ and introduced the trace norm $\left| \left|A\right| \right|_1 = \Tr{|A|}$.
		\end{theorem}
		There are a number of remarks to be made: first, it should be noticed that the relative entropy cannot be expressed in terms of $\chi^2_f$-divergences, so the result in Thm.~\ref{theo:chi2recoverybound} is not a direct generalisation of Eq.~\eqref{eq:106}. Nonetheless, it is a first step in this direction. Moreover, one should also notice the crucial role of the map $(\idO-\widetilde{\Phi}_{(f',f),\rho}\,\Phi)$: this was the key ingredient in proving Thm.~\ref{cf:theo:recovery_fisher} and it is now again at the centre of the above result. This further corroborates the interpretation of $\widetilde{\Phi}_{(f',f),\rho}$ as a retrieval map: indeed, the quality of the retrodiction can be assessed by comparing $\widetilde{\Phi}_{(f',f),\rho}\,\Phi$ (the map forward and backwards) with $\idO$ (not doing anything). The fact that the operator $(\idO-\widetilde{\Phi}_{(f',f),\rho}\,\Phi)$ naturally appears in many different constructions shows that this interpretation is well justified. Finally, we also point out that a similar result as Thm.~\ref{theo:chi2recoverybound} was found in~\cite{gao2023sufficient} (Lemma~4.9 therein).
		
		\begin{proof}
			The first equality in Eq.~\eqref{eq:108} directly follows from the definition of generalised Petz recovery map in Eq.~\eqref{cf:eq:generalisedPetzDef}:
			\begin{align}
				\chi^2_{f'}(\rho||\sigma) -\chi^2_f(\Phi(\rho)||\Phi(\sigma) )  &= \Tr{\Delta\, \J_{f'}^{-1}\big|_\rho[\Delta]} - \Tr{\Phi(\Delta)\, \J_{f}^{-1}\big|_{\Phi(\rho)}[\Phi(\Delta)]}  = \\
				&=\Tr{\Delta\, \J_{f'}^{-1}\big|_\rho[\Delta]} - \Tr{\Delta\, \J_{f'}^{-1}\big|_{\rho}[\widetilde{\Phi}_{(f',f),\rho}\Phi(\Delta)]} = \\
				&= \Tr{\Delta\,\J_{f'}^{-1}\big|_\rho\sqrbra{(\idO-\widetilde{\Phi}_{(f',f),\rho}\,\Phi)(\Delta)}}\,,
			\end{align}
			where we used the same definition of $\Delta$ as in the theorem. At this point we can use the property~\ref{item:spectrumInZeroOne} once more, which says that the spectrum of  $(\idO-\widetilde{\Phi}_{(f',f),\rho}\,\Phi)$ is in $[0,1]$ for $f$ and $f'$ as in the statement of the theorem. But from this it follows that $(\idO-\widetilde{\Phi}_{(f',f),\rho}\,\Phi)\geq (\idO-\widetilde{\Phi}_{(f',f),\rho}\,\Phi)^2$, which directly implies the first inequality in Eq.~\eqref{eq:109}. Finally, since for any $f$ it holds that $\mathcal{F}_{f,\rho}(\delta\rho)\geq \mathcal{F}_{B,\rho}(\delta\rho)\geq \left| \left|\delta\rho\right| \right|_1^2$ (see Sec.~\ref{Bures} where $\mathcal{F}_{B,\rho}(\delta\rho)$ is defined) the last inequality also follows. 
		\end{proof}
		
		We are now ready to generalise the result in Eq.~\eqref{eq:107}:
		\begin{corollary}\label{cor:eqChi2divergences}
			Suppose that $\chi^2_{f}(\rho||\sigma) = \chi^2_{f}(\Phi(\rho)||\Phi(\sigma) )$ for some standard monotone function $f$. Then, the map $\widetilde{\Phi}_{(f,f),\rho}$ satisfies:
			\begin{align}
				\widetilde{\Phi}_{(f,f),\rho}\,\Phi(\rho) = \rho \qquad\land\qquad\widetilde{\Phi}_{(f,f),\rho}\,\Phi(\sigma)=\sigma\,.\label{eq:113}
			\end{align}
			Moreover, if $\chi^2_{f}(\rho||\sigma) = \chi^2_{f}(\Phi(\rho)||\Phi(\sigma) )$ for $f_{SQ}(x) = \sqrt{x}$, then the equality holds also for all other $\chi^2_f$-divergences.
		\end{corollary}
		\begin{proof}
			From Thm.~\ref{theo:chi2recoverybound} it directly follows that:
			\begin{align}
				0=\chi^2_{f}(\rho||\sigma) -\chi^2_f(\Phi(\rho)||\Phi(\sigma) ) \geq \mathcal{F}_{f,\rho}\norbra{(\idO-\widetilde{\Phi}_{(f,f),\rho}\,\Phi)(\Delta)} \geq 0\,,
			\end{align}
			which implies $\mathcal{F}_{f,\rho}\norbra{(\idO-\widetilde{\Phi}_{(f,f),\rho}\,\Phi)(\Delta)}=0$. Since the Fisher information arises from a non-degenerate scalar product, this means that  $(\idO-\widetilde{\Phi}_{(f,f),\rho}\,\Phi)(\Delta) = 0$ $\implies$ $\widetilde{\Phi}_{(f,f),\rho}\,\Phi(\Delta) = \Delta$. It follows from property~\ref{item:retrievesThePrior} that $\widetilde{\Phi}_{(f,f),\rho}\,\Phi(\rho) = \rho$. Hence, the only way for $\widetilde{\Phi}_{(f,f),\rho}\,\Phi$ to retrieve $\Delta =(\rho-\sigma)$ is that $\widetilde{\Phi}_{(f,f),\rho}\,\Phi(\sigma) = \sigma$, proving the first part of the claim. 
			
			As it was mentioned in \hyperref[box:proofGeneralisedPetz]{Box 8}, the case $f_{SQ}(x) = \sqrt{x}$ is the only one in which $\widetilde{\Phi}_{(f_{SQ},f_{SQ}),\rho}$ is CP in general (in the following we will denote this quantity by $\widetilde{\Phi}_{P,\rho}$). Then, if $\chi^2_{f_{SQ}}(\rho||\sigma) = \chi^2_{f_{SQ}}(\Phi(\rho)||\Phi(\sigma) )$, this implies the existence of a channel recovering both $\rho$ and $\sigma$. Then, by using the contractivity of $\chi^2_f$-divergences it follows that:
			\begin{align}
				\chi^2_{f}(\rho||\sigma) \geq \chi^2_{f}(\Phi(\rho)||\Phi(\sigma) ) \geq \chi^2_{f}(\widetilde{\Phi}_{P,\rho}\Phi(\rho)||\widetilde{\Phi}_{P,\rho}\Phi(\sigma) ) = \chi^2_{f}(\rho||\sigma)\,.
			\end{align}
			This directly implies that $\chi^2_{f}(\rho||\sigma) = \chi^2_{f}(\Phi(\rho)||\Phi(\sigma) )$.
		\end{proof}
		The proof above shows that the exceptionality of the square root divergence derives from the fact that this is the only case  in which $\widetilde{\Phi}_{(f,f),\rho}$ is a channel. Moreover, allowing the two defining functions $f$ and $f'$ to be different does not help in general, as the following corollary shows:
		\begin{corollary}\label{cor:generalisedUniversalRecovery}
			Suppose that $\chi^2_{f}(\rho||\sigma) = \chi^2_{f}(\Phi(\rho)||\Phi(\sigma) )$ for some standard monotone function $f$. Then, choosing $f'\neq f$, the map $\widetilde{\Phi}_{(f',f),\rho}$ satisfies: 
			\begin{align}
				\widetilde{\Phi}_{(f',f),\rho}\,\Phi(\rho) = \rho \qquad\land\qquad\widetilde{\Phi}_{(f',f),\rho}\,\Phi(\sigma)=\sigma\,.\label{eq:116}
			\end{align}
			if and only if $[\rho,\sigma]= 0$. In this case, the equality for any $\chi^2_f$-divergence such that $\J_f^{-1}$ is CP implies the equality for all other $\chi^2_{f'}$-divergences.
		\end{corollary}
		\begin{proof}
			It should be noticed that the requirement $\widetilde{\Phi}_{(f',f),\rho}\,\Phi(\rho) = \rho$ directly follows from the property~\ref{item:retrievesThePrior}, so it can be imposed without problems. Let us now focus on the second part of Eq.~\eqref{eq:116}, namely:
			\begin{align}
				\widetilde{\Phi}_{(f',f),\rho}\,\Phi(\sigma) &=  \J_{f'}\big|_{\rho}\,\Phi^\dagger \;\J_f^{-1}\big|_{\Phi(\rho)}\,\Phi(\sigma) = \norbra{\J_{f'}\big|_{\rho}\J_{f}^{-1}\big|_{\rho}}\norbra{\J_{f}\big|_{\rho}\,\Phi^\dagger \;\J_f^{-1}\big|_{\Phi(\rho)}}\,\Phi(\sigma)=\\
				&=\norbra{\J_{f'}\big|_{\rho}\J_{f}^{-1}\big|_{\rho}}\widetilde{\Phi}_{(f,f),\rho}\,\Phi(\sigma) = \J_{f'}\big|_{\rho}\J_{f}^{-1}\big|_{\rho}[\sigma]\,,\label{eq:118}
			\end{align}
			where in the second line we used Corollary~\ref{cor:eqChi2divergences} to substitute $\widetilde{\Phi}_{(f,f),\rho}\,\Phi(\sigma) =\sigma$. Now, if $f\neq f'$, $\J_{f'}\big|_{\rho}\J_{f}^{-1}\big|_{\rho}[\sigma] =\sigma$ if and only if $[\rho,\sigma]= 0$, as it can be verified in coordinates (see Sec.~\ref{propertiesFisher}). 
			
			Regarding the second part of the Corollary suppose that there exists a standard monotone $f$ such that $\chi^2_{f}(\rho||\sigma) = \chi^2_{f}(\Phi(\rho)||\Phi(\sigma) )$ and $\J_f^{-1}$ is CP. Then, by choosing any $f'$ such that $\J_{f'}$ is CP, $\widetilde{\Phi}_{(f',f),\rho}$ is a quantum channel. Moreover, if $[\rho,\sigma]= 0$, it follows from Eq.~\eqref{eq:118} that $\widetilde{\Phi}_{(f',f),\rho}\,\Phi(\sigma)=\sigma$. Thus, for any  $\chi^2_{\tilde{f}}$-divergence it holds that:
			\begin{align}
				\chi^2_{\tilde{f}}(\rho||\sigma)\geq\chi^2_{\tilde{f}}(\Phi(\rho)||\Phi(\sigma) ) \geq \chi^2_{\tilde{f}}(\widetilde{\Phi}_{(f',f),\rho}\Phi(\rho)||\widetilde{\Phi}_{(f',f),\rho}\Phi(\sigma) )=\chi^2_{\tilde{f}}(\rho||\sigma)\,,
			\end{align}
			proving the claim. 
		\end{proof}
		In the result above it should be noticed that even if one needs the commutativity of $\rho$ and $\sigma$, their evolved versions $\Phi(\rho)$ and $\Phi(\sigma)$ can be general. Still, it should be noticed that if one requires $\widetilde{\Phi}_{(f',f),\rho}$ to be a channel, the only unconstrained recovery result of the form in Eq.~\eqref{eq:116} is given by the Petz recovery map, i.e., for $f(x) = f'(x) = \sqrt{x}$. 
		
		Finally, we provide yet another way of singling out the Petz recovery map from the generalised family $\widetilde{\Phi}_{(f',f),\rho}$. Indeed, this one corresponds to the best recovery in a specific sense:
		\begin{theorem}\label{thm:petzSup}
			Consider the family of positive superoperators $\widetilde{\Phi}_{(f',f),\rho}\Phi$ such that $\widetilde{\Phi}_{(f',f),\rho}$ is CP. One can introduce the partial order $\widetilde{\Phi}_{(f',f),\rho}\Phi\geq \widetilde{\Phi}_{(\tilde{f}',\tilde{f}),\rho}\Phi \iff(\widetilde{\Phi}_{(f',f),\rho}\Phi- \widetilde{\Phi}_{(\tilde{f}',\tilde{f}),\rho}\Phi)\geq0$. Then, there is a unique supremum given by:
			\begin{align}
				\sup_{f,f'} \, \widetilde{\Phi}_{(f',f),\rho}\Phi =  \widetilde{\Phi}_{P,\rho}\,\Phi\,.
			\end{align} 
		\end{theorem}
		This results means that among all possible recovery channels $\widetilde{\Phi}_{(f',f),\rho}$, the Petz is the one that maximises the spectrum of $\widetilde{\Phi}_{(f',f),\rho}\Phi $. 
		\begin{proof}
			As it was mentioned in \hyperref[box:proofGeneralisedPetz]{Box 8}, if $f(x)\leq f'(x)$ for every $x\in\RR^+$, this implies that $\J_{f}\big|_{\rho}\leq\J_{f'}\big|_{\rho}$ (see Eq.~\eqref{cf:eq:orderJ}). Moreover, a necessary condition for $\J_{f}\big|_{\rho}$ to be CP is that $f(x)\leq \sqrt{x}$, so one can maximise $\widetilde{\Phi}_{(f',f),\rho}$ by choosing the maximum $f'$:
			\begin{align}
				\sup_{f'} \, \widetilde{\Phi}_{(f',f),\rho}\Phi = \widetilde{\Phi}_{(f_{SQ},f),\rho}\Phi\,,
			\end{align}
			where we denote by $f_{SQ}(x)=\sqrt{x}$. On the other hand, for $\J_{f}^{-1}\big|_{\rho}\geq\J_{f'}^{-1}\big|_{\rho}$ if and only if $f(x)\leq f'(x)$ for every $x\in\RR^+$. Moreover, a necessary condition for $\J_{f}^{-1}\big|_{\rho}$ to be CP is that $f(x)\geq \sqrt{x}$. Hence, the maximum $\widetilde{\Phi}_{(f',f),\rho}$ in this case is reached for the minimum function $f$, that is once again for $f_{SQ}$. Finally, the fact that the two supremum are independent of each other proves the claim.
		\end{proof}
	\end{additional-info}

	\subsection{Fisher information and detailed balance}\label{detailedBalance}
	
	In this section we explain how the notion of detailed balanced evolutions can be given in terms of self-adjointness with respect to the Fisher information scalar product. To this end, it is useful to start with the study of the classical case to avoid the complications arising from the variety of different quantum Fisher information.
	
	\paragraph*{{Classical case.}}
	As discussed in~\hyperref[box:markovianEvolutionsClassical]{Box 6}, the dynamics of a classical divisible evolution~$\Phi_t$ is described by a stochastic matrix, which can be equivalently characterised in terms of its rate matrices $R_s$ ($\forall s\leq t$, see Eq.~\eqref{eq:rateMatrixClassical}), which induce the dynamics: 
	\begin{align}
		\frac{\de}{\dt}\,\Phi_t(\rho) =R_t (\rho)\,,
	\end{align}
	where $\rho$ is in this case a diagonal matrix.
	As it was mentioned in~\hyperref[box:markovianEvolutionsClassical]{Box 6}, rate matrices can be generically decomposed as:
	\begin{align}
		R_t=\sum_{i\neq j} \; a^{(t)}_{i\leftarrow j}\left( \ketbra{i}{j}-\ketbra{j}{j} \right)\,,
	\end{align}
	where $a^{(t)}_{i\leftarrow j}$ are real coefficients that are non-negative for Markovian evolutions. In this context, one can formulate the condition of detailed balanced evolution (with respect to some state $\pi$) in terms of the rates alone:
	\begin{align}
		a^{(t)}_{i\leftarrow j}\,\pi_j = a^{(t)}_{j\leftarrow i}\,\pi_i \,, \label{cf:eq:classicalDB}
	\end{align}
	where the equation above holds for every $i$ and $j$. This condition directly implies that $\pi$ is a steady state, as it can be readily verified by a straightforward computation. 
	Indeed, detailed balance corresponds to a stronger notion of equilibration: not only the dynamics has $\pi$ as a fixed points, but it is also time symmetric at equilibrium. In fact, if one interprets the rates $a^{(t)}_{j\leftarrow i}$ as the probability per unit of time of the transition $j\leftarrow i$, Eq.~\eqref{cf:eq:classicalDB} can be read as the condition that the probability of observing the transition $j\leftarrow i$ is equal to the one for the reverse transition $j\rightarrow i$, when the system is at equilibrium. For this reason, detailed balance encodes the request of microscopic reversibility of the general dynamics, i.e., the fact that at a molecular level the equations of motion are time symmetric.
	
	We show now that Eq.~\eqref{cf:eq:classicalDB} can be naturally formulated in terms of the Fisher information scalar product. In analogy with the quantum case (Eq.~\eqref{cf:eq:monotoneMetrics}), we use the notation for the scalar product:
	\begin{align}
		K_\pi(\delta\rho,\delta\sigma) := \Tr{\delta\rho\,\cJ_\pi^{-1}[\delta\sigma]} =\Tr{\delta\rho \,\delta\sigma\,\pi^{-1}}\,,
	\end{align}
	where we implicitly defined $\cJ_\pi$ to be the component-wise multiplication by $\pi$, and all the operators involved commute. This scalar product naturally emerges when one is considering variations of states (as proved in Thm.~\ref{cf:thm:Ruskai}), thus from a differential geometric point of view, the two vectors $\delta\rho$ and $\delta\sigma$ should be elements of the tangent space of the state space (i.e., Hermitian, traceless operators and, in this case, diagonal). 
	
	Moreover, one can also interpret $\cJ_\pi$ as the scalar product on the cotangent space, i.e., the space of observables: indeed, a metric on the tangent space naturally induces one on its dual by taking the pointwise matrix inverse~\cite{oneillSemiRiemannianGeometryApplications1983}. Hence, one can also define the Fisher scalar product on the space of observables as:
	\begin{align}
		\label{cf:eq:classicalObsScalProd}
		{K}^{o}_\pi(A,B) := \Tr{A\,\cJ_\pi[B]} =\Tr{A \,B\,\pi}\,,
	\end{align}
	where in this case $A$ and $B$ are not required to be traceless. It should be noticed that quantities analogous to ${K}^o_\pi(A,B)$ naturally emerge in statistical mechanics and in linear response theory when studying two-point correlation functions.
	
	Since the definition of the adjoint of a superoperator $\Phi$ depends on the underlying scalar product, it is useful to present its expression when considering $K_\pi$. Then, in this context it satisfies the property:
	\begin{align}
		\nonumber
		K_\pi(\delta\rho,\Phi(\delta\sigma)) &= \Tr{\delta\rho\,\cJ_\pi^{-1}[\Phi(\delta\sigma)]} = \Tr{\cJ_\pi^{-1}[\delta\rho]\Phi(\delta\sigma)} =\Tr{\Phi^\dagger\circ\cJ_\pi^{-1}[\delta\rho]\,\delta\sigma} =\\
		&= \Tr{\norbra{\cJ_\pi\,\Phi^\dagger\,\cJ_\pi^{-1}}(\delta\rho)\,\cJ_\pi^{-1}[\delta\sigma]}= K_\pi (\widetilde{\Phi}(\delta\rho),\delta\sigma)\,,\label{cf:eq:242}
	\end{align}
	where we implicitly defined $\widetilde{\Phi}:= \cJ_\pi\,\Phi^\dagger\,\cJ_\pi^{-1}$. Since the equation above holds for any $\delta\rho$ and $\delta\sigma$ in the tangent space, we can identify  $\widetilde{\Phi}$ with the adjoint of $\Phi$ with respect to the Fisher information metric. Then, self-adjointness in this context takes the form $\widetilde{\Phi}=\Phi$, which can be rewritten as $\Phi\circ\cJ_\pi=\cJ_\pi\circ\Phi^\dagger$. In coordinates this condition reads:
	\begin{align}
		\Phi_{i,j}\,\pi_j ={\Phi}_{j,i}\,\pi_i\,,\label{eq:128x}
	\end{align}
	which has a striking similarity with Eq.~\eqref{cf:eq:classicalDB}. Indeed, it follows from Eq.~\eqref{eq:128x} that the requirement that $R_t$ is detailed balanced exactly means that $\widetilde{R}_t=R_t$, i.e., the generator of the dynamics is self-adjoint with respect to the Fisher information scalar product.
	
	This result was derived in the Schrödinger picture, meaning that the states are the only evolving objects, while observables are static quantities. The dual situation, dubbed Heisenberg picture, is the one in which states are fixed in time, while the whole dynamics is relegated to observables. In this case it is well known that the generator of the dynamics is given by $R_t^\dagger$. Moreover, as it was argued above, the natural Fisher scalar product is the one given by ${K}^o_\pi$. We denote the adjoint with respect to this scalar product by $\widetilde{\Phi}^o := \cJ_\pi^{-1}\,\Phi^\dagger\,\cJ_\pi$, where this condition can be verified carrying out calculations completely analogous to the one that led to Eq.~\eqref{cf:eq:242}. It is easy to verify that:
	\begin{align}
		\widetilde{\Phi}=\Phi\qquad\iff\qquad (\widetilde{\Phi^\dagger})^o = \Phi^\dagger\,,
	\end{align} 
	which implies that $R_t$ is $K_\pi$-self-adjoint if and only if $R_t^\dagger$ is $K^o_\pi$-self-adjoint. This shows that one can formulate the condition of being detailed balanced both in the Schrödinger and in the Heisenberg picture, resorting only to the use of Fisher scalar products.
	
	Putting these results together we obtain the theorem:
	\begin{theorem}\label{cf:thm:classicalDB}
		The following conditions are equivalent in the classical case:
		\begin{enumerate}
			\item the dynamics is detailed balanced, i.e., the rate matrix is characterised by coefficients satisfying 
			\begin{align}
				a^{(t)}_{i\leftarrow j}\,\pi_j = a^{(t)}_{j\leftarrow i}\,\pi_i\,;\label{cf:eq:classDBrates}
			\end{align}
			\item the rate matrix $R_t$ is self-adjoint with respect to the Fisher scalar product:
			\begin{align}
				\widetilde{R}_t=R_t\,;
			\end{align}
			\item the rate matrix in the Heisenberg picture (i.e., $R_t^\dagger$) is self-adjoint with respect to the dual Fisher metric:
			\begin{align}
				(\widetilde{R_t^\dagger })^o= R_t^\dagger\,.
			\end{align}
		\end{enumerate}
	\end{theorem}
	
	This characterisation shows how the condition of being detailed balanced directly corresponds to the self-adjointness of the generator of the dynamics with respect to a properly defined scalar products. With this hindsight in mind, we can now pass to the quantum regime.
	
	\paragraph*{{Quantum case.}}
	Consider now the dynamics induced by the Lindbladian operator in the form in Eq.~\eqref{cf:eq:lindDiagonal}, i.e.:
	\begin{align}
		\lind[\rho] = -i[H,\rho] + \sum_{\alpha}^{d^2} \;\lambda_\alpha\,\norbra{A_\alpha \,\rho\, A_\alpha^\dagger - \frac{1}{2}\{A_\alpha^\dagger \,A_\alpha, \rho\}}\,,
	\end{align}
	where we dropped the time dependence to simplify the notation. Since the Fisher information is invariant under the action of a purely Hamiltonian dynamics, whereas it is contracting otherwise, we split the two terms in the Lindbladian by introducing the notation $\mathcal{U}(\rho) := -i[H,\rho]$, and we call \emph{dissipator} the difference $\mathcal{L}_\mathcal{D} := \lind-\mathcal{U}$. It should be noticed that $\mathcal{U}$ is skew-Hermitian with respect to the Hilbert-Schmidt scalar product, meaning that:
	\begin{align}
		\Tr{A \, \mathcal{U}(B)} = - \Tr{\mathcal{U}(A)\, B}\,.
	\end{align} 
	Given the structure of the Lindbladian above, we can introduce the notion of quantum detailed balance. Historically, one of the first formalisations of this notion was provided by Alicki in~\cite{alicki1976detailed}, and it is based on the following scalar product on the space of observables:
	\begin{align}
		{K}^{o}_\pi(A,B) := \Tr{A B\,\pi}\,,\label{cf:eq:quantumObsScalProd}
	\end{align}
	where in this case $A$, $B$ and $\pi$ are not required to commute in general (unlike in Eq.~\eqref{cf:eq:classicalObsScalProd}). Similarly to the case of classical systems, this scalar product is quite natural as it is related to two-points correlation functions, but it should also be kept in mind that it is not part of the Fisher family. Using the same notation as in the classical case, we denote  by $\widetilde{\Phi}^o$ the adjoint of the map $\Phi$ with respect to the scalar product in Eq.~\eqref{cf:eq:quantumObsScalProd}. Then, the definition proposed by Alicki reads:
	\begin{definition}[Heisenberg picture~\cite{alicki1976detailed}]\label{cf:def:alicki}
		The dynamics generated in the Heisenberg picture by the operator $\mathcal{L}^\dagger$ is detailed balanced if the three conditions are satisfied:
		\begin{enumerate}
			\item $\mathcal{L}^\dagger$ is normal with respect to the scalar product ${K}^{o}_\pi$:
			\begin{align}
				[\mathcal{L}^\dagger,(\widetilde{\mathcal{L}^\dagger})^o]=0\,;
			\end{align} 
			\item the commutator $\mathcal{U}$ is skew-Hermitian with respect to ${K}^{o}_\pi$:
			\begin{align}
				\widetilde{\mathcal{U}}^o = -\mathcal{U}\,;
			\end{align}
			\item the dissipator $\lind_\mathcal{D}^\dagger$ is self-adjoint with respect to ${K}^{o}_\pi$: 
			\begin{align}
				(\widetilde{\lind_\mathcal{D}^\dagger})^o=\lind_\mathcal{D}^\dagger \,.
			\end{align}
		\end{enumerate}
	\end{definition}
	Interestingly, from this definition one can deduce a structural characterisation of detailed balanced Lindbladians (see~\cite{alicki1976detailed}), which explicitly reads:
	\begin{definition}[Structural definition]
		\label{cf:def:breuer}
		The dynamics generated by the Lindbladian operator $\mathcal{L}$ satisfies detailed balance if its diagonal form can be written as:
		\begin{align}
			\lind(\rho)=-i[H,\rho]+\sum_{\omega, i} \;\lambda_{i}^\omega \left(A_i^{\omega}\,\rho \,(A_i^{\omega})^\dagger -\frac{1}{2}\{(A_i^{\omega})^\dagger A_i^{\omega},\rho \} \right)\,,\label{cf:eq:detailedBalancedLindbladian}
		\end{align}
		and the following conditions hold:
		\begin{enumerate}
			\item $[H,\pi]=0$;
			\item $(A_i^{\omega})^\dagger = A_i^{-\omega}$;
			\item $\pi \,A_i^{\omega}\, \pi^{-1}=e^{\omega}\, A_i^{\omega}$;\label{cf:item:cond3}
			\item $\lambda_{i}^\omega=e^{\omega}\,\lambda_{i}^{-\omega}$.\label{cf:item:DBrates}
		\end{enumerate} 
	\end{definition}
	In the literature one usually finds Def.~\ref{cf:def:breuer} as the usual definition of detailed balance, as it mirrors the same structural properties we saw for classical systems~\cite{breuer2002theory}. Indeed, from the condition~\ref{cf:item:cond3} we know that $A_i^{\omega}$ are eigenoperators of the modular operator $\LL_\pi\RR_\pi^{-1}$, with eigenvalue $e^\omega$. On the other hand,  all the eigenvalues of $\LL_\pi\RR_\pi^{-1}$ are of the form $\pi_i/\pi_j$,  so the only values of $\omega$ that are allowed are the ones satisfying the constraint $e^\omega = \pi_i/\pi_j$ for some $i$ and $j$. Then, by substituting this expression into condition~\ref{cf:item:DBrates}, one obtains the analogous characterisation at the level of the rates we found for classical detailed balanced dynamics (see Eq.~\eqref{cf:eq:classDBrates}).
	
	Despite this positive result, it should be noticed that the choice of the scalar product in Eq.~\eqref{cf:eq:quantumObsScalProd} is somehow arbitrary, as in the passage from commuting observables to the non-commuting case there are many different possible orderings that can be used to extend the multiplication operator $\cJ_\pi$.  Moreover, it should be noticed that the scalar product ${K}^{o}_\pi$ in Eq.~\eqref{cf:eq:quantumObsScalProd} is not monotone under CP-maps, a property that its classical counterpart had. For these reasons, in the following we show how one can define detailed balance through the help of quantum Fisher information metrics.
	
	A first possible definition is given by imposing $\widetilde \Phi_{f} = \Phi$ for some Fisher scalar product $K_{f,\pi}$ (see Eq.~\eqref{cf:eq:monotoneMetrics}), where we simplified the notation from the one used in the generalised Petz map\footnote{The reason why we do not consider the more general definition $\widetilde{\Phi}^{(f',f)}$ is that the adjoint operation should be involutive in finite dimensions, that is the adjoint of the adjoint should give back the original map. As explained in \hyperref[box:proofGeneralisedPetz]{Box 8} this is only true when $f'\equiv f$.} defined in Eq.~\eqref{cf:eq:generalisedPetzDef}, i.e., $\widetilde \Phi_{f}  :=\J_f\big|_{\pi}\,\Phi_t^\dagger\, \J_f^{-1}\big|_{\pi}$. As it could be expected, the structure induced is much richer in this case. In fact, one can show that different $K_{f,\pi}$ induce inequivalent notions of detailed balance, by constructing dynamics that are detailed balanced with respect to one $f$ but not with respect to others~\cite{temmeH2divergenceMixingTimes2010}. Moreover, these conditions are all weaker than the one provided by Def.~\ref{cf:def:alicki}. This means that dynamics that are detailed balanced with respect to $K_\pi^o$ are also detailed balanced with respect with $K_{f,\pi}$, but not the other way round.
	Since in principle there is no preferred definition of quantum Fisher information, we choose to impose the condition of detailed balance for all the $K_{f,\pi}$ at the same time: 
	\begin{definition}[Schroedinger picture]
		\label{cf:def:AS}
		The dynamics generated by the Lindbladian operator $\mathcal{L}$ is detailed balance if for every standard monotone function $f$ the following holds:
		\begin{enumerate}
			\item $\mathcal{L}$ is normal with respect to all the scalar products $K_{f,\pi}$:
			\begin{align}
				[\mathcal{L},\widetilde{\mathcal{L}}_f]=0\,;
			\end{align} 
			\item the commutator $\mathcal{U}$ is skew-Hermitian with respect to $K_{f,\pi}$:
			\begin{align}
				\widetilde{\mathcal{U}}_f = -\mathcal{U}\,;
			\end{align}
			\item the dissipator $\lind_\mathcal{D}$ is self-adjoint with respect to $K_{f,\pi}$: 
			\begin{align}
				(\widetilde{\lind_\mathcal{D}})_f=\lind_\mathcal{D}\,.
			\end{align}
		\end{enumerate}
	\end{definition}
	
	By this point we have three different definitions of what it means for a Markovian dynamics to be detailed balance. Then, we can wrap everything together, providing a characterisation of their interdependency:
	\begin{theorem}\label{cf:thm:eqDB}
		The following conditions are equivalent:
		\begin{enumerate}
			\item  the generator of the dynamics in the Heisenberg picture $\lind^\dagger$ satisfies the adjointness relations in Def.~\ref{cf:def:alicki};
			\item the Lindbladian $\lind$ satisfies the structural characterisation in Def.~\ref{cf:def:breuer}.
			\setcounter{counterN}{\value{enumi}}
		\end{enumerate}
		These conditions imply the condition:
		\begin{enumerate}
			\setcounter{enumi}{\value{counterN}}
			\item the generator of the dynamics in the Schroedinger picture $\lind$ satisfies the adjointness relations in Def.~\ref{cf:def:AS}.
		\end{enumerate}
		Moreover, if the Hamiltonian $H$ is non-degenerate the three conditions are equivalent.
	\end{theorem}
	
	The proof of this result is provided in App.~\ref{app:sec:DB}. This shows that the definition based on Fisher scalar products is weaker (i.e., it includes a larger set of dynamics) even when taking into consideration all the possible scalar products at the same time. This should be contrasted with the definition by Alicki, in which a single scalar product is used. The difference between the two arises in the way in which coherences in the eigenbasis of $\pi$ are handled. Still, it should be noticed that in both cases the evolution induced by the unitary part decouples from the dissipative dynamics. In fact, thanks to the normality of the generator, one has:
	\begin{align}
		[\mathcal{U} + \lind_\mathcal{D},-\mathcal{U} + \lind_\mathcal{D}] = 0\qquad\implies\qquad[\mathcal{U} , \lind_\mathcal{D}] = 0\,.
	\end{align}
	This generic property can be used to further constrain $\lind_\mathcal{D}$ in the case in which $H$ is non-degenerate, allowing for the identification of Def.~\ref{cf:def:alicki} and Def.~\ref{cf:def:AS} in this case.
	
	If $H$ is degenerate, on the other hand, we can provide the following structural characterisation of dissipators $\lind_\mathcal{D}$ that are detailed balanced according to Def.~\ref{cf:def:AS}:
	\begin{align}\label{cf:eq:structAS}
		\lind_\mathcal{D}[\rho] = \sum_{\omega, i} \;\norbra{\lambda_{i}^\omega \left(A_i^{\omega}\,\rho \,(A_i^{\omega})^\dagger -\frac{1}{2}\{(A_i^{\omega})^\dagger A_i^{\omega},\rho \} \right)+\mu^\omega_i\,B_i^\omega \,\rho^T (B_i^\omega )^\dagger}\,,
	\end{align}
	where the following conditions are satisfied:
	\begin{enumerate}
		\item $(A_i^{\omega})^\dagger = A_i^{-\omega}$ and $(B_i^{\omega})^\dagger = B_i^{-\omega}$;
		\item $\pi \,A_i^{\omega}\, \pi^{-1}=e^{\omega}\, A_i^{\omega}$ and $\pi \,B_i^{\omega}\, \pi^{-1}=e^{\omega}\, B_i^{\omega}$;
		\item $\lambda_{i}^\omega=e^{\omega}\,\lambda_{i}^{-\omega}$ and $\mu_{i}^\omega=e^{\omega}\,\mu_{i}^{-\omega}$;
		\item $\lambda_i^\omega \geq 0 $ and $\sum_i \mu_{i}^\omega = 0$ $\forall\omega$. \label{cf:item:cond44}
	\end{enumerate} 
	We defer the proof to App.~\ref{appM:sec:structuralCh}. It should be noticed that the characterisation in Eq.~\eqref{cf:eq:structAS} clearly generalises Def.~\ref{cf:def:breuer}, as the first can be considered a restriction of the latter to the case in which $\mu_i^\omega= 0$ for all $i$ and $\omega$. Moreover, it should also be pointed out that, thanks to the second part of condition~\ref{cf:item:cond44},  the extra term only acts on off-diagonal elements, i.e, it only affects the dynamics of the coherences.
	
	\section{Mathematical properties of  quantum Fisher information}\label{propertiesFisher}
	
	The aim of this section is to explore the main properties of the superoperator $\J_f\big|_\pi $, whose definition we repeat here for convenience:
	\begin{align}
		\J_f\big|_\pi := \RR_\pi \,f\norbraB{\LL_\pi\RR_\pi^{-1}}\,.\label{145}
	\end{align}
	Before doing so, however, it is useful to study the properties of the set of standard monotone functions, which we will denote by $\mathcal{M}$. 
	
	\subsection{The set of standard monotone functions}\label{sec:standardMonotonesProperties}
	As it was discussed in Sec.~\ref{contrastFunctionsToFisher},  every $f\in \mathcal{M}$ satisfies the following three properties:
	\begin{enumerate}
		\item is matrix monotone;
		\item is Fisher adjusted, meaning that $f(1)=1$;
		\item satisfies the symmetry ${f(x) = x f(x^{-1})}$.
	\end{enumerate}
	In~\hyperref[box:standardMonotones]{Box 2} it was shown that one can pointwise bound all the standard monotones as:
	\begin{align}
		\frac{2x}{x+1}\, \leq f(x)\leq\, \frac{x+1}{2}\,,\label{cf:eq:rangeStandardMonotones3}
	\end{align}
	which automatically proves that all $f\in \mathcal{M}$ are bounded for  $x\in[0,1]$, and diverge as $x\rightarrow\infty$ at most in a linear manner. We denote by $f_H(x)$ and $f_B(x)$ the smallest and the largest standard monotone functions. Thanks to the pointwise characterisation of Eq.~\eqref{cf:eq:rangeStandardMonotones3}, it is then natural to introduce the partial order on $\mathcal{M}$ defined by $f_1\leq f_2$ if and only if $f_1(x)\leq f_2(x)$ for every $x\in\RR^+$. Then, every $f\in \mathcal{M}$ satisfies $f_H\leq \,f\,\leq f_B$ (simply rewriting Eq.~\eqref{cf:eq:rangeStandardMonotones3}), but it should be noticed that this is not an if and only if condition: there are functions for which Eq.~\eqref{cf:eq:rangeStandardMonotones3} holds but that are not standard monotones\footnote{One can easily construct examples of this: consider for instance the function $f(x):= \cos(x+x^{-1})^2 \,f_B(x)+\sin(x+x^{-1})^2 f_H(x)$. By construction it is Fisher adjusted, satisfies the symmetry condition $f(x) = x f(x^{-1})$, and satisfies $f_H\leq \,f\,\leq f_B$. Still, it is clearly not monotone, as it oscillates infinitely many times between $f_H$ and $f_B$.}.

	
	There is a fundamental symmetry of $\mathcal{M}$ that was exploited to prove Eq.~\eqref{cf:eq:rangeStandardMonotones3}, given by the transformation $f(x)\rightarrow [Tf](x):= x/f(x)$. Thanks to condition~\ref{item:tTransform} in Lemma~\ref{cf:lemma:TF}, this  transformation maps standard monotones to standard monotones (i.e, $T:\mathcal{M}\rightarrow\mathcal{M}$). Moreover, $T$ is also order-reversing ($f\leq h \iff Th\leq Tf$), from which it directly follows that $Tf_B = f_H$: indeed, since for every $f$ in $\mathcal{M}$ one has $f\,\leq f_B$, then it also holds that $\forall f \in\mathcal{M}$, $Tf_B\,\leq Tf$, which is the defining property of $f_H$ (this is exactly part of the derivation used in  \hyperref[box:standardMonotones]{Box 2} to prove Eq.~\eqref{cf:eq:rangeStandardMonotones3}). Moreover, thanks to the symmetry $f(x) = x f(x^{-1})$ we can rewrite the action of $T$ as:
	\begin{align}
		[Tf](x) = \frac{x}{f(x)} = \frac{1}{f(x^{-1})}\,.\label{cf:eq:85}
	\end{align}
	From this expression one can directly verify that $T$ is involutive, i.e., $TTf = f$. One first consequence is that $Tf_H = TTf_B = f_B$. Moreover, this also means that we can partition $\mathcal{M}$ into two sets $\mathcal{M}_0$ and $\mathcal{M}_1$ such that $\mathcal{M}_0\cup \mathcal{M}_1 = \mathcal{M}$ and $\mathcal{M}_1 = T\mathcal{M}_0$. Moreover, we can choose the intersection between this two sets to be only one element, i.e., the only function for which $f = Tf$, namely $f_{SQ}(x) = \sqrt{x}$. Then, by defining $\mathcal{M}_{00}:= \mathcal{M}_0\setminus f_{SQ}$, we can express the set of standard monotones as:
	\begin{align}
		\mathcal{M} = \norbra{\bigsqcup_{f\in \mathcal{M}_{00} } \{f, Tf \}} \bigsqcup \,f_{SQ}\,,\label{eq:standardMonotonesDecomposition}
	\end{align}
	that is, as the disjoint union of sets of only two elements and, additionally, the square root function. 
	
	Thanks to the possibility of rewriting $T$ as in Eq.~\eqref{cf:eq:85}, this transformation is also particularly useful because it allows to lift some properties of $f\in \mathcal{M}$ to their multiplicative inverse, and vice versa. For example, it was proved in~\cite{hansen2008metric} that the inverse of a standard monotone can be rewritten as:
	\begin{align}
		\frac{1}{f(x)} = \int_0^1\de\mu_f(\lambda)\; \norbra{\frac{\lambda+1}{2}}\norbra{\frac{1}{x+\lambda}+\frac{1}{1+\lambda \,x}}\,,\label{cf:eq:80}
	\end{align}
	where $\de\mu_f(\lambda)$ is a probability measure on $[0,1]$. Moreover, by plugging into the integral above an arbitrary probability measure, the  resulting function is always the inverse of a standard monotone. It then follows from Eq.~\eqref{cf:eq:85} that we can express standard monotone functions as:
	\begin{align}
		f(x) = \frac{x}{Tf(x)}= \int_0^1\de\mu_{Tf}(\lambda)\; \norbra{\frac{1+\lambda}{2}}\norbra{\frac{x}{x+\lambda}+\frac{x}{1+\lambda \,x}}= \int_0^1\de\mu_{Tf}(\lambda)\;f_\lambda(x)\,.\label{cf:eq:81}
	\end{align}
	This shows that $\mathcal{M}$ is a convex set, which coincides with the convex hull of the extreme points denoted by $f_\lambda(x)$ (in particular since the set of $f_\lambda(x)$ is closed, this space has the structure of a Bauer simplex~\cite{hiai2013families}).
	
	\subsection{Properties of quantum Fisher operators}\label{sec:FisherOperatorProperties}
	Now that the main properties of the set $\mathcal{M}$ have been discussed, we can start studying how these reflect on the operators $\J_f\big|_\pi $ and $\J_f^{-1}\big|_\pi $. First, let us explore the main implications of the three defining properties:
	\begin{enumerate}
		\item the monotonicity of $f$ implies that for every CPTP map $\Phi$ the two inequalities hold:
		\begin{align}
			\Phi^\dagger\norbra{\J_f^{-1}\big|_{\Phi(\pi)}}\Phi\;&\leq\; {\J_f^{-1}\big|_{\pi}}\,,\qquad\qquad\qquad
			\Phi\norbra{\J_f\big|_{\pi}}\Phi^\dagger\;\leq\; {\J_f\big|_{\Phi(\pi)}}\,;\label{eq:monotonicityPropertiesPolar}
		\end{align}
		\item the property of $f$ to be Fisher adjusted implies that $\J_f\big|_{\pi}[A] = A\pi$ and $\J_f^{-1}\big|_{\pi}[A] = A\pi^{-1}$ for every $A$ such that $[A,\pi]=0$;
		\item the symmetry ${f(x) = x f(x^{-1})}$ implies that $\J_f\big|_{\pi}$ is adjoint-preserving, i.e., it maps self-adjoint operators into self-adjoint operators.
	\end{enumerate}
	The proofs of these facts, and the other manipulations of the section are deferred to App.~\ref{box:fisherComputations}. 
	
	In order to make the discussion more concrete, it is useful to give the coordinate expression of $\J_f\big|_{\pi}$ in the eigenbasis of $\pi$. To this end, we express the state in coordinates as:
	\begin{align}
		\pi = \sum_i \pi_i \ketbra{i}{i}\,.
	\end{align}
	It should be noticed that the modular operator is diagonal in the basis given by $\{\ketbra{\pi_i}{\pi_j}\}$, as
	$\LL_\pi\RR_\pi^{-1} [\ketbra{\pi_i}{\pi_j}] = \frac{\pi_i}{\pi_j}\ketbra{\pi_i}{\pi_j}$. Thanks to this, it is straightforward to compute the action of $\J_f\big|_{\pi}$ on $\ketbra{i}{j}$ by directly using the definition in Eq.~\eqref{145}:
	\begin{align}
		&\J_f\big|_{\pi} \sqrbra{\ketbra{i}{j}}  = f\norbra{\frac{\pi_i}{\pi_j}}\pi_j \; \ketbra{i}{j}\,;\qquad\qquad\qquad
		\J_f ^{-1}\big|_{\pi}\sqrbra{\ketbra{i}{j}}  = \norbra{f\norbra{\frac{\pi_i}{\pi_j}}\pi_j}^{-1} \; \ketbra{i}{j}\,.\label{cf:eq:182x}
	\end{align}
	Consider now an operator $A$ which can be written in coordinates as $A  := \sum A_{i,j}\ketbra{i}{j}$. Then, $\J_f\big|_{\pi}$ acts as:
	\begin{align}\label{cf:eq:183as}
		\J_f\big|_{\pi}[A] =  \sum_{i,j}  f\norbra{\frac{\pi_i}{\pi_j}}\pi_j\; A_{i,j} \ketbra{i}{j} = \sum({J}_{f,\pi} \circ A)_{i,j} \ketbra{i}{j}\,,
	\end{align}
	where we introduced the matrix ${J}_{f,\pi} := \sum_{i,j} \J_f\big|_{\pi}[\ketbra{i}{j}]$, and we use the circle to denote the Hadamard product, that is a product acting as $A\circ B = \sum_{i,j} A_{i,j}B_{i,j}\ketbra{i}{j}$. Analogously, the same computation can be carried out for $\J_f^{-1}\big|_{\pi}$.
	
	Since all the $\J_f\big|_{\pi}$ are diagonal in the same basis, we can lift the partial order from $\mathcal{M}$ to the operators. Indeed, as it can be verified in coordinates, one has that:
	\begin{align}
		f_1\leq f_2 \quad\implies\quad\J_{f_1}\big|_{\pi} \leq\, \J_{f_2}\big|_{\pi} \,.\label{cf:eq:orderJ}
	\end{align}  
	Hence, it also follows that for every $f\in \mathcal{M}$ the corresponding operator satisfies $\J_{f_H}\big|_{\pi}\leq \J_{f}\big|_{\pi}\leq \J_{f_B}\big|_{\pi}$. This allows to prove a key property of the Fisher operators: if one restricts the spectrum of $\pi$ to be bounded away from zero, then $ \J_{f}^{-1}\big|_{\pi}$ is a bounded operator (as a function of $\pi$). Since for every $\J_{f}^{-1}\big|_{\pi}$ it holds that $\J_{f_B}^{-1}\big|_{\pi}\leq \J_{f}^{-1}\big|_{\pi}\leq \J_{f_H}^{-1}\big|_{\pi}$ (as the inversion reverts the inequalities), we can verify this fact just by studying the behaviour of $\J_{f_B}^{-1}\big|_{\pi}$. Then, by explicitly computing the spectrum using Eq.~\eqref{cf:eq:182x} for $f_B$, which takes the form $\curbra{\frac{2}{\pi_i+\pi_j}}_{i,j}$, we see that this is a bounded set  whenever $\pi$ is bounded away from zero. This property was the key ingredient in the proof of Thm.~\ref{cf:theo:CPiffContracts}.
	
	As it was mentioned in the previous section, the transformation $T$ lifts properties of $f$ to its inverse $1/f$. Indeed, the same also holds when considering Fisher information operators, thanks to the following chain of equalities:
	\begin{align}
		\J_{Tf}\big |_\pi &= \RR_\pi \,Tf(\LL_\pi\RR_\pi^{-1}) = \RR_\pi \frac{1}{f}(\LL_\pi^{-1}\RR_\pi) =\RR_{\pi^{-1}}^{-1}\frac{1}{f}(\LL_{\pi^{-1}}\RR_{\pi^{-1}}^{-1}) = \J^{-1}_{f}\big |_{\pi^{-1}}\,.\label{cf:eq:TTransformApp}
	\end{align}
	This means that the properties of $\J_{f}\big |_\pi$ are in one to one correspondence with the properties of $\J^{-1}_{Tf}\big |_{\pi^{-1}}$, fact that will be particularly relevant in the next section.
	
	Finally, it should be noticed that the convex structure of $\mathcal{M}$ also reflects on the operators $\J_{f}\big|_{\pi}$. For this reason, it is interesting to study the operators associated to $f_\lambda$, namely:
	\begin{align}
		\J_{f_\lambda} \big|_{\pi}[A] 
		&= \norbra{\frac{1+\lambda}{2}} \norbra{(\LL_\pi+\lambda\RR_\pi)^{-1} + (\RR_\pi+\lambda\LL_\pi)^{-1}}[\pi A\pi] \,,
	\end{align}
	Then, thanks to Eq.~\eqref{cf:eq:81}, we have that any arbitrary quantum Fisher operator $\J_{f}\big|_{\pi}$ can be written as:
	\begin{align}
		\J_{f} \big|_{\pi}[A]&=\int_0^1\de\mu_{Tf}(\lambda)\; \norbra{\frac{1+\lambda}{2}} \norbra{(\LL_\pi+\lambda\RR_\pi)^{-1} + (\RR_\pi+\lambda\LL_\pi)^{-1}}[\pi A\pi]\,.\label{eq:157}
	\end{align}
	Moreover, using the transformation in Eq.~\eqref{cf:eq:TTransformApp}, we also  have the expression for generic $\J_{f}^{-1}\big|_{\pi}$, namely:
	\begin{align}
		\J_{f}^{-1} \big|_{\pi}[A] &= \J_{Tf} \big|_{\pi^{-1}}[A]=\int_0^1\de\mu_{f}(\lambda)\; \norbra{\frac{1+\lambda}{2}} \norbra{(\LL_\pi+\lambda\RR_\pi)^{-1} + (\RR_\pi+\lambda\LL_\pi)^{-1}}[ A]\,.\label{eq:158}
	\end{align}
	The two formulas above are derived in App.~\ref{box:fisherComputations}. It should be noticed that the expression for $\J_{f}^{-1} \big|_{\pi}$ could  also be directly inferred from the integral expansion of $1/f$ in Eq.~\eqref{cf:eq:80}.
	
	\begin{additional-info}{Box 10. Some additional properties of contrast functions~\hyperref[sec:CPfisher]{$\rightarrow$}}\label{box:contrastAdditionalProperties}
		Now that the quantum Fisher operators have been more thoroughly characterised, we can discuss some other properties of the contrast functions. In particular, by using their integral expression in the symmetrised case (and taking the adjoint of the second term in Eq.~\eqref{cf:eq:AppIntegralExpression}) we have up to corrections of order $ \bigo{\varepsilon^3}$ that:
		\begin{align}
			H^{\rm symm}_g(\pi + \varepsilon A || \pi+ \varepsilon B)  &=  \frac{\varepsilon^2}{2}\int_0^1 \de N_g(s)\;\Tr{(A-B)\norbra{(\LL_\pi + s \RR_\pi)^{-1}+(\RR_\pi + s \LL_\pi)^{-1}}[(A-B)]}=\label{eq:hgSymmetrisedLocal}\\
			&=\frac{\varepsilon^2}{2}\, \Tr{(A-B)\,\J_f^{-1}\big|_{\pi} [(A-B)]}\,,\label{cf:eq:AppIntegralExpression2}
		\end{align}
		where in the second line we simply used Thm.~\ref{cf:thm:Ruskai}. The expression in Eq.~\eqref{eq:hgSymmetrisedLocal} should be compared with the one obtained in Eq.~\eqref{eq:158}, which shows that the two defining measures are connected by the relation:
		\begin{align}
			\de N_g(\lambda) := \de\mu_{f}(\lambda)\; \norbra{\frac{1+\lambda}{2}}\,.\label{eq:cf:102o}
		\end{align}
		Since $\de\mu_{f}$ is a probability distribution, this also means that $\de N_g$ satisfies:
		\begin{align}
			\int_0^1 \de N_g(s)\; \frac{2}{1+s} = 1\,.\label{eq:normalisationFisherMeasure}
		\end{align}
		Indeed, this normalisation condition could have even been verified directly from the requirement that $H^{\rm symm}_g(\pi + \varepsilon A || \pi+ \varepsilon B)$ needs to reduce to the classical Fisher information for commuting variables. Indeed, for $[A,\pi]=0$ one has $(\LL_\pi + s \RR_\pi)^{-1}[A] =\frac{A\pi^{-1}}{(1+s)}$. Moreover, it also directly follows by comparing Eq.~\eqref{eq:21nuova} with Eq.~\eqref{cf:eq:80}. 
		
		Let us define the set of standard convex functions to be the set of functions $g\in\mathcal{G}$ satisfying the following three properties:
		\begin{enumerate}
			\item $g$ is matrix convex;
			\item $g$ is Fisher adjusted, that is $g''(1) = 1$;
			\item $g$ satisfies the symmetry $g(x)=x g(x^{-1})$.
		\end{enumerate}
		As it was mentioned in condition~\ref{item:3standardMonotoneFirst} (Sec.~\ref{contrastFunctionsToFisher}), the normalisation in the second requirement is chosen so that the corresponding contrast function will correctly reduce to a quantum Fisher information. Finally, since $H_g(\rho||\sigma) $ coincides with $H_{\tilde g}(\sigma||\rho) $, with $\tilde{g}(x)= x g(x^{-1})$ (as discussed in Eq.~\eqref{cf:eq:integralRepresentationG}), the last requirement restricts our attention to symmetric contrast functions. Then, there exists a bijective map between the set of standard convex functions and the set of standard monotones $L:\mathcal{G}\rightarrow\mathcal{M}$ which, thanks to Eq.~\eqref{cf:eq:correspondenceFG}, can be explicitly expressed as $Lg(x) =\frac{ (x-1)^2}{2 g(x)}$. Moreover, it is easy to verify that this transformation is involutive, meaning $L^2g = g$. Thanks to Eq.~\eqref{cf:eq:gIntegralExpression}, together with the discussion at the beginning of the section, we can also see that $\mathcal{G}$ has the structure of a Bauer simplex, with extreme points given by:
		\begin{align}
			g_\lambda(x) = \norbra{\frac{\lambda+1}{2}}\norbra{\frac{(x-1)^2}{x+\lambda}+\frac{(x-1)^2}{1+\lambda \,x}}\,.
		\end{align}
		
		It should be noticed that $L$ reverses the partial order $\leq$ defined on $\mathcal{M}$, meaning that $f_1\leq f_2 \iff Lf_2\leq Lf_1$. This implies that there is a larger and a smaller elements in $\mathcal{G}$, respectively corresponding to $Lf_H$ and $Lf_B$. In formulae, this reads:
		\begin{align}
			\frac{(x-1)^2}{x+1}\,\leq \,g(x) \,\leq\, \frac{(x+1)(x-1)^2}{4x}\,.
		\end{align}
		Interestingly, the same ordering is also present for symmetrised contrast functions. Indeed, we have that for every $g\in\mathcal{G}$ it holds that (see App.~\ref{app:orderContrast}):
		\begin{align}
			H^{\rm symm}_{Lf_B} ( \rho || \sigma)\leq H_g^{\rm symm} ( \rho || \sigma) \leq H_{Lf_H}^{\rm symm} ( \rho || \sigma)\,,\label{eq:165}
		\end{align}
		where we highlighted the fact that the contrast functions considered here are symmetric. 
		
		Finally, it should be noticed that we can lift the transformation $T:\mathcal{M}\rightarrow\mathcal{M}$ to $\tilde{T}:\mathcal{G}\rightarrow\mathcal{G}$ by requiring the following diagram to commute:
		\begin{center}
			\begin{tikzcd}
				\mathcal{G}\arrow{r}{\tilde{T}} \arrow[swap]{d}{{ L}} & \tilde T\mathcal{G}\arrow{d}{{ L}} \\%
				\mathcal{M}\arrow{r}{T}& T\mathcal{M}\;.
			\end{tikzcd}
		\end{center}
		Thanks to the involutivity of $L$, this is given by $\tilde{T} : = LTL$, which explicitly reads:
		\begin{align}
			[\tilde Tg](x) := \frac{(x-1)^4}{4\,x\,g(x)}\,.
		\end{align}
		The procedure just presented exemplifies a general method to lift transformations from $\mathcal{M}$ to $\mathcal{G}$ and vice-versa.
	\end{additional-info}
	
	\subsection{Complete positivity of the Fisher information functionals}\label{sec:CPfisher}
	
	A linear map is physically realisable only if it is completely positive. For this reason, it is particularly relevant to give a characte
	risation of the cases in which $\J_f\big|_{\pi}$ or its inverse are CP. To this end, we introduce the two sets:
	\begin{align}
		\mathcal{M}^{+}:= \curbra{f\in \mathcal{M}\,| \,\J_f\big|_{\pi} \,\text{is CP }\forall\,\pi}\,;\qquad\qquad
		\mathcal{M}^{-}:= \curbra{f\in \mathcal{M} \,|\, \J_f^{-1}\big|_{\pi} \,\text{is CP }\forall\,\pi}\,.
	\end{align} 
	Before starting with the exploration of this two sets, it should be noticed that  in the context of Fisher operators positivity implies complete positivity. In fact, $\J_f\big|_{\pi}$ is CP means that $\J_f\big|_{\pi} \otimes \idO_n$ is P for any $n$. Then, thanks to the identity $\J_f\big|_{\pi} \otimes \idO_n=\J_f\big|_{\pi\otimes \id_n} $, if $\J_f\big|_{\pi}$ is P for all states (and all dimensions), it automatically follows that it will also be CP, proving the claim. 
	
	The expression in Eq.~\eqref{cf:eq:183as} is particularly useful in this context because it says that one can interpret the application of $\J_f|_\pi$ to a state $\rho$ as the Hadamard product of $J_{f,\pi}$ with $\rho$. Then, thanks to Schur product theorem~\cite{hiaiIntroductionMatrixAnalysis2014}, the resulting matrix is positive if both $J_{f,\pi}$ and $\rho$ are positive (the latter is true as $\rho$ is a state). Since this should hold for any $\rho$, it directly follows that $\J_f|_\pi$ is CP if and only if $J_{f,\pi}\geq0$.
	
	Following the discussion above, it is easy to find that a necessary condition for $f\in\mathcal{M^+}$ is that $f\leq f_{SQ}$. Indeed, since every principal sub-matrix of a positive matrix is also positive semidefinite, imposing the positivity of the determinant for the $2\times2$ matrix containing only the $(i,j)$-components of $J_{f,\pi}$, namely $\pi_i\pi_j-f(\pi_i/\pi_j)^2 \pi_j^2\geq0$, we obtain the condition $f(\pi_i/\pi_j)\leq \sqrt{\pi_i/\pi_j}$. Moreover, from the fact that this should hold for any non-zero probability vector $\{\pi_i\}_{i\in\{1,\dots,n\}}$ and any $n\in\mathbb{N}$, one can deduce that a necessary condition for $\J_f|_\pi$  to be positive preserving is that $f(x) \leq \sqrt{x}$ for every $x$. A similar argument can also be given for $\mathcal{M}^-$, showing that if $f\in\mathcal{M}^-$ then $f_{SQ}\leq f$.
	
	It should be noticed though that the one just presented is not a sufficient condition. Indeed, if one considers the family of extremal functions $f_\lambda$ defined in  Eq.~\eqref{cf:eq:81}, it holds that $f_\lambda\leq f_{SQ}$ for all $\lambda\in[3-2\sqrt{2},1]$. Still, it has been proved in~\cite{hiai2013families} that for all $\lambda\neq 1$ $f_\lambda\notin\mathcal{M}^+$. Moreover, we can deduce another interesting property of the set of monotones from the inspection of its extreme points: for any $\lambda\in(0,3-2\sqrt{2})$ the corresponding function has two additional crossing with the graph of the square root (other than at $x=0$ and $x=1$) implying that $f_\lambda$ neither lays all beneath nor all above $f_{SQ}$. Since these are necessary condition for $f_\lambda$ to be in $\mathcal{M}^+$ or $\mathcal{M}^-$, this remark implies that $\mathcal{M} \neq \mathcal{M}^+\cup \mathcal{M}^-$, that is, there are standard monotone functions for which neither $\J_f|_\pi$ or $\J_f^{-1}|_\pi$ are CP.
	
	It was shown in Eq.~\eqref{eq:standardMonotonesDecomposition} that we could partition $\mathcal{M}$ using the transformation $T$, and that this could be used to express $\J_f|_\pi$ in terms of $\J_{Tf}^{-1}|_{\pi^{-1}}$ (see Eq.~\eqref{cf:eq:TTransformApp}). Thanks to this last property it follows that $T\mathcal{M}^+ = \mathcal{M}^-$. Indeed, suppose $f\in\mathcal{M}^+$. Then, $\J_f|_\pi$ admits the Kraus form:
	\begin{align}
		\J_f\big|_{\pi} \sqrbra{\rho} = \sum_i K_i(\pi)\,\rho\, K_i(\pi)^\dagger\,,
	\end{align}
	where $i$ ranges on a possibly uncountable set, and $K_i(\pi)$ are $\pi$-dependent Kraus operators. Then, thanks to Eq.~\eqref{cf:eq:TTransformApp}, we can express $\J^{-1}_{Tf}|_\pi$ as:
	\begin{align}
		\J_{Tf}^{-1}\big|_{\pi} \sqrbra{\rho}= \J_f\big|_{\pi^{-1}} \sqrbra{\rho} = \sum_i K_i(\pi^{-1})\,\rho \,K_i(\pi^{-1})^\dagger\,,
	\end{align}
	proving that $\J^{-1}_{Tf}|_\pi$ is also CP, and thus $Tf\in\mathcal{M}^-$. This shows that the subset $\mathcal{M}^+\cup \mathcal{M}^-$ of $\mathcal{M}$ is stable under the transformation $T$. It should be noticed that since this is a strict subset, it was not obvious from the beginning that this would be the case. 
	
	We have seen that one can give a necessary condition for $f$ to be in $\mathcal{M}^+$ in terms of the partial order induced by the pointwise order of real functions. It is then a remarkable fact that one can introduce a partial order that implies the pointwise one, and that can be used to completely characterise $\mathcal{M}^+$ and $\mathcal{M}^-$. To this end, we need to define the concept of  positive definite continuous functions. These are functions $h:\RR\rightarrow\mathbb{C}$ such that for any vector of reals $\{t_i\}_{i\in\{1,\dots,n\}}$ of arbitrary size $n$, the matrix $A$ defined in coordinates as $A_{i,j} = h(t_i-t_j) $ is positive semidefinite. This class of functions is closed under multiplication, and the elements of the set are uniformly bounded by their value in zero. Finally, a key result called Bochner's theorem says that a function is positive definite if and only if it is the Fourier transform of a finite positive measure on $\RR$~\cite{reed1975ii}. This last property is the key ingredient of Thm.~\ref{cf:thm:cpCharacterisation}. 
	
	Then, we define the partial order $\preceq$ on $\mathcal{M}$ by saying that $f_1\preceq f_2$ if $f_1(e^t)/f_2(e^t)$ is positive definite or, equivalently, if the matrix with entries:
	\begin{align}
		A_{i,j} := \frac{f_1 (\pi_i/\pi_j)}{f_2 (\pi_i/\pi_j)}
	\end{align}
	is positive semidefinite for any non-zero probability vector $\{\pi_i\}_{i\in\{1,\dots,n\}}$ of any fixed size $n$. Before proving that $\preceq$ is an actual order relation, it is useful to point out that $f_1\preceq f_2$ implies $f_1\leq f_2$~\cite{hiai1999means}. Then, we can verify that $\preceq$ satisfies all the necessary conditions to induce a partial order:
	\begin{enumerate}
		\item \emph{reflexivity:} i.e., $f\preceq f$, which follows from the fact that a matrix with all entries equal to $1$ is positive semidefinite;
		\item \emph{antisymmetry:} i.e., the two relations $f_1\preceq f_2$ and $f_2\preceq f_1$ imply $f_1= f_2$. Since $f_1\preceq f_2\,\implies\,f_1\leq f_2$, then antisymmetry is a consequence of the same relation for the pointwise order;
		\item \emph{transitivity:} this condition says that if $f_1\preceq f_2$ and $f_2\preceq f_3$, then it should hold that $f_1\preceq f_3$. The proof of this fact is a consequence of the closure of the class of positive definite functions under multiplication, as one can rewrite:
		\begin{align}
			\frac{f_1(e^t)}{f_3(e^t)} =\norbra{ \frac{f_1(e^t)}{f_2(e^t)}}\norbra{ \frac{f_2(e^t)}{f_3(e^t)}}\,.
		\end{align}
		Since the two functions on the right hand side of the equation are positive definite (as it follows from the assumption that $f_1\preceq f_2$ and $f_2\preceq f_3$) then also their product is. This proves the claim.
	\end{enumerate}
	
	The definition of $\preceq$ in this context could look rather arbitrary. Still, it can be argued that it is actually a very natural relation when it comes to the characterisation of the sets of completely positive quantum Fisher operators. Indeed, one can specify all the functions $f\in\mathcal{M}^+$ as exactly the ones satisfying $f\preceq f_{SQ}$. The reason for this is simple: $f\preceq f_{SQ}$ if and only if the matrix with entries: 
	\begin{align}
		A_{i,j} :&= f \norbra{\frac{\pi_i}{\pi_j}}\sqrt{\frac{\pi_j}{\pi_i}} = \norbra{f \norbra{\frac{\pi_i}{\pi_j}} \pi_j} \frac{1}{\sqrt{\pi_i\pi_j}}=\frac{1}{\sqrt{\pi_i}}\,(J_{f,\pi})_{i,j}\, \frac{1}{\sqrt{\pi_j}}\,,
	\end{align}
	is positive semidefinite. Interestingly, the last equality shows that $A$ is connected to $J_{f,\pi}$ by a transformation which preserves positivity. In other words, the condition that $f\preceq f_{SQ}$ is equivalent to the requirement that $J_{f,\pi}$ is positive semidefinite for any $\pi$, which exactly corresponds to $\J_f|_\pi$ being CP. Thus, for any $f\in\mathcal{M}$ there is the equivalence $f\preceq f_{SQ}\iff f\in\mathcal{M}^+$.
	
	It is also easy to verify that $T$ reverts the order $\preceq$. Indeed, if $f_1\preceq f_2$, then it follows from:
	\begin{align}
		\frac{f_1(e^{t})}{f_2(e^{t})}=\frac{\cancel{e^{-t}}\,Tf_2(e^t)}{\cancel{e^{-t}}\,Tf_1(e^t)}\,,
	\end{align}
	which directly implies $Tf_2\preceq Tf_1$. Indeed, thanks to the bijection between $\mathcal{M}^+$ and $\mathcal{M}^-$, it also follows that we can completely characterise the latter as  all the functions satisfying $f_{SQ}\preceq f$. This also shows that we can partition $\mathcal{M}$ as:
	\begin{align}
		\mathcal{M} = \norbra{\mathcal{M}^+ \cup \mathcal{M}^-}\bigsqcup \bigg( \bigsqcup_{f\npreceq f_{SQ}\land  f_{SQ}\npreceq f}\{f \}\bigg)\,,
	\end{align}
	i.e., in functions for which either $\J_f|_\pi$ or $\J_f^{-1}|_\pi$ are CP, and ones that are incomparable with $f_{SQ}$.	Moreover, from the reflexivity of $\preceq$  we also have that $\mathcal{M}^+ \cap \mathcal{M}^-  = f_{SQ}$ (this could also be deduced from the fact that $T$ has a unique fixed point, together with the fact that $T\mathcal{M}^+ = \mathcal{M}^-$).
	
	The use of positive definite functions does not just help in characterising which elements of $\mathcal{M}$ are in $\mathcal{M}^+$ or $\mathcal{M}^-$, but it also allows to give an analytical expression for general CP quantum Fisher operators~\cite{hiai2013families}:
	\begin{theorem}\label{cf:thm:cpCharacterisation}
		For any $f\in\mathcal{M}^+$ there exists a symmetric probability distribution $\de\nu^+_f(s)$ on $\RR$ such that:
		\begin{align}
			\J_f|_{\pi}[A] = \int_{-\infty}^{\infty} \de\nu^+_f(s)\; \pi^{is+\frac{1}{2}}\,A\,\pi^{-is+\frac{1}{2}}\,.\label{cf:eq:115z}
		\end{align}
		Similarly, for any $f\in\mathcal{M}^-$ there exists a symmetric probability distribution $\de\nu^-_f(s)$ on $\RR$ such that:
		\begin{align}
			\J_f^{-1}|_{\pi}[A] = \int_{-\infty}^{\infty} \de\nu^-_f(s)\; \pi^{is-\frac{1}{2}}\,A\,\pi^{-is-\frac{1}{2}}\,.\label{cf:eq:116z}
		\end{align}
		Moreover, the defining probability distributions coincide when mapping $\mathcal{M}^+$ into $\mathcal{M}^-$ through $T$, that is, for every $f\in \mathcal{M}^+$ we have that $\de\nu^+_f = \de\nu^-_{Tf}$.
	\end{theorem}
	\begin{proof}
		This result is a straightforward application of Bochner's theorem. Let us first analyse the case in which $f\in\mathcal{M}^+$. Since $f\preceq f_{SQ}$, this means that $e^{-t/2} f(e^t)$ is positive definite. Hence, from Bochner's theorem there exists a unique probability measure $\de\nu^+_f(s)$ on $\RR$ such that:
		\begin{align}
			e^{-t/2} f(e^t) = \int_{-\infty}^{\infty} \de\nu^+_f(s)\;e^{i t s}\,.\label{cf:eq:117z}
		\end{align}
		Since $f$ is standard it follows from the symmetry $e^{-t/2} f(e^{t})=e^{-t/2}e^{t} f(e^{-t})  = e^{t/2} f(e^{-t})$ that $\de\nu^+_f(s)$ has to be symmetric with respect to zero. Then, using the definition of $\J_f\big|_{\pi}$ in Eq.~\eqref{cf:eq:65} we have:
		\begin{align}
			\J_f\big|_\pi[A] &= \RR_\pi \,f\norbraB{\LL_\pi\RR_\pi^{-1}}[A] = \RR_\pi\int_{-\infty}^{\infty} \de\nu^+_f\;(\LL_\pi\RR_\pi^{-1})^{is +1/2}[A]
			=\int_{-\infty}^{\infty} \de\nu^+_f(s)\;\pi^{is +1/2}\,A\,\pi^{-is +1/2}\,,
		\end{align}
		proving Eq.~\eqref{cf:eq:115z}. At this point we can use the bijection between $\mathcal{M}^+$ and $\mathcal{M}^-$ to show that for every $f\in\mathcal{M}^-$ we can express $\J_f^{-1}\big|_\pi$ as:
		\begin{align}
			\J_f^{-1}\big|_\pi[A] = \J_{Tf}\big|_{\pi^{-1}}[A] = \int_{-\infty}^{\infty} \de\nu^+_{Tf}(s)\;\pi^{-is -1/2}\,A\,\pi^{is -1/2} = \int_{-\infty}^{\infty} \de\nu^-_{f}(s)\;\pi^{is -1/2}\,A\,\pi^{-is -1/2}\,,
		\end{align}
		where in the last step we used the symmetry of $\de\nu^+_{Tf}$ to perform the change of variables $s\rightarrow(-s)$ and we implicitly defined $\de\nu^-_{f}$. This equality proves Eq.~\eqref{cf:eq:116z} and the last claim.
	\end{proof}
	
	The theorem just presented gives the most general expression of Fisher operators $\J_f|_{\pi}$ and $\J_f^{-1}|_{\pi}$ that are CP. Still, it should be noticed that not all symmetric probability distributions generate a quantum Fisher operator, as it can be explicitly verified by plugging an arbitrary probability distribution in Eq.~\eqref{cf:eq:117z}. For this reason, the next result gives a characterisation of the possible measures one can use:
	\begin{theorem}\label{thm:generalExpressionCPmeasure}
		For every $f\in\mathcal{M}^+$ the defining measure $\de\nu^+_f$ in Eq.~\eqref{cf:eq:115z} satisfies:
		\begin{align}
			\de\nu^+_f(s) = \frac{1}{2\pi} \int_{-\infty}^\infty \dt \; e^{- (\frac{1}{2}+is)t}f(e^t) = \int_0^1\de\mu_{Tf}(\lambda)\;\cosh\norbra{\frac{\log\lambda}{2}}\frac{\cos(s\log\lambda)}{\cosh \pi s}\,;\label{180echis}
		\end{align}
		where $\de\mu_{Tf}$ is the probability distribution on $[0,1]$ defined in Eq.~\eqref{cf:eq:81}.
	\end{theorem}
	One can obtain this result by simply plugging the decomposition of general standard monotone functions from Eq.~\eqref{cf:eq:81} into Eq.~\eqref{cf:eq:117z} and performing an inverse Fourier transform. This theorem shows the difficulty of giving  a full characterisation of the allowed measures $\de\nu^+_f$: indeed, it should be noticed that the functions in the last integral of Eq.~\eqref{180echis} are not probability distributions (they are quasi-probabilities, i.e, they can be negative in general) so plugging generic $\de\mu_{Tf}$ will give negativities in the corresponding $\de\nu^+_f$. Indeed, whereas the decomposition in Eq.~\eqref{cf:eq:81} holds for any $f\in\mathcal{M}$, it makes sense to define $\de\nu^+_f$ only for $f\in\mathcal{M}^+$.
	
	Finally, it should also be pointed out that whereas $f\in\mathcal{M} \implies f_H\,\leq \,f \, \leq\,f_B$, the same does not hold for the partial order $\preceq$. Indeed, for $\lambda$ small enough (different from zero), the extreme points defined in Eq.~\eqref{cf:eq:81} satisfy $f_\lambda\npreceq f_B$: this can be proved by noticing that $f_\lambda(e^t)/f_B(e^t)$ is not a positive definite function, as it does not arise as the Fourier transform of a positive measure on $\RR$. This also justifies the need to introduce the two partial orders $\leq$ and $\preceq$: one gives necessary conditions to be in $\mathcal{M}$, the other completely specifies $\mathcal{M}^+$ and $\mathcal{M}^-$.
	
	\begin{additional-info}{Box 11. The set of inverse standard monotone functions~\hyperref[gardenFisher]{$\rightarrow$}}\label{box:inverseMonotones}
		As it was discussed in Sec.~\ref{sec:standardMonotonesProperties}, the set $\mathcal{M}$ is not only convex but it even has the structure of a simplex with a continuous number of vertices. It is easy to verify that the convex structure of $\mathcal{M}$  is also present for $\mathcal{M}^+$, since $\J_{c f_1+(1-\alpha) f_2}\big|_\pi = \alpha\,\J_{ f_1}\big|_\pi+(1-\alpha)\,\J_{ f_2}\big|_\pi$ and  convex sums preserve complete positivity. Still, the simplicial structure appears to be lost since the only extreme point for which $f_\lambda\in \mathcal{M}^+$ is for $\lambda = 1$.
		
		When considering $\mathcal{M}^-$ the situation becomes even worst: it has been shown in~\cite{hiai2013families} that there are pairs of functions in $\mathcal{M}^-$ for which any non-trivial convex sum is not an element of $\mathcal{M}^-$. Since convexity is a desirable property for analytical manipulations, we introduce in here a new set that will facilitate the treatment of $\mathcal{M}^-$. In particular, define $\mathcal{K}$ to be the set of inverse standard monotones, i.e., the set of functions given by:
		\begin{align}
			\mathcal{K} := \curbra{k\,\bigg| k(x) =\frac{1}{f(x)}, \, f\in\mathcal{M}}\,.
		\end{align}
		The members of $\mathcal{K}$ are matrix convex functions that satisfy $k(x) = k(x^{-1})x^{-1}$ and $k(1)=1$. Actually, since the inverse of a matrix convex function is matrix monotone~\cite{hiaiIntroductionMatrixAnalysis2014}, one could equivalently define $\mathcal{M}$ as the set of inverses of the members of $\mathcal{K}$. Indeed, in the literature both choices are used interchangeably, and which set to use boils down to a question of taste. We can pass from one set to the other using the two bijections $I_{1,2}:\mathcal{M}\rightarrow\mathcal{K}$ given by:
		\begin{align}
			[I_1f](x) = \frac{1}{f(x)}\,,\qquad\qquad\qquad[I_2f](x) = f(x^{-1})\,.
		\end{align}
		It should be noticed that $I_1$ and $I_2$ are both involutive, they commute (in the sense that $I_1I_2 f = I_2I_1 f $), and they satisfy the relation $I_1I_2 = T$, as it can be directly verified from Eq.~\eqref{cf:eq:85}. This also means that the two transformations are related two one another by $I_{1/2} = I_{2/1} T$.
		
		From the expression in  Eq.~\eqref{cf:eq:80} one can deduce that also $\mathcal{K}$ has the structure of a Bauer simplex, and that the extreme points are given by:
		\begin{align}
			k_\lambda(x) := \norbra{\frac{\lambda+1}{2}}\norbra{\frac{1}{x+\lambda}+\frac{1}{1+\lambda \,x}}\,.
		\end{align}
		Moreover, both partial orders $\leq$ and $\preceq$ can be defined on $\mathcal{K}$, and in this context it should be noticed that $I_1$ is order reversing, while $I_2$ preserves the order, since it is the composition of two order reversing maps. Hence, we can also bound the elements of $\mathcal{K}$ by using the relation $f_H\leq \,f\,\leq f_B$ and mapping it through either of $I_{1/2}$. Hence, any $k\in\mathcal{K}$ satisfies:
		\begin{align}
			[I_1 f_B](x) = [I_2 f_H](x) = \frac{2}{x+1} \leq\, k(x)\,\leq \frac{x+1}{2x} = [I_1 f_H](x) = [I_2 f_B](x)\,.
		\end{align}
		Finally, one can also define the quantum Fisher information operators directly in terms of $k\in\mathcal{K}$, where we use the new notation:
		\begin{align}
			\Omega_k\big|_\pi := k(\LL_\pi\RR_\pi^{-1})\,\R_\pi^{-1} = \J_{I_1k}^{-1}\big|_\pi = \J_{I_2k}\big|_{\pi^{-1}} \,.\label{cf:eq:106biss}
		\end{align}
		Thanks to this definition, one can lift the convex structure of $\mathcal{K}$ to the quantum Fisher operators, since $\Omega_{\alpha k_1 + (1-\alpha)k_2}|_\pi  = \alpha\,\Omega_{ k_1}|_\pi + (1-\alpha)\Omega_{ k_2}|_\pi $. If we then define the two sets:
		\begin{align}
			\mathcal{K}^{+}:= \curbra{k\in \mathcal{K}\,| \,\Omega_k\big|_{\pi} \,\text{is CP }\forall\,\pi}\,;\qquad\qquad
			\mathcal{K}^{-}:= \curbra{k\in \mathcal{K}\,| \,\Omega_k^{-1}\big|_{\pi} \,\text{is CP }\forall\,\pi}\,,
		\end{align}
		it follows from the remark above that $\mathcal{K}^+$ is a convex set, while it was proved in~\cite{hiai2013families} that $\mathcal{K}^-$ is not. Interestingly, it holds that $\mathcal{K}^{\pm} = I_1\mathcal{M}^{\mp}$, so this shows that depending on whether one is more interested in the complete positivity of $\J_f\big|_\pi$ or of $\J_f^{-1}\big|_\pi$ , it will more suitable to work with $\mathcal{M}$ or $\mathcal{K}$. 
		
		The following diagram summarises all the relations between the sets introduced in this section:
		\begin{center}
			\begin{tikzcd}
				\mathcal{M}^+\arrow{r}{{T}} \arrow[swap]{d}{{I_2}}\arrow{rd}[above]{~I_1} & \mathcal{M}^-\arrow{d}{{\rm I_2}} \\%
				\mathcal{K}^+\arrow{r}[below]{T}\arrow{ur}& \mathcal{K}^-\;.
			\end{tikzcd}
		\end{center}
		It should be noticed that all the transformations are involutive, so all the arrows can be reversed without changing the corresponding label. Moreover, the two sets on the right ($\mathcal{M}^+$ and $\mathcal{K}^+$) are convex.
	\end{additional-info}

	\section{A garden of quantum Fisher information}\label{gardenFisher}
	
	Now that the general theory of quantum Fisher information has been laid down, it is time to enter the tangle of disparate expressions that this quantity can present. One of the founders of the subject, Dénes Petz, having to define the richness of different examples one is confronted with, chose to call it ``a garden of monotonic metrics''~\cite{petz2002covariance}. In his honour, we use here the same term to designate this section, which is aimed at providing a field guide for the reader to orient themself in this florid jungle.

	To facilitate the identification of different metrics, we have summarised the ones treated here in Table~\ref{cf:fig:figstandardmonotones}, where we list the standard monotones in decreasing order (according to $\leq$), the corresponding contrast functions and quantum Fisher operators. The letter on each row refers to the subsection in which each of the different cases is treated, which we list here for convenience together with their main properties:
	\begin{enumerate}[label=\textbf{\Alph*.}]
		\item \hyperref[Bures]{\textbf{The Bures metric}:} \emph{the smallest among all the Fisher informations, and for this reason also the most studied in the literature. This is one of the two cases for which a closed form of the geodesic distance is known (see Eq.~\eqref{cf:eq:buresDistance}). Moreover, it also appears in a central result in estimation theory, the quantum Cramér-Rao bound (Eq.~\eqref{cf:eq:cramerRao130}), connecting the Fisher information to the minimal variance of an unbiased estimator. We give a generalised version of this bound in Eq.~\eqref{cf:eq:132b}.}
		\item \hyperref[Heinz]{\textbf{The Heinz family}:} \emph{a one-parameter family of functions in $\mathcal{M}^-$.} 
		\item \hyperref[alpha]{\textbf{The family of $\alpha$-divergences}:} \emph{a fundamental class of standard monotone functions. The corresponding contrast functions are related to Rényi divergences through Eq.~\eqref{contrastRenyi}. Most of the examples in this list are particular cases of this family.}
		\item \hyperref[wignerYanase]{\textbf{The Wigner-Yanase skew information ($\alpha = 1/2$)}:} \emph{the largest function in the family of $\alpha$-divergences. It's the only metric having constant positive curvature, making the space of states in this case isometric to a $n$-sphere. This allows to find a closed expression of the geodesic distance and of the geodesics (see Eq.~\eqref{cf:eq:wyGeodesicDistance} and Eq.~\eqref{cf:eq:Geodesic}). Moreover, it naturally appears in the context of hypothesis testing, in particular in the quantum Chernoff bound, expressed in Eq.~\eqref{chernoffBoundLocal}.}
		\item \hyperref[relativeEntropy]{\textbf{The relative entropy ($\alpha=0$)}:} \emph{the most famous among the contrast functions, the corresponding metric, called Kubo-Mori-Bogoliubov, also appears often in statistical mechanics, as it is related to the linear response of thermal states. In particular, in Eq.~\eqref{cf:eq:181} we discuss the interpretation of the generalised Cramer-Rao bound is this context.}
		\item  \hyperref[cf:sec:relEntVariance]{\textbf{The quantum information variance}:} \emph{this metric is related to the second derivative of the $\alpha$-divergences in zero (see Eq.~\eqref{265x}), and also to estimation theory for thermal states.}
		\item \hyperref[geometricMean]{\textbf{The geometric mean}:} \emph{the only standard monotone in the intersection of $\mathcal{M}^+$ and $\mathcal{M}^-$. This was used in Sec.~\ref{retrodiction} to define the Petz recovery map.}
		\item \hyperref[harmonic]{\textbf{The harmonic mean ($\alpha=2$)}:} \emph{the largest Fisher information (and the smallest standard monotone). It is also the only case in which the corresponding contrast function can be expressed as a $\chi_f^2$-divergence (see \hyperref[box:statisticalQuantifiers]{Box 3}).}
	\end{enumerate}
	
	\begin{figure}
		\centering
		\includegraphics[width=1\linewidth]{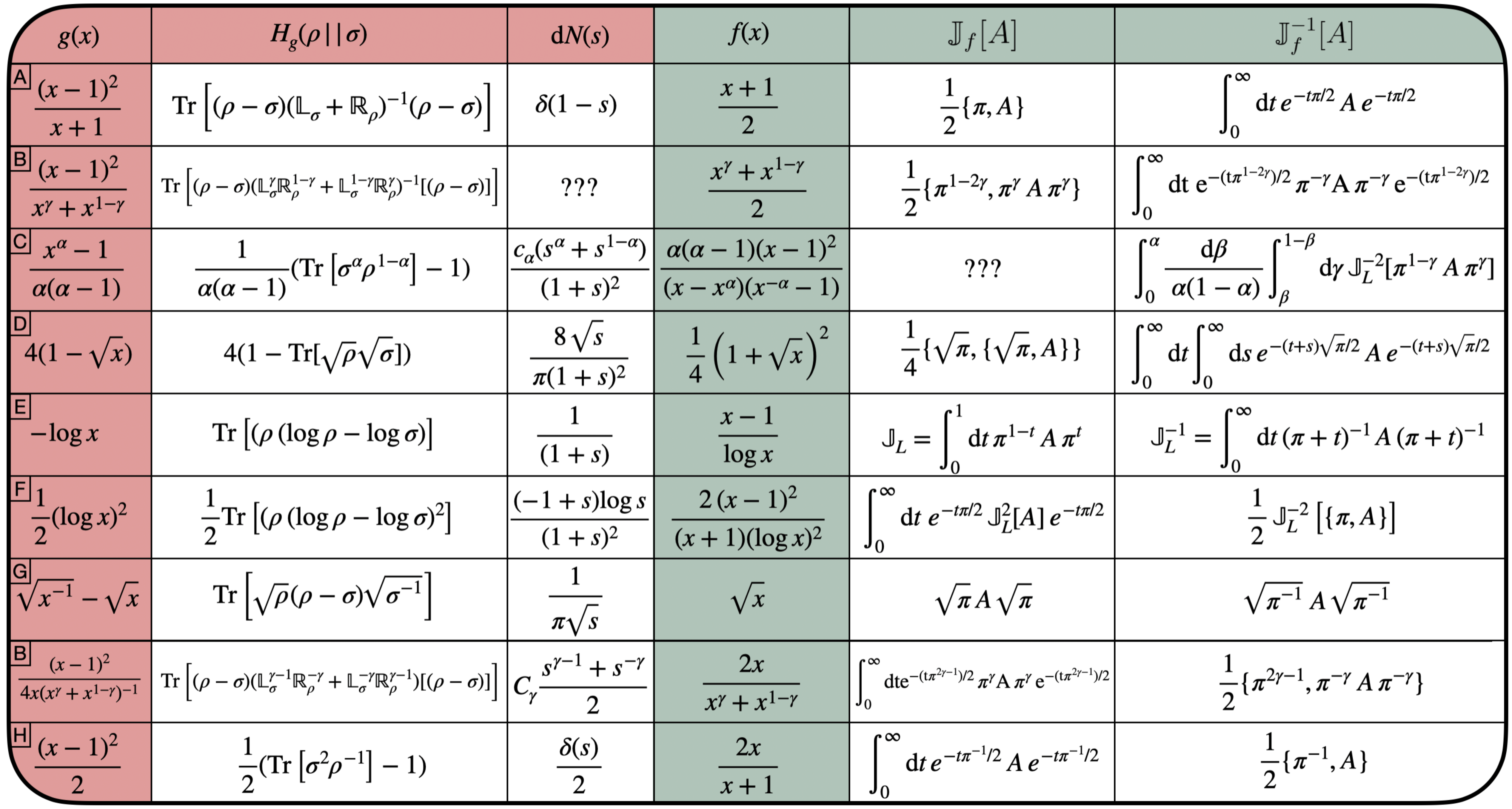}
		\caption{In the table above we summarise the expressions of standard monotone functions $f(x)$ analysed in the text. The table is divided in two parts: in the first half we present the properties of the associated contrast function, in the second the corresponding standard monotone. In particular, notice that the measure $\de N(s)$ is the one defined in Eq.~\eqref{cf:eq:AppIntegralExpression} for symmetric contrast functions. The question marks correspond to the entries for which no explicit form has been found. For reason of space, we have also introduced the constants $c_\alpha := \frac{\sin (\pi  \alpha )}{\pi  \alpha(1-\alpha )  }$ and $C_\gamma = \frac{\sin (\pi  \gamma )}{\pi}$. The letter on each row refers to the subsection in which each of the different cases is treated.}
		\label{cf:fig:figstandardmonotones}
	\end{figure}
	
	\subsection{The Bures metric}\label{Bures}
	Among the standard monotone functions, the maximum is given by:
	\begin{align}
		f_B(x) = \frac{x+1}{2}\,.\label{cf:eq:BuresFunction}
	\end{align}
	The associated standard convex function $g_B(x)$ can be obtained using the map $L$ defined in \hyperref[box:contrastAdditionalProperties]{Box 10}, giving:
	\begin{align}
		g_B(x) = [Lf_B](x)= \frac{(x-1)^2}{x+1} = \frac{1}{2}\int_0^1 \de N_B(s)\;\norbra{\frac{(x-1)^2}{x+s}+\frac{(x-1)^2}{1+s x}}\,,
	\end{align}
	where we implicitly defined the measure $\de N_B(s) := \delta(1-s)$.
	The corresponding contrast function is given by:
	\begin{align}
		H_{ B} ( \rho || \sigma)  =  \Tr{(\rho-\sigma)(\LL_\sigma+\RR_\rho)^{-1}(\rho-\sigma)} = \int_0^{\infty} \dt\; \Tr{(\rho-\sigma)\,e^{-t\sigma}\,(\rho-\sigma)e^{-t\rho}}\,,\label{eq:BuresContrastFunction}
	\end{align}
	where we used the integral representation of the operator $(\LL_\sigma+\RR_\rho)^{-1}$ proved in App.~\ref{app:inverseAnticomm}. By doing so we also highlight the symmetry in the contrast function, i.e., $H_B(\rho||\sigma) = H_B(\sigma||\rho)$, which directly follows from the fact that $L$ maps standard monotones in standard convex functions (see the discussion in \hyperref[box:contrastAdditionalProperties]{Box 10} for more details). 
	
	\begin{figure}
		\centering
		\hspace{-0.75cm}
		\includegraphics[width=.9\linewidth]{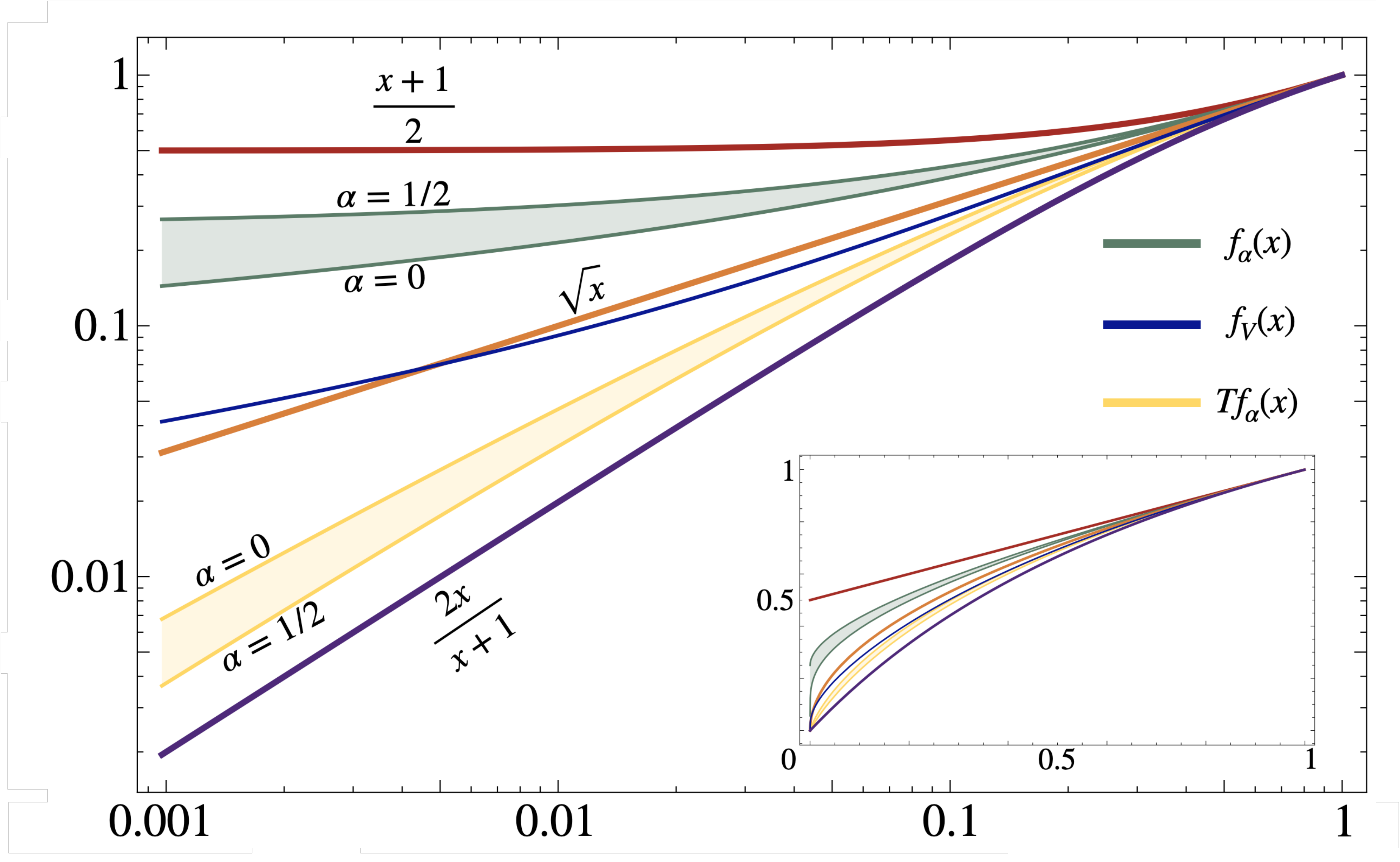}
		\caption{In the figure some of the most notable standard monotones are presented in a log-log scale (in the inset we show their behaviour in linear coordinates). In particular, we show the two extrema ($f_B$ from Sec.~\ref{Bures} and $f_H$ from Sec.~\ref{harmonic}), the square root (Sec.~\ref{geometricMean}), the family $f_\alpha$ of $\alpha$-divergences in the range $\alpha\in[0,1]$ and its transform $Tf_\alpha$ (Sec.~\ref{alpha}), together with the standard monotone associated with the quantum information variance $f_V(x)$ defined in Eq.~\eqref{cf:eq:fVariance} (Sec.~\ref{cf:sec:relEntVariance}). The shading in the two curves associated to $f_\alpha$ indicates that this family interpolates between  $f_0$, the standard monotone for the relative entropy (Sec.~\ref{relativeEntropy}), and the maximum value $f_{1/2} = \frac{1}{4}(1+\sqrt{x})^2$, corresponding to the Wigner-Yanase skew information (Sec.~\ref{wignerYanase}). It is interesting to notice that the monotone associated to the quantum information variance  does not satisfy $f_V(x)\geq\sqrt{x}$ nor $f_V(x)\leq\sqrt{x}$. This shows that there are monotones for which both $\J_f\big|_\pi$ and $\J_f^{-1}\big|_\pi$ are not CP (see Sec.~\ref{sec:CPfisher}).}
		\label{cf:fig:figstandardmonotones2}
	\end{figure}
	
	Applying Thm.~\ref{cf:thm:Ruskai}, it directly follows from Eq.~\eqref{eq:BuresContrastFunction} that the quantum Fisher operators are given by:
	\begin{align}
		\J_B\big|_{\pi}[A]  = \frac{1}{2}\,(\LL_\pi+\RR_\pi)[A] = \frac{1}{2}\,\{\pi, A\}\,; \qquad\qquad\qquad\J_B^{-1}\big|_{\pi}[A]  = 2(\LL_\pi+\RR_\pi)^{-1}[A] = \int_0^{\infty} \dt \; e^{-t\pi/2} \, A \,e^{-t\pi/2}\,.\label{cf:eq:buresJs}
	\end{align}
	It should be noticed that since $\J_B^{-1}\big|_{\pi}$ is in Kraus form, $f_B\in\mathcal{M}^-$.
	
	The metric generated by $\J_B^{-1}\big|_\pi$ is called Bures metric.
	Interestingly, this is one of the two examples for which a closed form for the geodesic distance exists~\cite{uhlmann1995geometric}, which is given by:
	\begin{align}
		d_B(\rho,\,\sigma)&= 2\arccos \norbra{\max\curbra{\Tr{WX^\dagger}\big| WW^\dagger = \rho,\;XX^\dagger = \sigma}} =2\arccos \norbra{\Tr{\sqrt{\sqrt{\rho}\sigma\sqrt{\rho}}}} =\\
		&= 2\arccos \norbra{F(\rho,\,\sigma)}\,,\label{cf:eq:buresDistance}
	\end{align}
	where in the last step we have implicitly defined the fidelity $F(\rho,\sigma)$. This quantity can be related to another quantifier of statistical distance, the Bures length, which takes the form~\cite{uhlmannSpheresHemispheresQuantum1996}:
	\begin{align}
		D_B(\rho,\,\sigma)^2 &= \min\curbra{\Tr{(W-X)(W-X)^\dagger}\big| WW^\dagger = \rho,\;XX^\dagger = \sigma} = 2\norbra{1-\Tr{\sqrt{\sqrt{\rho}\sigma\sqrt{\rho}}}} =\\&= 2\norbra{1-\cos\norbra{\frac{d_B(\rho,\,\sigma)}{2}}}\,.\label{cf:eq:bureslength}
	\end{align}
	It is  worthwhile to briefly sketch the main ideas that lead to these expressions, but we refer the interested reader to~\cite{bhatia2019bures,bengtssonGeometryQuantumStates2017} for a more extensive discussion. Consider the set of full rank matrices $\mathbb{M}_n$, endowed with the Hilbert-Schmidt scalar product (defined as $\langle A|B\rangle := \Tr{A^\dagger B}$), which can be thought as an extension of the Euclidean metric to matrices. There is a canonical projection $\pi:\mathbb{M}_n\rightarrow\mathbb{P}_n $on the set of positive matrices $\mathbb{P}_n$ defined as $A\rightarrow \rho := A A^\dagger$. It should be noticed that $A$ and $AU$ , with $U$ an arbitrary unitary, project on the same state. This gives to $\mathbb{M}_n$ the structure of a principal bundle, called the purification bundle. Moreover, $\mathbb{P}_n = \mathbb{M}_n/ U(n)$, that is positive matrices are the quotient space of full rank ones under the equivalence relation $A = B$ if $A = B U$ for some unitary $U$. Then, there exists a unique metric on $\mathbb{P}_n$ that makes $\pi$ an isometry on the image~\cite{gallot1990riemannian}. This metric turns out to be exactly the Bures one. Then, the set of matrices that project onto normalised states coincides with the sphere given by the equation $\Tr{AA^\dagger} = 1$. At this point, there are two natural distances defined in this space: the one on the surface of the sphere (i.e., the angle spanned by the great circle passing between any two points), and the chordal one (given by the shortest line in the Euclidean ambient space). Indeed, this is exactly what was computed in Eq.~\eqref{cf:eq:buresDistance} and Eq.~\eqref{cf:eq:bureslength} respectively, exploiting the gauge freedom in the purification to obtain the minimal distances in each case. 
	
	Moreover, one can easily show that the Fisher information associated to $\J_B^{-1}\big|_\pi$ upper bounds the square of the trace norm~\cite{temmeH2divergenceMixingTimes2010}. Indeed, thanks to the identity $||X||_1  =\sup_U |\Tr{XU}|$, it follows from a straightforward application of the Cauchy-Schwarz inequality that:
	\begin{align}
		||X||_1^2  &= \sup_{U| UU^\dagger = \id} |\Tr{XU}|^2  =  \sup_{U| UU^\dagger = \id} \left|\Tr{\J_B^{-1/2}\big|_\pi[X]\, \J_B^{1/2}\big|_\pi[U]}\right|^2  \leq\\
		&\leq \Tr{\J_B^{-1/2}[X]\,\J_B^{-1/2}[X]}\,\sup_{U| UU^\dagger = \id}\Tr{\J_B^{1/2}\big|_\pi[U^\dagger]\,\J_B^{1/2}\big|_\pi[U]}\,.
	\end{align}
	The first factor in the second line simply coincides with $\mathcal{F}_{B,\pi}(X)$. The second term instead can be explicitly computed as:
	\begin{align}
		\sup_{U| UU^\dagger = \id}\Tr{U^\dagger\,\J_B\big|_\pi[U]} = \frac{1}{2}\,\sup_{U| UU^\dagger = \id}\, \norbra{\Tr{U^\dagger \pi U}+\Tr{U^\dagger  U\pi}} = 1\,,
	\end{align}
	thanks to the cyclicity of the trace. Then, wrapping everything together, we finally obtain:
	\begin{align}
		||X||_1^2\,\leq \,\mathcal{F}_{B,\pi}(X)\,\leq\, \mathcal{F}_{f,\pi}(X)\,,\label{eq:buresInfoTraceNorm}
	\end{align}
	where the last inequality follows from the discussion in Sec.~\ref{sec:FisherOperatorProperties}. This shows that the Fisher information associated to the Bures metric is an upper bound to the square of the trace norm, and lowers bound all other Fisher informations.
	
	It should be pointed out that in the physics literature when one refers to the quantum Fisher information, actually what one has in mind is usually the Bures metric. Its importance is justified by the prominent role it plays in quantum metrology: indeed, one of the main results in this context is the famous Cramér-Rao bound, which gives a bound on the quality with which one can estimate parameters encoded in a state. For the sake of the discussion, consider in fact a family of density matrices $\rho(\theta)$ depending on some parameters $\theta$. Without loss of generality, suppose the true value of the parameter to be $\theta=0$ (this just corresponds to a change of variables). We define a locally unbiased estimator to be an observable satisfying~\cite{petzIntroductionQuantumFisher2011}:
	\begin{align}
		\frac{\partial}{\partial\theta}\,\Tr{\rho(\theta) A}\bigg|_{\theta=0} = 1\,.\label{cf:eq:localUnbiasedOperators}
	\end{align}
	This equation tells us that by measuring $A$ in a neighbourhood of $\theta=0$, one obtains the correct value of $\theta$ on average:
	\begin{align}
		\exists\,\varepsilon\;\big|\quad\Tr{\rho(\theta) A} = \theta \quad\quad \forall \theta \in(-\varepsilon, \varepsilon)\,.
	\end{align}
	This is a desirable feature, as it means that no transformation on the measured average is needed to give a good estimate of $\theta$. 
	Define now the symmetric logarithmic derivative (SLD) of $\rho(\theta)$ to be the operator $L_B$ such that, for any $X$, we have:
	\begin{align}
		\frac{\partial}{\partial\theta}\,\Tr{X\,\rho(\theta)}\bigg|_{\theta=0} = {\rm Re}\sqrbra{\,\Tr{X\,L_B\,\rho(0)}} =\frac{1}{2}\norbra{\Tr{X\,L_B\,\rho(0)}+\Tr{\rho(0)\,L_B\,X}}= \Tr{X\,\J_B\big|_{\rho(0)}[L_B]}\,,\label{cf:eq:102}
	\end{align}
	where one takes the real part since the left hand side is real, and then in the last equality we used the cyclicity of the trace to introduce $\J_B\big|_{\rho(0)}$.
	Since this relation holds for arbitrary $X$, one can impose the equality at the operator level, giving:
	\begin{align}
		\frac{\partial\rho(\theta)}{\partial\theta}\bigg|_{\theta=0} =\J_B\big|_{\rho(0)}[L_B]= \frac{1}{2}\norbra{L_B\rho(0) + \rho(0)L_B} \,.\label{cf:eq:103}
	\end{align}
	This expression makes it clear where the name for SLD comes from, as $L_B$ can be thought as a symmetric generalisation of the differential operator to non-commutative variables. We are now ready to prove the Cramer-Rao bound: this is a fundamental limit on how small the variance of locally unbiased operators can be. Writing it down explicitly, we have:
	\begin{align}
		&\Tr{\rho(0) A^2} =\, \Tr{A\,\J_B\big|_{\rho(0)}[A]}\geq \frac{\left| \Tr{A\,\J_B\big|_{\rho(0)}[L_B]}\right|^2 }{\Tr{L_B\,\J_B\big|_{\rho(0)}[L_B]}} = \frac{1}{\Tr{\partial_\theta \rho(\theta)\,\J^{-1}_B\big|_{\rho(\theta)}[\partial_\theta \rho(\theta)]}\big|_{\theta=0}}\,,\label{cf:eq:cramerRao130}
	\end{align}
	where the inequality is a simple application of  Cauchy-Schwartz for $\J_B\big|_{\rho(0)}$, and in the last step we used the definition of local unbiased operators (Eq.~\eqref{cf:eq:localUnbiasedOperators}) and inverted Eq.~\eqref{cf:eq:103}. This shows that the ability to estimate the parameter $\theta$ is intrinsically connected to the statistical difference between $\rho(0)$ and $\rho(0+\partial\theta)$. 
	
	It should be pointed out that the steps presented above can in principle be replicated for  any other quantum Fisher information. In fact, it is sufficient to define the generalised derivative $L_f$ as:
	\begin{align}
		\frac{\partial\rho(\theta)}{\partial\theta}\bigg|_{\theta=0} = \J_f\big|_{\rho(0)} [L_f]\,,\label{exponentialFamily}
	\end{align}
	to obtain:
	\begin{align}
		\Tr{A\,\J_f\big|_{\rho(0)} [A]} \geq \frac{1}{\Tr{L_f\,\J_f\big|_{\rho(0)} [L_f]} } = \frac{1}{\Tr{\partial_\theta \rho(\theta)\,\J_f^{-1}\big|_{\rho(0)} [\partial_\theta \rho(\theta)]}\big|_{\theta=0} }\,.\label{cf:eq:132b}
	\end{align}
	This procedure gives a whole family of bounds. Whereas for the Bures case the variance of an observable is a quite standard object, the problem here is to find an operational interpretation to the generalised variance on the left in the above equation. Still, this can be done in some contexts, as for example for the relative entropy, Sec.~\ref{relativeEntropy}.
	
	\subsection{The Heinz family}\label{Heinz}
	This family of standard monotone functions is given by: 
	\begin{align}
		f_{\gamma_>}(x) := \frac{x^\gamma+x^{1-\gamma}}{2}\,.\label{cf:eq:xGamma}
	\end{align}	
	Thanks to the  Löwner-Heinz inequality we know that $x^\gamma$ is matrix monotone for all $\gamma\in[0,1]$~\cite{hiaiIntroductionMatrixAnalysis2014}, which means that $f_{\gamma_>}(x)$ is matrix monotone as well. Hence, we only consider this range for the parameter $\gamma$. Then, using the transformation $L$ from \hyperref[box:contrastAdditionalProperties]{Box 10} we can also introduce the corresponding standard convex functions $g_{\gamma_>}(x)$ as:
	\begin{align}
		g_{\gamma_>}(x) = [L f_{\gamma_>}](x)= \frac{(x-1)^2}{x^\gamma+x^{1-\gamma}}\,,
	\end{align}
	for $\gamma\in [0,1]$, which give rise to the contrast functions:
	\begin{align}
		H_{\gamma_>}(\rho|| \sigma) = \Tr{(\rho-\sigma)(\LL_\sigma^\gamma\RR_\rho^{1-\gamma} +\LL_\sigma^{1-\gamma}\RR_\rho^{\gamma})^{-1}[(\rho-\sigma)]}\,,\label{cf:eq:144}
	\end{align}
	It is useful for what follows to give an integral expression to this quantity. To this end, we point out the following rewriting:
	\begin{align}
		(\LL_\sigma^\gamma\RR_\rho^{1-\gamma} +\LL_\sigma^{1-\gamma}\RR_\rho^{\gamma})[A] &= \sigma^{\gamma} \,A\,\rho^{1-\gamma} + \sigma^{1-\gamma} \,A\,\rho^{\gamma}=(\sigma^{\gamma} \,A\,\rho^\gamma)\rho^{1-2\gamma} + \sigma^{1-2\gamma} \,(\sigma^{\gamma} \,A\,\rho^\gamma)=\\
		&= (\LL_\sigma^{1-2\gamma} +\RR_\rho^{1-2\gamma})\LL_\sigma^{\gamma}\RR_\rho^\gamma[A] \,.
	\end{align}
	Then, the inverse in Eq.~\eqref{cf:eq:144} can be taken by inverting first  $(\LL_\sigma\RR_\rho)^\gamma$ (which simply gives $(\LL_\sigma\RR_\rho)^{-\gamma}$) and the superoperator $(\LL_\sigma^{1-2\gamma} +\RR_\rho^{1-2\gamma})$ independently. As the latter can be rewritten as $(\LL_{\sigma^{1-2\gamma}} +\RR_{\rho^{1-2\gamma}})$, we can use the result from App.~\ref{app:inverseAnticomm} to finally express Eq.~\eqref{cf:eq:144} as:
	\begin{align}
		H_{\gamma_>}(\rho|| \sigma) = \int_0^\infty \dt \; \Tr{(\rho-\sigma)e^{-t\sigma^{1-2\gamma}}\sigma^{-\gamma}(\rho-\sigma)\rho^{-\gamma}e^{-t\rho^{1-2\gamma}}}\,.\label{cf:eq:149x}
	\end{align}
	Once again the symmetry in the arguments is evident, and it follows from the fact that $g_{\gamma_>}$ is standard convex. Moreover, one can also apply Thm.~\ref{cf:thm:Ruskai} to the expression just obtained to derive the two Fisher operators:
	\begin{align}
		\J_{f_{\gamma_>}}\big|_{\pi}[A] &= \frac{1}{2}\,(\LL_\pi^{1-2\gamma} +\RR_\pi^{1-2\gamma})(\LL_\pi\RR_\pi)^\gamma[A] = \frac{1}{2}\, \{\pi^{1-2\gamma},\pi^\gamma \,A\,\pi^\gamma\}\,;\label{cf:eq:155x}\\
		\J^{-1}_{f_{\gamma_>}}\big|_{\pi}[A] &=2\,(\LL_\pi^{1-2\gamma} +\RR_\pi^{1-2\gamma})^{-1}(\LL_\pi\RR_\pi)^{-\gamma}[A]=\int_0^{\infty} \dt \; e^{-(t\pi^{1-2\gamma})/2}\,\pi^{-\gamma}A\,\pi^{-\gamma}\,e^{-(t\pi^{1-2\gamma})/2}\,.\label{cf:eq:154x}
	\end{align}
	Since $\J^{-1}_{f_{\gamma_>}}\big|_{\pi}$ is given in Kraus form, it is immediately clear that it is completely positive. This shows that $f_{\gamma_>}\in\mathcal{M}^{-}$ for all $\gamma\in[0,1]$. This not only provides a one-parameter family of functions in $\mathcal{M}^{-}$, but also allows to define a whole class of functions in the set. Indeed, thanks to the convex structure of $\mathcal{K}^+$ (see \hyperref[box:inverseMonotones]{Box 11}) one also has that:
	\begin{align}
		f_{\de\mu(\gamma)_>} (x) := \norbra{\int_0^1\de\mu(\gamma)\;\frac{2}{x^\gamma+x^{1-\gamma}} }^{-1}\,,\label{210}
	\end{align}
	for any probability distribution $\de\mu(\gamma)$ on $[0,1]$ is also in $\mathcal{M}^-$. This exemplifies how one can exploit convexity to pass from a one-parameter family to a much bigger class of standard monotones satisfying the same property.
	
	There are some more remarks to be made about the defining functions $f_{\gamma_>}$. First, it should be noticed that all of these functions lay above $f_{SQ}$, as $x^\gamma$ has a unique minimum as a function of $\gamma$ exactly for $\gamma = 1/2$, in which case $f_{(1/2)_>} = f_{SQ}$. Moreover, it is straightforward to verify that $f_{0_>}=f_{1_>}=f_B$, so the Heinz family interpolates between the largest and the smallest elements of $\mathcal{M}^-$ (according to $\leq$). It is then an open question whether one could obtain all the elements of $\mathcal{M}^-$ through the procedure in Eq.~\eqref{210} (it should be noticed that this would imply that $\mathcal{K}^+$ is a Bauer simplex). 
	
	Finally, we can use the integral expression for powers $\gamma\in(0,1)$:
	\begin{align}
		x^\gamma &= \frac{\sin\pi\gamma}{\pi}\int_0^\infty\de\lambda\;\lambda^{\gamma-1}\frac{x}{x+\lambda} =\frac{\sin\pi\gamma}{\pi}\int_0^1\de\lambda\;\norbra{\lambda^{\gamma-1}\frac{x}{x+\lambda}+\lambda^{-\gamma}\frac{x}{1+\lambda\,x}} \,,\label{211}
	\end{align}
	to rewrite $f_{\gamma_>}$ in terms of the extreme points (defined in Eq.~\eqref{cf:eq:81}):
	\begin{align}
		f_{\gamma_>}(x) = \int_0^1\de\lambda\norbra{\frac{\sin\pi\gamma}{\pi}\norbra{\frac{\lambda^{\gamma-1}+\lambda^{-\gamma}}{1+\lambda}}}\,f_\lambda (x)\,,\label{cf:eq:153o}
	\end{align}
	where \ms{the terms inside the parenthesis} define the measure $\de\mu_{Tf_{\gamma_>}}(\lambda)$. This equation can be used to give alternative expressions to the Fisher operator in Eq.~\eqref{cf:eq:155x}.
	
	Interestingly, we can use the map $T$ to define an similar family on $\mathcal{M}^{+}$, namely:
	\begin{align}
		f_{\gamma_<}(x) := [Tf_{\gamma_>}](x)= \frac{2x}{x^\gamma + x^{1-\gamma}}\,.\label{cf:eq:156z}
	\end{align}
	
	The standard convex functions in this case are given by:
	\begin{align}
		g_{\gamma_<}(x) = [L f_{\gamma_<}](x)= \frac{(x^\gamma+x^{1-\gamma})(x-1)^2}{4x}\,
	\end{align}
	which result in contrast functions of the form:
	\begin{align}
		H_{\gamma_>}(\rho|| \sigma) =\frac{1}{4}\; \Tr{(\rho-\sigma)(\LL_\sigma^{\gamma-1}\RR_\rho^{-\gamma} +\LL_\sigma^{-\gamma}\RR_\rho^{\gamma-1})[(\rho-\sigma)]}
	\end{align}
	Interestingly, thanks to the integral expression in Eq.~\eqref{211}, we can also express the standard convex functions as:
	\begin{align}
		g_{\gamma_<}(x) =\frac{1}{2}\;\int_0^1\de s\norbra{\frac{\sin\pi\gamma}{\pi}\norbra{\frac{s^{\gamma-1}+s^{-\gamma}}{2}}}\norbra{\frac{(x-1)^2}{x+s}+\frac{(x-1)^2}{1+s x}}\,
	\end{align}
	showing that the defining measure (see Eq.~\eqref{cf:eq:AppIntegralExpression}) in this case is given by $\de N_g(s):= \frac{\sin\pi\gamma}{\pi}\norbra{\frac{s^{\gamma-1}+s^{-\gamma}}{2}}\de s$.  
	
	Thanks to the relation between $\J_f\big|_\pi$ and $\J_{Tf}^{-1}\big|_{\pi^{-1}}$  given in Eq.~\eqref{cf:eq:TTransformApp}, we can directly deduce from Eq.~\eqref{cf:eq:155x} and Eq.~\eqref{cf:eq:154x} that:
	\begin{align}
		&\J_{f_{\gamma_<}}\big|_{\pi}[A]  = \J_{Tf_{\gamma_>}}^{-1}\big|_{\pi^{-1}}[A] = \int_0^{\infty} \dt \; e^{-(t\pi^{2\gamma-1})/2}\,\pi^{\gamma}A\pi^{\gamma}\,e^{-(t\pi^{2\gamma-1})/2}\,;\label{cf:eq:111}\\
		&\J_{f_{\gamma_<}}^{-1}\big|_{\pi}[A] =\J_{Tf_{\gamma_>}}\big|_{\pi^{-1}}[A]  =  \frac{1}{2} \,\{ \pi^{2\gamma-1}, \pi^{-\gamma} A \pi^{-\gamma} \}\,.\label{cf:eq:112}
	\end{align}
	As it was expected from the action of $T$, all $\J_{f_{\gamma_<}}$ are CP, as they are presented in Kraus form. Hence, it should be noticed that the convex hull of $f_{\gamma_<}$ is contained in $\mathcal{M}^+$, and that this family ranges from the smallest element in the set, $f_{0_<} = f_{1_<}=f_H$, to the largest, namely $f_{(1/2)_<} = f_{SQ}$. Interestingly, we can also define a family of functions in $\mathcal{M}^+$ in a way complete analogous with the procedure presented in Eq.~\eqref{210}. Indeed, since $\mathcal{M}^+$ is convex, we have that the functions defined as:
	\begin{align}
		f_{\de\mu(\gamma)_<} (x) := \int_0^1\de\mu(\gamma)\;\frac{2x}{x^\gamma+x^{1-\gamma}} \,,
	\end{align}
	are also in $\mathcal{M}^+$ for any probability distribution $\de\mu(\gamma)$ on $[0,1]$. Moreover, it is straightforward to verify that $Tf_{\de\mu(\gamma)_>} =	f_{\de\mu(\gamma)_<}$ for the same probability distribution.
	
	\subsection{The family of $\alpha$-divergences}\label{alpha}
	This family of standard monotone functions takes the form:
	\begin{align}
		f_\alpha(x) = \frac{\alpha(1-\alpha )(x-1)^2}{(1-x^\alpha)(1-x^{1-\alpha})}\,.\label{eq:alphaMonotones}
	\end{align}
	Despite the complicated expression of Eq.~\eqref{eq:alphaMonotones}, this arises from the family of functions:
	\begin{align}
		g_\alpha(x) =\frac{x^{\alpha}-1}{\alpha(\alpha -1)} \,,\label{cf:eq:159d}
	\end{align}
	which are matrix convex for $\alpha\in[-1,2]$ and give rise to the following contrast functions:
	\begin{align}
		H_\alpha (\rho|| \sigma) &= \frac{1}{\alpha(\alpha -1)} \norbra{ \Tr{\sigma^\alpha\rho^{1-\alpha}}-1}= \frac{1}{\alpha(\alpha -1)}\int_0^\alpha \de \beta \; \Tr{ \rho^{1-\beta}(\log\sigma - \log\rho)\sigma^{\beta}}\,,\label{cf:eq:hAlphaDivergence}
	\end{align}
	called $\alpha$-divergences. The integral expression in the equation above is obtained simply by differentiating with respect to $\alpha$ and integrating again. A similar family is the one of Rényi divergences, given by:
	\begin{align}
		S_\alpha (\rho|| \sigma) :&= \frac{1}{\alpha-1} \log \, \Tr{\rho^\alpha\sigma^{1-\alpha}} = \frac{1}{\alpha-1} \log \norbra{1 + \int_0^\alpha \de \beta \;\Tr{\rho^{\beta} (\log\rho - \log\sigma)\sigma^{1-\beta}}}\,.\label{cf:eq:intExpr}
	\end{align}
	It is easy to verify that the two are related by the equation:
	\begin{align}
		H_\alpha (\rho|| \sigma)  = \frac{e^{-\alpha S_{1-\alpha}(\rho||\sigma)} -1}{\alpha(\alpha -1)}\,.\label{contrastRenyi}
	\end{align}
	Interestingly, the $\alpha$-contrast functions and the corresponding Rényi divergences locally give rise to the same metric structure. Before entering into this, it should be noticed that we can find an integral expression for $g_\alpha$ of the form in Eq.~\eqref{cf:eq:gIntegralExpression}, namely:
	\begin{align}
		g_\alpha(x) &= \frac{\sin \pi \alpha}{\pi\alpha(1-\alpha)}\int_0^\infty \de s \; \frac{s^\alpha}{(1+s)^2}\,\norbra{\frac{(x-1)^2}{x+s}} - \frac{1}{\alpha-1}(1-x)=\\
		&=\frac{\sin \pi \alpha}{\pi\alpha(1-\alpha)}\norbra{\int_0^1 \de s \; \frac{s^\alpha}{(1+s)^2}\,\norbra{\frac{(x-1)^2}{x+s}}+\int_0^1 \de s \; \frac{s^{1-\alpha}}{(1+s)^2}\,\norbra{\frac{(x-1)^2}{1+s x}}} - \frac{1}{\alpha-1}(1-x)\,,\label{eq:224}
	\end{align}
	where we can ignore the extra linear term as it does not contribute to the contrast function (see Sec.~\ref{contrastFunctionsToFisher}). Indeed, once one symmetrises this function, any linear terms erase, so that we have:
	\begin{align}
		g^{\rm symm}_\alpha(x) = [Lf_\alpha](x)  =\frac{(1-x^\alpha)(1-x^{1-\alpha})}{2\alpha(1-\alpha)} = \frac{\sin \pi \alpha}{2\pi\alpha(1-\alpha)} \int_0^1 \de s \;\frac{(s^\alpha+s^{1-\alpha})}{(1+s)^2}\norbra{\frac{(x-1)^2}{x+s}+\frac{(x-1)^2}{1+s x}}\,,
	\end{align}
	where we can identify the defining measure from Eq.~\eqref{cf:eq:AppIntegralExpression} to be $\de N_g(s):= \frac{\sin \pi \alpha}{\pi\alpha(1-\alpha)}\frac{(s^\alpha+s^{1-\alpha})}{(1+s)^2}\,\de s$. It should be noticed that this expression of $\de N_g(s)$ corrects the one used in most of the literature, most probably arising from a typo in~\cite{lesniewskiMonotoneRiemannianMetrics1999}.
	
	In order to give a closed expression for $\J_\alpha^{-1}\big|_{\pi}$ it is useful to first study how one can rewrite $H^{\rm symm}_\alpha (\rho|| \sigma)$ using the integral expansion in Eq.~\eqref{cf:eq:hAlphaDivergence}. In particular, carrying out the same procedure one more time, we obtain:
	\begin{align}
		&2\alpha(\alpha -1)\,H^{\rm symm}_\alpha (\rho|| \sigma) = \norbra{ \Tr{\sigma^\alpha\rho^{1-\alpha}}+\Tr{\rho^\alpha\sigma^{1-\alpha}}-2}=\\
		&= \int_0^\alpha \de \beta \; \norbra{\Tr{ \rho^{1-\beta}(\log\sigma - \log\rho)\sigma^{\beta}} + \Tr{ \rho^{\beta}(\log\rho - \log\sigma)\sigma^{1-\beta}}}=\label{cf:eq:intStep114}\\
		&=\int_0^\alpha \de\beta \left( \int_0^\beta\de\gamma\; \Tr{\rho^{1-\gamma} (\log\sigma - \log\rho)\sigma^{\gamma} (\log\sigma - \log\rho)} -\int_\beta^1\de\gamma\; \Tr{\rho^{\gamma} (\log\rho - \log\sigma)\sigma^{1-\gamma} (\log\rho - \log\sigma)}\right)\,,\label{cf:eq:intStep130}
	\end{align}
	where the easiest way to verify the last equality is to explicitly carry out the integrals in Eq.~\eqref{cf:eq:intStep130} and see that it correctly retrieves Eq.~\eqref{cf:eq:intStep114}. Substituting $\gamma\rightarrow1-\gamma$ in the second integral of Eq.~\eqref{cf:eq:intStep130}, one finally obtains:
	\begin{align}
		2\alpha(\alpha -1)\,H^{\rm symm}_\alpha (\rho|| \sigma)&=\int_0^\alpha \de\beta \norbra{\int_0^\beta\de\gamma - \int_{0}^{1-\beta}\de\gamma} (\Tr{\rho^{1-\gamma} (\log\sigma - \log\rho)\sigma^{\gamma} (\log\sigma - \log\rho)})=\\
		&=-\int_0^\alpha \de\beta \int_{\beta}^{1-\beta}\de\gamma \;\Tr{\rho^{1-\gamma} (\log\sigma - \log\rho)\sigma^{\gamma} (\log\sigma - \log\rho)}\,.\label{cf:eq:118}
	\end{align}
	It should be noticed that the symmetry in the arguments of $H^{\rm symm}_\alpha (\rho|| \sigma)$ is reflected in the symmetry under the transformation $\alpha\rightarrow1-\alpha$. Moreover, thanks to the appearance of the two differences of logarithms in Eq.~\eqref{cf:eq:118}, it is straightforward to give the local expansion
	of the $\alpha$-divergences. In fact, denote by  $\J_L^{-1}\big|_{\pi}$ the Fréchet derivative of the logarithm, which reads in formulae~\cite{hiaiIntroductionMatrixAnalysis2014}:
	\begin{align}
		\J_L^{-1}\big|_{\pi}[A] := \lim_{\varepsilon\rightarrow0}\frac{\log(\pi + \varepsilon A) - \log\pi  }{\varepsilon} = \int_0^{\infty} \dt \; (\pi + t)^{-1} \, A \,(\pi + t)^{-1}\,. \label{eq:logExpansion}
	\end{align}
	One can use this expression to approximate the operator $\log(\pi + \varepsilon A)$ by $\log(\pi) + \varepsilon\, \J_L^{-1}\big|_{\pi}[A]$. Then, a simple substitution shows that the expansion of the $\alpha$-divergences is given by:
	\begin{align}
		H_\alpha (\pi|| \pi + \varepsilon \delta \rho) &=  \frac{S_{1-\alpha}(\pi|| \pi + \varepsilon \delta \rho)}{1-\alpha} = \frac{\varepsilon^2}{2\,\alpha(1-\alpha)}\int_0^\alpha \de \beta \int_\beta^{1-\beta}\de \gamma \;\Tr{\pi^{1-\gamma}\, \J_L^{-1}\big|_{\pi}[\delta\rho]\,\pi^\gamma \J_L^{-1}\big|_{\pi}[\delta\rho]}+\bigo{\varepsilon^3}=\\
		&= \frac{\varepsilon^2}{2\,\alpha(1-\alpha)}\int_0^\alpha \de \beta \int_\beta^{1-\beta}\de \gamma \;{\rm cov}^\gamma_\pi(\J_L^{-1}\big|_{\pi}[\delta\rho],\, \J_L^{-1}\big|_{\pi}[\delta\rho])+\bigo{\varepsilon^3}\,,\label{cf:eq:alphaDivergences}
	\end{align}
	where in the last line we implicitly defined the $\gamma$-covariance, which generically reads:
	\begin{align}
		{\rm cov}^\gamma_\pi(A,B) := \Tr{\pi^{1-\gamma} A\pi^{\gamma} B} - \Tr{ \pi\,A}\Tr{\pi\,B}\,.\label{eq:gammaCovDef}
	\end{align}
	It should be noticed that, since perturbations of states are traceless $\Tr{\pi\,\J_L^{-1}\big|_{\pi}[\delta\rho]} =\Tr{\delta\rho} = 0$. For this reason, the second term in Eq.~\eqref{eq:gammaCovDef} has zero contribution.
	
	We are now ready to apply Thm.~\ref{cf:thm:Ruskai} to Eq.~\eqref{cf:eq:alphaDivergences}, which allows to deduce that the corresponding family of quantum Fisher information operators takes the form:
	\begin{align}
		\J_\alpha^{-1}\big|_{\pi} [A]  =  \frac{1}{\,\alpha(1-\alpha)}\int_0^\alpha \de \beta \int_\beta^{1-\beta}\de \gamma \;\J_L^{-1}\big|_{\pi}\sqrbra{\,\pi^{1-\gamma}\, \J_L^{-1}\big|_{\pi}[A]\,\pi^\gamma}\,.\label{cf:eq:183x}
	\end{align}
	This expression is valid in general, i.e., for $\alpha\in[-1,2]$. Still, if one restricts to the smaller interval given by $\alpha\in(0,1)$, it was shown in~\cite{besenyei2011completely} that one can actually express $\J_\alpha^{-1}\big|_{\pi} $ as:
	\begin{align}
		\J_\alpha^{-1}\big|_{\pi} [A]=  \frac{\sin^2 (\pi \alpha)}{\,\pi^2 \alpha(1-\alpha)}\int_0^\infty \de s\int_0^\infty \de t \; s^\alpha\,t^{1-\alpha}\, (\pi+s)^{-1}(\pi+t)^{-1}A(\pi+s)^{-1}(\pi+t)^{-1}\,.
	\end{align}
	Since $\J_\alpha^{-1}\big|_{\pi}$ is in Kraus form, this shows that $f_\alpha\in\mathcal{M}^{-}$ for $\alpha\in[0,1]$. Moreover, it was also proved in~\cite{hiai2013families} that $f_\alpha\in\mathcal{M}^{+}$ for $\alpha\in[-1,-\frac{1}{2}]\cup[\frac{3}{2},2]$ (i.e., in this parameter range $\J_\alpha\big|_{\pi}$ is CP), and that for  $\alpha\in(-\frac{1}{2},0)\cup(1,\frac{3}{2})$ neither $\J_\alpha\big|_{\pi}$ nor $\J_\alpha^{-1}\big|_{\pi}$ are CP. This gives an exhaustive list of the possible behaviours of the Fisher information operators associated to $f_\alpha$. Unfortunately, we were not able to find a general expression for $\J_\alpha\big|_{\pi}$, not even in the parameter range for which it is completely positive.
	
	Before moving on to the next sections, it is interesting to connect the local behaviour of the $\alpha$-divergences to the Wigner-Yanase-Dyson skew information, which is defined by the formula:
	\begin{align}
		I^\gamma(\pi, X ) &= -\frac{1}{2}\Tr{[\pi^{\gamma}, X][\pi^{1-\gamma}, X]}=\Tr{X^2\,\pi} -\Tr{\pi^{1-\gamma}\, X\,\pi^{\gamma}\, X}\,.\label{cf:eq:wyd}
	\end{align}
	This object can be interpreted as a quantifier of the quantum uncertainty of the observable $X$ as measured in the state $\pi$~\cite{wignerINFORMATIONCONTENTSDISTRIBUTIONS1963}. Then, by adding and subtracting the variance of $\J_L^{-1}\big|_{\pi}[\delta\rho]$ to Eq.~\eqref{cf:eq:alphaDivergences}, one obtains:
	\begin{align}
		H_\alpha &(\pi|| \pi + \varepsilon \delta \rho) =\frac{\varepsilon^2}{2}  \Tr{\norbra{\J_L^{-1}\big|_{\pi}[\delta\rho]}^2\,\pi} +  \frac{\varepsilon^2}{2\,\alpha(\alpha-1)}\int_0^\alpha \de \beta \int_\beta^{1-\beta}\de \gamma \;I^\gamma(\pi, \J_L^{-1}\big|_{\pi}[\delta\rho] )\,.
	\end{align}
	This expression is particularly useful when one wants to isolate the effects of the coherences in the basis of $\pi$. Notice in fact that for a full rank state $\pi$, the Wigner-Yanase-Dyson skew information $I^\gamma(\pi, \J_L^{-1}\big|_{\pi}[\delta\rho] ) = 0$  if and only if  $[\pi,\delta\rho]=0$, while it is positive in all other cases. This means that all $\alpha$-divergences can be locally decomposed in a sum of the variance of $\J_L^{-1}\big|_{\pi}[\delta\rho]$ with a positive correction accounting for the non-commutativity between $\pi$ and $\delta\rho$. Moreover, if $\pi$ and $\delta\rho$ all the $\alpha$-divergences collapse into the same quantity (which is a direct consequence of the uniqueness of the classical Fisher information).
	
	In the next sections, we present different examples of quantum Fisher information as $\alpha$ varies. First, it should be noticed that thanks to the symmetry $f_{\alpha} = f_{1-\alpha}$, it is sufficient to characterise the interval $[0,2]$ alone, as $[0,1]$ is mapped into itself, and $[1,2]$ into $[-1,1]$, completing the range of allowed parameters. Then, for $\alpha\in[0,1/2]$ the value $f_\alpha(x)$ is monotonically increasing in $\alpha$, whereas for $\alpha\in[1/2,2]$ this behaviour inverts, and $f_\alpha(x)$ becomes monotonically decreasing. There are three limits that are notable enough to deserve a name: the Fisher information associated with the relative entropy, obtained in the limit $\alpha\rightarrow0$, called Kubo-Mori-Bogoliubov (KMB) inner product; the one given by the limit $\alpha\rightarrow1/2$, i.e, the largest element of the family according to $\leq$, called the Wigner-Yanase metric; and finally, as $\alpha\rightarrow2$ one gets the minimal function $f_H(x) = 2x/(x+1)$, called the harmonic mean.
	
	\subsection{The Wigner-Yanase skew information ($\alpha = 1/2$)}\label{wignerYanase}
	The first case we consider is the family of $\alpha$-divergence is the one of the Wigner-Yanase metric, corresponding to the case $\alpha = \frac{1}{2}$. In this context, the convex function in Eq.~\eqref{cf:eq:159d} and the corresponding divergence are given by:
	\begin{align}
		\label{eq:WY_contrast}
		g_{WY} (x)=  4(1-\sqrt{x})\,,\qquad H_{WY} (\rho|| \sigma)  =  4 ( 1- \Tr{\sqrt{\rho}\,\sqrt{\sigma}})\,,
	\end{align}
	while the standard monotone function takes the particularly simple form:
	\begin{align}
		f_{WY} (x)=  \norbra{\frac{1+\sqrt{x}}{2}}^2\,.
	\end{align}
	By comparing this expression with $f_B$, it can be easily verified that $f_{WY} (x)\equiv \norbra{f_B(\sqrt{x})}^2$. This relation allows to compute the quantum Fisher operators for the Wigner-Yanase metric directly from the ones for the Bures. Indeed, it follows from the straightforward manipulations:
	\begin{align}
		\J_{WY}\big|_{\pi}[A] = \RR_\pi \,f_{WY}(\LL_\pi\RR_\pi^{-1})[A] = \RR_{\sqrt{\pi}}^2\, f_{B}(\LL_{\sqrt{\pi}}\RR_{\sqrt{\pi}}^{-1})^2[A] = \J_{B}\big|_{\sqrt{\pi}}[\J_{B}\big|_{\sqrt{\pi}}[A]]\,,\label{245}
	\end{align}
	that one can rewrite the Wigner-Yanase Fisher operators as:
	\begin{align}
		&\J_{WY}\big|_{\pi} [A]  = \frac{1}{4}\,\{\sqrt{\pi},\{\sqrt{\pi}, A\}\}\,;
		\qquad\qquad\J_{WY}^{-1}\big|_{\pi} [A]  = \int_0^{\infty} \dt \int_0^{\infty} \de s \; e^{-(t+s)\sqrt{\pi}/2} \, A \,e^{-(t+s)\sqrt{\pi}/2}\,.
	\end{align}
	The identification in Eq.~\eqref{245} also allows to explicitly compute the geodesics associated to the Wigner-Yanase metric. The construction needed is completely analogous to the one for classical states~\cite{bengtssonGeometryQuantumStates2017}, which we briefly sketch  for completeness. Denote the set of Hermitian matrices of dimension $n$ by $\mathcal{H}_n$. We then define the map\footnote{This corresponds to a global section in the purification bundle discussed in the context of the Bures metric (see Sec.~\ref{Bures}).} $S: \mathcal{S}_n\rightarrow\mathcal{H}_n$ associating to each state $\pi$ its unique positive square root $\sqrt{\pi}$. The target space is naturally endowed with the Hilbert-Schmidt scalar product, so that the image of $\mathcal{S}_d$ can be characterised by the equation:
	\begin{align}
		S(\mathcal{S}_n) := \curbra{X\in\mathcal{H}_n \;|\; X\geq 0 \;\;\land\;\;\Tr{X^2} = 1}\,,\label{247}
	\end{align}
	which means that $S(\mathcal{S}_n)$ is just given by the positive octant of a $(d^2-1)$-sphere in $\mathcal{H}_n$. The geodesic distance in this context is well known: since geodesics are given by great circles, the geodesic distance simply coincides with the angle $\theta$ they subsume. Moreover, since on  $X,Y\in S(\mathcal{S}_d)$ it holds that $\cos\theta = \Tr{XY}$, we finally obtain:
	\begin{align}
		d_{S(\mathcal{S}_n) } = \arccos \Tr{X Y}\,.
	\end{align}
	Given the simplicity of this construction, it is interesting to consider what is the pullback of the Hilbert-Schmidt metric on $\mathcal{S}_n$. To this end, we need to compute the differential of the map $S$, which is defined by:
	\begin{align}
		\de S\big|_\pi[\delta\rho]  := \frac{\de}{\de\varepsilon}\,\sqrt{\pi+\varepsilon\,\delta\rho}\;\Big|_{\varepsilon=0}\,.\label{248}
	\end{align}
	This can be computed by noticing that:
	\begin{align}
		\norbra{\de S\big|_\pi[\delta\rho]}\,\sqrt{\pi} + \sqrt{\pi} \,\norbra{\de S\big|_\pi[\delta\rho]}= \frac{\de}{\de\varepsilon}\norbra{\sqrt{\pi+\varepsilon\,\delta\rho}\,\sqrt{\pi+\varepsilon\,\delta\rho}} \big|_{\varepsilon=0}= \frac{\de}{\de\varepsilon}\norbra{\pi+\varepsilon\,\delta\rho}\big|_{\varepsilon=0} = \delta\rho\,,
	\end{align}
	where the first equality can be explicitly verified from the definition in Eq.~\eqref{248}. Then, since the equation above can be rewritten as $\de S\big|_\pi[\delta\rho] = \frac{1}{2} \,\J^{-1}_{B}\big|_{\sqrt{\pi}}[\delta\rho]$, we also obtain that the pullback of the Hilbert-Schmidt metric takes the form:
	\begin{align}
		\Tr{\de S\big|_\pi[A]\,\de S\big|_\pi[B]} = \frac{1}{4}\,\Tr{\J^{-1}_{B}\big|_{\sqrt{\pi}}[A]\,\J^{-1}_{B}\big|_{\sqrt{\pi}}[B]} = \frac{1}{4}\,\Tr{A\,\J^{-1}_{WY}\big|_{\pi}[B]}\,.
	\end{align}
	This proves a remarkable fact: that, up to a factor $\frac{1}{4}$, the Wigner-Yanase skew information is the pullback of the Hilbert-Schmidt metric by the root map $S$. This means that, thanks to Eq.~\eqref{247},  the space $\mathcal{S}_n$ with the metric $\J_{WY}^{-1}\big|_{\pi}$  is isometric\footnote{This was already noticed in~\cite{gibiliscoWignerYanaseInformation2003}, where it was pointed out that the Wigner-Yanase metric is the only known Fisher information that has a constant positive curvature, a property that uniquely identifies subsets of hyperspheres.} to an $n$-sphere of radius 2.  Thus, one can give a closed form for the geodesic distance, given by:
	\begin{align}
		d_{WY}(\rho,\sigma) = 2\,\arccos \Tr{\sqrt{\rho}\sqrt{\sigma}}\,,\label{cf:eq:wyGeodesicDistance}
	\end{align}
	and a simple expression for the geodesic path connecting any two density matrices $\rho$ and $\sigma$, namely:
	\begin{align}
		\gamma_{WY}^{\rho\rightarrow\sigma} (t) = \frac{\norbra{(1-t)\,\sqrt{\rho}+t\,\sqrt{\sigma}}^2}{\Tr{\norbra{(1-t)\,\sqrt{\rho}+t\,\sqrt{\sigma}}^2}}\,.\label{cf:eq:Geodesic}
	\end{align}
	It is interesting to compare Eq.~\eqref{cf:eq:wyGeodesicDistance} with the Bures geodesic distance obtained in Eq.~\eqref{cf:eq:buresDistance}: these two quantities coincide for commuting states, whereas in general one has $d_{B}(\rho,\sigma)\leq d_{WY}(\rho,\sigma)$, due to the inequality $\Tr{\sqrt{\rho}\sqrt{\sigma}}\leq F(\rho,\sigma)$~\cite{friedland1994product} and the fact that the arccosine is monotonically decreasing. To the best of the authors' knowledge, these are the only two cases for which one has an analytical expression for the geodesic distance. 
	
	Another important property of this metric is that it can be used to express the quantum Chernoff bound~\cite{audenaert2007discriminating}. This arises in the following setting: consider the task of distinguishing two different states $\rho_0$ and $\rho_1$, knowing that each is prepared with a probability $p_0$ and $p_1$. In this context, the symmetric distinguishability
	problem consists in finding a POVM (positive operator-valued measure) $\{E_0,E_1\}$ such that the probability of error $P_e:= p_0 \Tr{E_1\rho_0}+p_1 \Tr{E_0\rho_1}$ is minimal. By defining the positive and negative part of a Hermitian operator as $A_{\pm}:= (|A|\pm A)/2$, one can prove that the optimal measurement is obtained by setting $E_1$ to be the projector on the range of $(p_1\rho_1-p_0\rho_0)_+$, yielding the following expression for the minimum error probability~\cite{audenaert2007discriminating}:
	\begin{align}
		P_{e,min} = \frac{1}{2}(1-\Tr{|p_1\rho_1-p_0\rho_0|})\,.
	\end{align}
	This discussion was done in the single copy scenario. If one allows more copies of $\rho_{0/1}$ to be prepared at the same time with the probability $p_{0/1}$, one can again infer that the optimal error probability is given by:
	\begin{align}
		P_{e,min,n} = \frac{1}{2}(1-\Tr{|p_1\rho_1^{\otimes n}-p_0\rho_0^{\otimes n}|})\,.
	\end{align}
	Differently from what happened for the single copy scenario, though, this probability scales with $n$, and in particular it asymptotically decreases as $P_{e,min,n}\simeq e^{-\xi_{QCB} n}$ for $n\gg1$. In~\cite{audenaert2007discriminating} it was proved that the exponent takes the form:
	\begin{align}
		\xi_{QCB} := -\lim_{n\rightarrow\infty} \, \frac{\log P_{e,min,n}}{n} = \max_{0\leq s\leq 1}(-\log \Tr{\rho_0^{s}\rho_1^{1-s}})\,.\label{cf:eq:chernoffBound}
	\end{align}
	This result goes under the name of quantum Chernoff bound. The position at which the maximum is found usually depends on the particular form of $\rho_0$ and $\rho_1$. Still, if one restricts to the  case in which $\rho_1 = \rho_0+\delta\rho$, with $\delta\rho\ll1$, then one can apply the methods from Sec.~\ref{alpha} to express $\xi_{QCB}$ in terms of $\J_\alpha^{-1}\big|_{\rho_0}$. In this context, the unique maximum is attained for $s=\frac{1}{2}$, meaning that:
	\begin{align}
		\xi_{QCB} = \frac{1}{8}\Tr{\delta\rho \, \J_{WY}^{-1}\big|_{\rho_0} [ \delta\rho] }\,.\label{chernoffBoundLocal}
	\end{align}
	This further motivates the interest in the Wigner-Yanase metric.
	
	\subsection{The relative entropy ($\alpha=0$)}\label{relativeEntropy}
	The most renowned among the $\alpha$-divergences, and among the contrast functions in general, is the one obtained in the limit $\alpha\rightarrow0$, namely the relative entropy. In fact, carrying out the limit of 
	Eq.~\eqref{cf:eq:alphaDivergences} one gets:
	\begin{align}
		g_L(x) := \lim_{\alpha\rightarrow0}\; \frac{1}{\alpha(\alpha -1)}\int_0^\alpha \de \beta \; (x^ \beta \, \log x) = -\log x \,.
	\end{align}
	The corresponding contrast function takes the familiar form:
	\begin{align}
		S(\rho||\sigma) := H_L (\rho|| \sigma)  =\Tr{\rho \,(\log\rho-\log\sigma)}  \,.\label{cf:eq:relativeEntropy}
	\end{align}
	Moreover, its symmetrised version has the following integral expression:
	\begin{align}
		\frac{g_L(x)+x \,g_L(x^{-1})}{2}= \frac{1}{2}\int_0^1 \de s\;\frac{1}{(1+s)}\,\norbra{\frac{(x-1)^2}{x+s}+\frac{(x-1)^2}{1+s x}}\,,
	\end{align}
	allowing to identify $\de N_L(s) := \frac{1}{(1+s)} \de s$. This divergence has the special property that it is additive on tensor products:
	\begin{align}\label{cf:eq:tensorLog}
		S(\rho_A\otimes\rho_B||\sigma_A\otimes\sigma_B) = 	S(\rho_A||\sigma_A) + 	S(\rho_B||\sigma_B)\,,
	\end{align}
	whereas in general one can just prove that $H_g (\rho\otimes\tau||\sigma\otimes\tau)= H_g (\rho||\sigma)$ (see Eq.~\eqref{cf:eq:131n}).
	The standard monotone function in this case is given by:
	\begin{align}
		f_L(x) = \frac{x-1}{\log x} = \int_0^1 \de\gamma \; x^\gamma = \int_0^1 \de\gamma \; f_{\gamma_>}(x)\,,\label{eq:253}
	\end{align}
	where the last equality shows how $f_L$ can be defined as a uniform mixture of $f_{\gamma_>}$. Moreover, this integral expression allows the immediate calculation of $\J_L\big|_{\pi}$ as:
	\begin{align}
		&\J_L\big|_{\pi}[A]  = \RR_\pi \int_0^1 \de\gamma \;(\LL_\pi\RR_\pi^{-1})^\gamma[A]= \int_0^{1} \de\gamma\; \pi^{\gamma}  A \,\pi^{1-\gamma} \,.\label{cf:eq:appJL}
	\end{align}
	Interestingly, this superoperator is the same one gets from the Dyson series of the exponential, i.e., $e^{\log(\pi)+\varepsilon A}\simeq\pi + \varepsilon\,\J_L\big|_{\pi}[A]$~\cite{hiaiIntroductionMatrixAnalysis2014}. Then, thanks to this identification, we can deduce that $\J^{-1}_L[A]$ will be given by the first term in the expansion of $\log(\pi + \varepsilon A)$. Indeed, this follows from the equalities:
	\begin{align}
		\pi +\varepsilon A = e^{\log(\pi + \varepsilon A)} = e^{\log(\pi) + \varepsilon\, \J_L^{-1}\big|_{\pi}[A]} = \pi + \varepsilon\,\J_L\big|_{\pi}\J_L^{-1}\big|_{\pi}[A]	+\bigo{\varepsilon^2}\,,
	\end{align}
	where we expanded at first order in $\varepsilon$ and used the derivative of the logarithm defined in Eq.~\eqref{eq:logExpansion}. Thus, it follows that $\J^{-1}_L\big|_{\pi}$ coincides with the Fréchet derivative of the logarithm, i.e.:
	\begin{align}
		\J^{-1}_L\big|_{\pi}[A]  = \int_0^{\infty} \dt \; (\pi + t)^{-1} A \,(\pi + t)^{-1}\,.\label{cf:eq:appJL1}
	\end{align}
	Since $\J^{-1}_L\big|_{\pi}$ is in Kraus form, it follows that $f_L\in\mathcal{M}^{-}$. It should be noticed that this was not clear a priori, as $f_L$ is obtained from the convex combination of elements in $\mathcal{M}^{-}$, which is not a convex set~\cite{hiai2013families}: indeed, one should compare the expression for the subfamily of $\mathcal{M}^-$ generated by $f_{\gamma_>}$ in Eq.~\eqref{210}, with the one given in Eq.~\eqref{eq:253}.
	
	Interestingly, one can express the relative entropy in terms of the corresponding Fisher information operator. Indeed, define the path $\gamma(t)$ connecting $\rho$ to $\sigma$ as $\gamma(t):= (1-t)\rho + t\,\sigma$. Then, the following equalities hold:
	\begin{align}
		S(\rho||\sigma)  &= S(\rho||\gamma(1))-S(\rho||\gamma(0)) = \int_0^1\dt\; \norbra{\frac{\de}{\dt}\, S(\rho||\gamma(t))} = -\int_0^1\dt\;\norbra{\frac{\de}{\dt}\Tr{\rho \log\gamma(t)}} = \\
		&=\int_0^1\dt\;\Tr{\rho\, \J_L^{-1}\big|_{\gamma(t)}[(\rho-\sigma)]} \,.
	\end{align}
	This identity is particularly useful when one needs to compare global and local behaviours, as it allows to express the contrast function with the same operator appearing in the definition of the Fisher information~\cite{hiai2016contraction}.
	
	The scalar product defined by $\J_L\big|_\pi$ on the space of observables is called Kubo-Mori-Bogoliubov (KMB) inner product, which is of key importance in the context of linear response theory of thermal states. Indeed, using the definition of Gibbs states $\pi (H) := \frac{e^{-\beta H}}{\mathcal{Z}(H)}$, with $\mathcal{Z}(H) := \Tr{e^{-\beta H}}$ denoting the partition function, one can derive the following identity from the Dyson series in Eq.~\eqref{cf:eq:appJL}:
	\begin{align}\label{cf:eq:diffFreeEn}
		\frac{\partial^2}{\partial x \partial y}\,\log\mathcal{Z}(H+x\,A+y\,B)\bigg|_{x=y=0} = \beta^2\,\Tr{\Delta_{\pi (H)} A\; \J_L\big|_{\pi (H)} [\Delta_{\pi (H)}B]}\,.
	\end{align}
	It should be noticed that in the right hand side of the equation we implicitly defined $\Delta_{\pi(H)}X : = X -\id \Tr{X\,\pi (H)}$. When $A=B$,  Eq.~\eqref{cf:eq:diffFreeEn} describes the (quantum) thermal fluctuations of the operator $A$. Moreover, this expression is connected to transport coefficients in the linear response by the fluctuation-dissipation relation (see~\cite{petzBogoliubovInnerProduct1993c,altland2010condensed} for more details), showing the physical significance of $\J_L\big|_\pi$. On top of this, we can also give an interpretation to the Cramer-Rao bound of the type in Eq.~\eqref{cf:eq:132b} in this context. Indeed, suppose we have a Hamiltonian $H(\theta)$ depending on some unknown parameter $\theta$ that we want to estimate. To this end, we use a locally unbiased estimator $A$, defined by the condition:
	\begin{align}
		\frac{\partial}{\partial\theta}\,\Tr{\pi_\beta (H(\theta))\, A}\bigg|_{\theta=0} = 1\,.
	\end{align}
	Moreover, it is useful to give an explicit form to the generalised derivative $L_f$ (defined in Eq.~\eqref{exponentialFamily}) in this context, which reads.,
	\begin{align}
		\frac{\partial}{\partial\theta}\pi(H(\theta))\bigg|_{\theta=0} = \J_L\big|_{\pi(H(0))} [L_L]\,.\label{eq:261}
	\end{align}
	Explicitly carrying out the differentiation on the left, one can verify that $L_L= -\beta \Delta_{\pi (H(0))} \dot H$, where we slightly abused the notation to define $ \dot H:= \partial_\theta H(\theta)\big|_{\theta=0}$.
	Then, using the generalised Cramer-Rao bound in Eq.~\eqref{cf:eq:132b}, together with the expression in Eq.~\eqref{cf:eq:diffFreeEn}, we have the following inequality:
	\begin{align}
		\frac{\partial^2}{\partial x^2}\,\log\mathcal{Z}(H(0)+x\,A)\bigg|_{x=0}&\geq\frac{1}{\Tr{\Delta_{\pi (H(0))} \dot H\,\J_L\big|_{\pi (H(0))} [\Delta_{\pi (H(0))} \dot H]}\bigg|_{\theta=0} } = \\
		&= \beta^2\,\norbra{\frac{\partial^2}{\partial \theta^2}\log\mathcal{Z}(H(0)+\theta\,\dot H)}^{-1}\bigg|_{\theta=0}\,.\label{cf:eq:181}
	\end{align}
	This equation gives us an interesting result: the minimum fluctuations of an unbiased estimator cannot be smaller than the inverse of the fluctuations of the operator inducing the change in the Hamiltonian (namely, $\dot H$). Moreover, in order for the bound to be saturated, one should choose $A$ to be parallel to $\dot H$. A practical application of this result can be seen in the context of thermometry~\cite{mehboudi2019thermometry}: in there, the optimal measurement turns out to be the one of the energy (since the operator associated to the variation of $\beta$ is exactly the Hamiltonian), and the fluctuations in the estimator are bounded by the inverse of the heat capacity. It directly follows, then, that in order to have a good estimation of the temperature one needs to choose a thermal state with the heat capacity as big as possible.
	
	\subsection{The quantum information variance}\label{cf:sec:relEntVariance}
	The function that we consider in this case is not in the family of $\alpha$-divergences, but it is closely related to them:
	\begin{align}
		f_V(x) = \frac{2\, (x-1)^2}{(x+1)(\log x)^2}\,.\label{cf:eq:fVariance}
	\end{align}
	Indeed, this standard monotone is associated to the following matrix convex function:
	\begin{align}
		g_V(x)  = \frac{1}{2}\frac{\partial^2}{\partial\alpha^2}\, \norbra{\alpha(\alpha -1)g_\alpha(x)}\bigg|_{\alpha=0} = \frac{1}{2}\frac{\partial^2}{\partial\alpha^2}\,x^{\alpha}\bigg|_{\alpha=0}= \frac{1}{2}\,(\log x)^2\,,\label{265x}
	\end{align}
	which gives rise to what is called the quantum information variance, given by:
	\begin{align}
		H_V (\rho|| \sigma) = \frac{1}{2} \, \Tr{\rho \,(\log\rho-\log\sigma)^2}\,.\label{eq:qInfoVariance}
	\end{align}
	It was shown in~\cite{tomamichel2013hierarchy} that this quantity could be interpreted as a quantifier of the fluctuations in the distinguishability of $\rho$ with respect to $\sigma$.
	
	Differentiating Eq.~\eqref{eq:224} we can also obtain the integral expression of $g_V$ as:
	\begin{align}
		g_V(x) = \int_0^\infty \de s \; \frac{(-\log s)}{(1+s)^2}\,\norbra{\frac{(x-1)^2}{x+s}} = \int_0^1 \de s \; \frac{(-\log s)}{(1+s)^2}\,\norbra{\frac{(x-1)^2}{x+s}}+\int_0^1 \de s \; \frac{s\,\log s}{(1+s)^2}\,\norbra{\frac{(x-1)^2}{1+s x}}\,.\label{cf:eq:gVIntegral}
	\end{align}
	Interestingly, we can express the function in Eq.~\eqref{cf:eq:fVariance} as $f_V(x) = f_L(x)^2/f_B(x)$. Similarly to what happened for the Wigner-Yanase skew information (see Sec.~\ref{wignerYanase}), this identification allows us to directly compute the Fisher operators in terms of the one associated to the Bures metric and to the relative entropy. Indeed, these take the form:
	\begin{align}
		\J_V\big|_{\pi} [A] =\J_B^{-1}\big|_{\pi}\,\J_L^{2}\big|_{\pi} [A] = \int_0^{\infty} \dt \; e^{-t\pi/2} \, \J_L^2[A] \,e^{-t\pi/2}\,; \qquad\qquad\J_V^{-1}\big|_{\pi} [A] =\J_B\big|_{\pi}\,\J_L^{-2}\big|_{\pi}[A] = \frac{1}{2}\,\J_L^{-2}\sqrbra{\{\pi,A\}}\,.
	\end{align}
	Even if $\J_B^{-1}\big|_{\pi}$ is completely positive, this is not the case for $\J_L\big|_{\pi}$. Thus $\J_V\big|_{\pi} $ is not CP in general, and $f_V\notin\mathcal{M}^+$. A similar reasoning also holds for $\J_V^{-1}\big|_{\pi}$, proving that $f_V\notin\mathcal{M}^-$ as well. Indeed, one could reach this conclusion by noticing neither $f_V\leq f_{SQ}$ nor $f_{SQ}\leq f_{V}$ hold in this case (see Fig.~\ref{cf:fig:figstandardmonotones2}), which are necessary conditions for the membership in $\mathcal{M}^+$ or $\mathcal{M}^-$.
	
	The formal similarity with the Bures metric allows for the following application. Suppose we want to estimate a parameter $\theta$ encoded in the Hamiltonian of a thermal state. Then, it follows from Eq.~\eqref{cf:eq:cramerRao130} that the variance of a locally unbiased operator $A$ can be bounded as:
	\begin{align}
		\Tr{\pi(H(0)) A^2} \geq\frac{1}{\Tr{\partial_\theta \pi(H(\theta))\,\J_B^{-1}\big|_{\pi(H(\theta))}[\partial_\theta \pi(H(\theta))]}\Big|_{\theta=0}}\,.\label{269}
	\end{align}
	At this point, applying Eq.~\eqref{eq:261}, together with the fact that $L_L = -\beta \Delta_{\pi (H(0))} \dot H$,  we obtain~\cite{abiuso2024fundamental}:
	\begin{align}
		\Tr{\pi(H(0)) A^2} \geq\frac{1}{\beta^2\,\Tr{\Delta_{\pi (H(0))} \dot H\,\J_V\big|_{\pi(H(0))}[\Delta_{\pi (H(0))} \dot H]}}\,,\label{cramerRaoFluc}
	\end{align}
	Thus, the Cramer-Rao bound associated to the Bures metric when considering the variation of states (the right hand side of Eq.~\eqref{269}), translates to a bound involving the variation of the Hamiltonian, expressed in terms of $\J_V\big|_{\pi}$.
	
	\subsection{The geometric mean}\label{geometricMean}
	In the set of standard monotone functions a special role is taken by the square root, i.e.:
	\begin{align}
		f_{SQ}(x) = \sqrt{x}\,.
	\end{align}
	Indeed, as it was discussed in Sec.~\ref{sec:CPfisher}, the two sets $\mathcal{M}^+$ and $\mathcal{M}^-$ can be defined in terms of $f_{SQ}$. Moreover, it is the only element in the intersection 
	$\mathcal{M}^+\cap\mathcal{M}^-$, i.e., the only case in which both $\J_{SQ}\big|_{\pi}$ and $\J_{SQ}^{-1}\big|_{\pi}$ are CP. Some consequences of this property were explored in \hyperref[box:contractionCoefficients]{Box 9} in the context of recovery maps.
	
	The corresponding convex function is given by:
	\begin{align}
		g_{SQ}(x) = \sqrt{x^{-1}}-\sqrt{x}\,,
	\end{align}
	which gives rise to the contrast function:
	\begin{align}
		H_{SQ} (\rho|| \sigma)  = \Tr{\sqrt{\rho}\,(\rho-\sigma)\sqrt{\sigma^{-1}}}\,.
	\end{align}
	Its symmetrised version has the following integral expression:
	\begin{align}
		g_{SQ}^{\rm symm} (x) =\frac{(x-1)^2}{2\sqrt{x}} = \frac{1}{2}\int_0^1 \de s\; \frac{1}{\pi\sqrt{s}} \;\norbra{\frac{(x-1)^2}{x+s}+\frac{(x-1)^2}{1+s x}}\,,
	\end{align}
	so that we can identify $\de N_g(s) :=\frac{1}{\pi\sqrt{s}} $. 
	
	Finally, the Fisher operators in this case are given by:
	\begin{align}
		\J_{SQ}\big|_{\pi}[A]  = \sqrt{\pi} \, A \,\sqrt{\pi}\,; \qquad\qquad\qquad\J_{SQ}^{-1}\big|_{\pi}[A]  = \sqrt{\pi^{-1}}\, A \,\sqrt{\pi^{-1}}\,,
	\end{align}
	which explicitly shows that both $\J_{SQ}\big|_{\pi}$ and $\J_{SQ}^{-1}\big|_{\pi}$ are CP, as they are given in Kraus form.
	
	\subsection{The harmonic mean ($\alpha=2$)}\label{harmonic}
	We close the review with the smallest among the standard monotone functions, namely the harmonic mean:
	\begin{align}
		f_H(x) = \frac{2x}{x+1}\,.
	\end{align}
	This is part of the family of $\alpha$-divergences, and in particular it is obtained in the limit $\alpha\rightarrow2$. Then, the corresponding convex function is given by:
	\begin{align}
		g_H(x) = \frac{x^{2}-1}{2} = \frac{(x-1)^{2}}{2}+(x-1)\,.
	\end{align}
	As usual, we can discard the linear contribution. Then, its symmetrised version has the following integral expression: 
	\begin{align}
		g_H^{\rm symm}(x)&= \frac{1}{2}\norbra{\frac{(x-1)^{2}}{2}+ \frac{(x-1)^{2}}{2x}} =\frac{1}{2}\int_0^1 \de N_{H}(s)\;\norbra{\frac{(x-1)^2}{x+s}+\frac{(x-1)^2}{1+s x}}\,,
	\end{align}
	with $\de N_{H}(s) = \delta(s)/2$. The corresponding  contrast function reads:
	\begin{align}
		H_{H} ( \rho || \sigma)  =  \frac{1}{2}\,\Tr{(\rho-\sigma)\rho^{-1}(\rho-\sigma)} = \frac{1}{2}\,\norbra{\Tr{\sigma^2\rho^{-1}}-1}\,.
	\end{align}
	It was shown in \hyperref[box:statisticalQuantifiers]{Box 3} that this is the only contrast function that can be expressed in terms of a $\chi^2_f$-divergence, namely for $f = f_H$. Then, from Eq.~\eqref{eq:165} it follows that:
	\begin{align}
		H_g^{\rm symm} ( \rho || \sigma) \leq H_{H}^{\rm symm} ( \rho || \sigma) =  \frac{1}{2}\,\norbra{\chi^2_{f_H}(\rho||\sigma)+\chi^2_{f_H}(\sigma||\rho)}\,.
	\end{align}
	
	The Fisher information operators can be expressed in terms of the Bures ones, since $f_H = Tf_B$. Then,  by applying the relation in Eq.~\eqref{cf:eq:TTransformApp} one directly obtains:
	\begin{align}
		&\J_H\big|_{\pi}[A]  = \int_0^{\infty} \dt \; e^{-t\pi^{-1}/2} \, A \,e^{-t\pi^{-1}/2}\,; \qquad\qquad\qquad\J_H^{-1}\big|_{\pi}[A]  = \frac{1}{2}\{\pi^{-1}, A\}\,.
	\end{align}
	Since $\J_H\big|_{\pi}$ is in Kraus form, it follows that $f_H\in\mathcal{M}^+$.
	
	This concludes the survey of the quantum Fisher information metrics.
	
	\section{Conclusions and open questions}\label{conclusions}

	

	Fisher information is ubiquitous in classical and quantum physics, as it quantifies key figures of merit in different fields - parameter estimation precision in metrology, states' distinguishability in information theory, work dissipation and fluctuations in thermodynamics, to name a few. 
	This fact alone, that the same quantity fundamentally characterises such variety of distinct physical scenarios, looks almost miraculous.
	From a completely different angle, the classical Chentsov theorem guarantees that the Fisher metric is, for the set of physical states, the unique Riemannian metric that contracts under physical evolutions.
	These remarkable properties make the case for a detailed study of Fisher information, both in the classical and quantum domain, where
	there exists a family of such quantum Fisher informations (Eq.~\eqref{cf:eq:monotoneMetrics}), characterised by the Petz theorem (Thm.~\ref{cf:thm:Petz}). All the elements of this family 
	collapse to the classical Fisher information in the case of commutative diagonal states, while in the presence of off-diagonal contributions they present a rich phenomenology. This wide range of possible behaviours makes the Fisher operators mathematically interesting objects with different operational properties (cf. Sec.~\ref{gardenFisher}).
	
	In this work we reviewed, systematized and introduced new results regarding the dynamical and mathematical properties of the quantum Fisher informations. From a statistical point of view, these can be thought as the local expansion of contrast functions, which are introduced in Sec.~\ref{contrastFunctionsToFisher} together with a self-contained historical introduction to the topic. From the geometrical point of view, a natural perspective is that of studying the Fisher information metrics in Eq.~\eqref{cf:eq:monotoneMetrics}, their properties as scalar products and their interplay with the action of completely positive linear maps (i.e., physical evolutions). 
	In fact, an attentive look at the contraction properties of the Fisher information on the set of physical states unveils the deep connection between physical dynamics and such geometric structures.
	So much so that one can even define physical evolutions as exactly the ones that contract the Fisher metric, as showed in our Theorem~\ref{cf:thm:contractIf_P} and its corollary, which can be considered as a dual of the Chentsov-Petz theorem (also see the related work~\cite{abiuso2022characterizing,scandi2023physicality}).
	Moreover, the discussions of Sec.~\ref{FisherDynamics}, corroborate the claim that the Fisher information is an inherently dynamical quantity, a fact that is not completely acknowledged in the literature. In fact, the quantum Fisher information metrics can be used: 
	\ref{markovianity}) to fully characterize Markovianity, when identified with CP-divisibility, as well as operationally detect non-Markovian evolutions;
	\ref{retrodiction}) to generalize the notion of physical retrodiction based on Bayes/Petz maps;
	\ref{detailedBalance}) to characterize quantum microreversibility, as detailed balance becomes a geometric property of the dynamical generators.

	On a more technical side, in Sec.~\ref{sec:FisherOperatorProperties} we provided an organic discussion of many mathematical properties of the Fisher information functionals that are both relevant to the proofs, as well as retaining their own mathematical interest. In particular we focused on the characterisation of matrix monotone functions and their connection to the complete positivity of the induced functionals. 
	Finally in Sec.~\ref{gardenFisher} we detailed an extensive list of several quantum Fisher information metrics and their applications to different areas of quantum information science.
	
	Our work can therefore serve as a self-contained introduction to the topic, as well as a manual for researchers working in the related areas. 
	In particular, we gave a comprehensive review of many results scattered in the literature, that are not always easy to combine. In addition, we complemented them with different new results, most of which in Section~\ref{FisherDynamics}, as well as technical developments scattered throughout the work.
	
	The applications and mathematical properties of the Fisher information metrics are an ongoing theme of study, and a number of problems remain open. A particularly natural question is the following: once the class of Fisher $f$-metrics in Eq.~\eqref{cf:eq:monotoneMetrics} is introduced, is there a closed-form for each corresponding  $f$-geodesic distance? And what is the relation between such finite distance on the set of states and the contrast functions introduced in Sec.~\ref{contrastFunctionsToFisher}? To answer this question, one needs to solve the geodesics relative to each metrics. This is nontrivial, and an analytic answer, to the best of our knowledge, only exists for the case of the Bures metric (Sec.~\ref{Bures}) and the Wigner-Yanase metric (Sec.~\ref{wignerYanase}). The closed-form expression for any $f$ remains unknown. 
	For the interested reader, in \hyperref[box:statisticalQuantifiers]{Box 12} we show a universal mapping that might simplify the quest of such geodesics lengths.

	Finally,
	our work corroborates the identification of the Fisher information as a fundamental object describing statistical, geometric and dynamical properties of classical and quantum systems. It is then natural to wonder up to what degree these insights point towards the foundations of physics, if deeper connections have yet to be unveiled, and whether a similar approach could be used to constrain and interpret theories beyond quantum mechanics. We cannot but speculate that Fisher information tools might be key in formalising the connection of underlying informational principles to the mathematical structure of reality.

	\begin{additional-info}{Box 12. Fisher geodesics for normalised states vs general positive operators
		}\label{box:geodesics}
		
		In the quantum information literature, and throughout this work, we considered the quantum Fisher scalar product 
		$K_{f,\rho}(A,B)	 := \Tr{A\, \,\J_f^{-1}\big|_\rho [B]}$
		to induce a metric on the set of normalised quantum states $\rho$, whose tangent space is given by operators $\delta\rho$ that are Hermitian and traceless, i.e.,
		\begin{align}
			\Tr{\rho}=1\;,\ \rho\geq 0\;. 
			\qquad
			\Tr{\delta\rho}=0\;,\ \delta\rho^\dag=\delta\rho\;.
		\end{align}
		It is however straightforward to extend the scalar product to the set of all Hermitian operators. This corresponds to removing the trace constraints above and consider unnormalised states $\rho'$ with tangent vectors $\delta\rho'$ that only satisfy
		\begin{align}
			\rho'\geq 0\;,\ \delta\rho'^\dag=\delta\rho'\; .
		\end{align}
		Clearly, each unnormalised state can be decomposed in a scalar times a normalised state
		\begin{align}
			\rho'=r\rho\;, \quad r=\Tr{\rho'}\;.
		\end{align}
		When considering the Fisher metrics applied to the set of unnormalised states $\rho'$, one gets explicitly in these coordinates
		\begin{align}
			\Tr{\delta\rho' \,\J_f^{-1}\big|_{\rho'} [\delta\rho']}= 
			\frac{1}{r} \Tr{\delta\rho' \,\J_f^{-1}\big|_{\rho} [\delta\rho']}= 
			\frac{1}{r} \Tr{(\delta r\rho+r\delta\rho) \,\J_f^{-1}\big|_{\rho} [\delta r\rho+r\delta\rho]}\;.
		\end{align}
		Thanks to the fact that $\J_f^{-1}\big|_{\rho} [\rho]=\id$, and the self-adjointness of $\J_f\big|_{\rho}$, it is easy to realize that the cross-terms are zero in such expression, thanks to the relation $\Tr{\delta \rho \,\J_f^{-1}\big|_{\rho} [\rho]}=\Tr{\delta\rho}=0$. Therefore we find:
		\begin{align}
			\Tr{\delta\rho' \,\J_f^{-1}\big|_{\rho'} [\delta\rho']}=\frac{(\delta r)^2}{r}+r\,\Tr{\delta\rho \,\J_f^{-1}\big|_{\rho} [\delta\rho]}.
		\end{align}
		The infinitesimal squared distance $\delta l^2=\Tr{\delta\rho' \,\J_f^{-1}\big|_{\rho'} [\delta\rho']}$  can thus be written as
		\begin{align}
			\label{eq:q_geo_spheric}
			\delta l^2=4 \left( \delta q^2+ q^2 \delta\theta^2\right)\;,\quad
			\text{with}\quad q\equiv\sqrt{r}\;,\quad \delta\theta^2\equiv\frac{1}{4}\Tr{\delta\rho\,\J_f^{-1}\big|_{\rho} [\delta\rho]}\;.
		\end{align}
		We see then that it is possible to separate the contribution due to the normalization coordinate $r\equiv q^2$, and the resulting underlying geometry is spherical with radius $2q$. This has a nontrivial consequence. Namely, for all $f$-defined Fisher metrics in Eq.~\eqref{cf:eq:monotoneMetrics}, solving geodesics on the set of physical states ($r=q^2=1$), is equivalent to solving geodesics on the set of positive operators ($r\in(0,\infty)$). 
		Formally speaking, one defines the geodesic length on the normalised ($\rho$) and unnormalised ($\rho'$) set as
		\begin{align}
			\label{eq:geo_def_L}
			(\Delta l_f)^2:=& \min_{\rho'(t)|\rho'(0)=r_0\rho_0,\rho'(1)=r_1\rho_1} \int_0^1 \de t\; K_{f,\rho'(t)}(\dot{\rho}'(t),\dot{\rho}'(t))\;,\\
			\label{eq:geo_def_theta}
			(\Delta \theta_f)^2:=& \min_{\rho(t)|\rho(0)=\rho_0,\rho(1)=\rho_1}  \frac{1}{4}\int_0^1 \de t\; K_{f,\rho(t)}(\dot{\rho}(t),\dot{\rho}(t))\;.
		\end{align}
		The direct consequence of the above discussion resulting in Eq.~\eqref{eq:q_geo_spheric} is the following relation, that holds independently from the choice of $f$,
		\begin{align}
			\label{eq:deltaL_deltatheta}
			(\Delta l_f)^2 = (r_0+r_1)-2 \sqrt{r_0 r_1}\cos(\Delta \theta_f)\;.
		\end{align}
		For example, in the case of the Bures metric~\ref{Bures}, one has that the angle $\Delta\theta_{\rm Bures}\equiv d_{B}/2$ is the Bures distance given in Eq.~\eqref{cf:eq:buresDistance} and $\cos(\Delta \theta_{\rm Bures})=\Tr{\sqrt{\sqrt{\rho_0}\rho_1\sqrt{\rho_0}}}$ corresponds to the fidelity between initial and final state, leading to the length $D_B$ in Eq.~\eqref{cf:eq:bureslength}.
		Similarly, for the Wigner-Yanase metric one has $\Delta\theta_{\rm WY}\equiv d_{\rm WY}/2$(Eq.~\eqref{cf:eq:wyGeodesicDistance}), and therefore $\cos(\Delta \theta_{\rm WY})=\Tr{{\sqrt{\rho_0}\sqrt{\rho_1}}}$, from which it follows immediately that the unnormalised geodesics length coincides with the contrast function in Eq.~\eqref{eq:WY_contrast}, i.e., $\Delta l_{\rm WY}\equiv \sqrt{H_{\rm WY}}$.
		
		The above equality Eq.~\eqref{eq:deltaL_deltatheta} shows that integrating any $f$-metrics on the set of normalised and unnormalised states is equivalent. This consideration might help in finding geodesics for other Fisher metrics by suitably choosing the problem between Eq.~\eqref{eq:geo_def_L} and Eq.~\eqref{eq:geo_def_theta} that is easier to solve.

		Finally, notice that in the classical, diagonal  case~\cite{abiuso2022characterizing,bengtssonGeometryQuantumStates2017}, $\Delta \theta$ and $\Delta l$ are known in the literature respectively as the \emph{Bhattacharyya angle} (or Bhattacharyya distance) and the \emph{Hellinger distance}.
		
	\end{additional-info}
	\newpage
	\section*{Acknowledgements}
	\begin{wrapfigure}{R}{-3cm}
		\includegraphics[width=2cm]{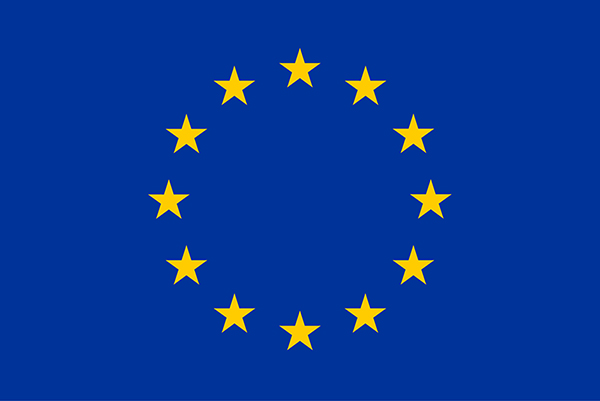}
	\end{wrapfigure}
	We warmly thank the late Mary-Beth Ruskai for pointing us to very useful literature, and for her determinant contribution to the field.
	M. S. acknowledges support from the European Union’s Horizon 2020 research and innovation programme under the Marie Sklodowska-Curie grant agreement No 713729, and from the Government of Spain (FIS2020-TRANQI and Severo Ochoa CEX2019-000910- S), Fundació Cellex, Fundació Mir-Puig, Generalitat de Catalunya (SGR 1381 and CERCA Programme). P.A. is supported by the QuantERA II programme that has received funding from the European Union’s Horizon 2020 research and innovation programme under Grant Agreement No 101017733, and from the Austrian Science Fund (FWF), project I-6004. Research at the Perimeter Institute for Theoretical Physics is supported by the Government of Canada through the Department of Innovation, Science and Economic Development Canada and by the Province of Ontario through the Ministry of Research, Innovation and Science. D.D.S. is supported
	by the research project “Dynamics and Information Research
	Institute - Quantum Information, Quantum Technologies”
	within the agreement between UniCredit Bank and Scuola
	Normale Superiore di Pisa (CI14\_UNICREDIT\_MARMI),
	the Spanish Government (FIS2020-TRANQI and Severo
	Ochoa CEX2019-000910-S), the ERC AdG CERQUTE, the
	AXA Chair in Quantum Information Science, Fundacio
	Cellex, Fundacio Mir-Puig and Generalitat de Catalunya
	(CERCA, AGAUR SGR 1381).
	
	\newpage
	\appendix
	
	\section{Derivations from Sec.~\ref{contrastFunctionsToFisher}}\label{app:contrastFunctions}
	
	We present here some derivations that were left implicit in Sec.~\ref{contrastFunctionsToFisher}. First, it should be noticed that the canonical expression of matrix convex functions such that $g(1) = 0$ is given by~\cite{lesniewskiMonotoneRiemannianMetrics1999}:
	\begin{align}
		g(x) = a\,(x-1) + b \, (x-1)^2 + c\, \frac{(x-1)^2}{x} + \int_0^\infty\de\mu_g(s)\; \frac{(x-1)^2}{x+s}\,,\label{app:eq:contrastFunctionsBasic}
	\end{align}
	where $b$ and $c$ are positive constants, and $\de\mu_g(s)$ is a positive measure with finite mass (i.e., $\int_0^\infty\de\mu_g(s)<\infty$). First, it should be noticed that we can ignore the linear term since, as it is discussed in the main text, it does not contribute to the contrast function. Moreover, it is useful to point out that the limit of integration can be made bounded as follows:
	\begin{align}
		\int_0^\infty\de\mu_g(s)\; \frac{(x-1)^2}{x+s} &= \int_0^1\de\mu_g(s)\; \frac{(x-1)^2}{x+s} + \int_1^\infty\de\mu_g(s)\; \frac{(x-1)^2}{x+s}= \\
		&= \int_0^1\de\mu_g(s)\; \frac{(x-1)^2}{x+s} + \int_0^1\frac{\de\mu_g(s^{-1})}{s}\; \frac{(x-1)^2}{1+sx}\,,\label{eq:A3}
	\end{align}
	where we performed the change of variable $s\rightarrow s^{-1}$ in the second integral. Then, redefining the measure, we can include the second and the third term in Eq.~\eqref{app:eq:contrastFunctionsBasic} into the integral. This is accomplished by introducing the measure $\de\nu_g(s) = \de{\mu}_g(s) + b\,\delta(s^{-1})/s + c \,\delta(s)$, where $\delta(x)$ denote the Dirac delta function. Plugging this in, we obtain the general form of matrix convex functions such that $g(1)= 0$:
	\begin{align}
		g(x) =\int_0^1\de\nu_g(s)\; \frac{(x-1)^2}{x+s} + \int_0^1\frac{\de\nu_g(s^{-1})}{s}\; \frac{(x-1)^2}{1+sx}\,,\label{app:eq:integralExpressionG}
	\end{align}
	up to irrelevant linear terms.  
	
	Another expression that was not proved in the main text is Eq.~\eqref{cf:eq:HGAlternative}.This follows from the two identities:
	\begin{align}
		&(\LL_\sigma\RR_\rho^{-1}-\id)\sqrbra{\rho^{1/2}} = (\sigma - \rho)\rho^{-1/2} =  \RR_\rho^{-1/2}(\sigma - \rho)\,,\label{app:eq:a6}\\
		&(\RR_\sigma\LL_\rho^{-1}-\id)\sqrbra{\rho^{1/2}}= \rho^{-1/2}(\sigma - \rho) =  \LL_\rho^{-1/2}(\sigma - \rho)\,,\label{app:eq:a7}
	\end{align}
	which can be verified by explicit computation. Then, one can rewrite the contrast functions in Eq.~\eqref{cf:eq:HGexpressionApp} as:
	\begin{align}
		H_g (\rho||\sigma) &= \Tr{ g\boldsymbol{(}\LL_\sigma\RR_\rho^{-1}\boldsymbol{)}\sqrbra{\rho}} =  \Tr{ \rho^{1/2}\,g\boldsymbol{(}\LL_\sigma\RR_\rho^{-1}\boldsymbol{)}\sqrbra{\rho^{1/2}}} =\\
		&=\Tr{{\Big((\RR_\sigma\LL_\rho^{-1}-\id)^{-1}[(\rho-\sigma)]\Big)}\;{\Big(g\boldsymbol{(}\LL_\sigma\RR_\rho^{-1}\boldsymbol{)}(\LL_\sigma\RR_\rho^{-1}-\id)^{-1}[(\rho-\sigma)]\Big)}\;\rho^{-1}}=\\
		&= \Tr{(\rho-\sigma)\,\Big(g\boldsymbol{(}\LL_\sigma\RR_\rho^{-1}\boldsymbol{)}(\LL_\sigma\RR_\rho^{-1}-\id)^{-2}[(\rho-\sigma)]\Big)\,\rho^{-1}}=\\
		&=\Tr{(\rho-\sigma)\, \RR_\rho^{-1}h\boldsymbol{(}\LL_\sigma\RR_\rho^{-1}\boldsymbol{)}[(\rho-\sigma)]} ,
	\end{align}
	where in the first line we used the fact that since the superoperator $g\boldsymbol{(}\LL_\sigma\RR_\rho^{-1}\boldsymbol{)}$ acts on $\rho$ on the right only through $\RR_\rho^{-1}$ one can extract $\rho^{1/2}$ from the square parenthesis, and use the cyclicity of the trace to put it in front; in the second line we used both Eq.~\eqref{app:eq:a6} and Eq.~\eqref{app:eq:a7}, in the third line the fact that $(\RR_\sigma\LL_\rho^{-1}-\id)^\dagger = (\LL_\sigma\RR_\rho^{-1}-\id)$ and finally the definition of the function $h(x) := g(x)/(x-1)^{2}$. This proves Eq.~\eqref{cf:eq:HGAlternative}.
	
	Hence, from the integral expression in Eq.~\eqref{app:eq:integralExpressionG} it follows that we can rewrite generic contrast functions as:
	\begin{align}
		H_g( \rho || \sigma)  = &\int_0^1\de\nu_g(s) \,\Tr{(\rho-\sigma)\, \RR_\rho^{-1}(\LL_\sigma\RR_\rho^{-1}+s)^{-1}[(\rho-\sigma)]}+\\
		&+\int_0^1\frac{\de\nu_g(s^{-1})}{s} \,\Tr{(\rho-\sigma)\, \RR_\rho^{-1}(1+s\LL_\sigma\RR_\rho^{-1})^{-1}[(\rho-\sigma)]}\,.
	\end{align}
	Let us first focus on the first trace. In particular, it should be noticed that:
	\begin{align}
		\RR_\rho^{-1}(\LL_\sigma\RR_\rho^{-1}+s)^{-1} = ((\LL_\sigma\RR_\rho^{-1}+s)\RR_\rho)^{-1}  =(\LL_\sigma+s\RR_\rho)^{-1} \,,
	\end{align}
	proving that the first integral coincides with the first integral of Eq.~\eqref{cf:eq:integralRepresentationG}. Doing the same transformation on the second trace we obtain:
	\begin{align}
		\Tr{(\rho-\sigma)\, \RR_\rho^{-1}(1+s\LL_\sigma\RR_\rho^{-1})^{-1}[(\rho-\sigma)]} &= \Tr{(\rho-\sigma)\, (\RR_\rho+s\LL_\sigma)^{-1}[(\rho-\sigma)]} = \\
		&=\Tr{\big ((\LL_\rho+s\RR_\sigma)^{-1}[(\rho-\sigma)]\big )\, (\rho-\sigma)}\label{app:eq:a16}\,,
	\end{align}
	where in the last step we used the fact that $(\RR_\rho+s\LL_\sigma)^\dagger = (\LL_\rho+s\RR_\sigma)$. Hence, putting everything together, and using the cyclicity of the trace in Eq.~\eqref{app:eq:a16}, we finally obtain:
	\begin{align}
		H_g( \rho || \sigma)  
		= \int_0^1 \de\nu_g(s)\;\Tr{(\rho-\sigma)(\LL_\sigma + s \RR_\rho)^{-1}[(\rho-\sigma)]}+\int_0^1\frac{\de\nu_g(s^{-1})}{s}\;\Tr{(\rho-\sigma)(\LL_\rho + s \RR_\sigma)^{-1}[(\rho-\sigma)]}\,.
	\end{align}
	proving Eq.~\eqref{cf:eq:integralRepresentationG}.

	\section{Derivation of the flux of Fisher information}\label{app:sec:fisherFlow}
	We present here the derivation of Thm.~\ref{cf:thm:FisherFlow}. In particular, we want to study the evolution of the Fisher information:
	\begin{align}
		\mathcal{F}_{f,t}:= \Tr{\delta\rho_t\,\J_f^{-1}\big|_{\pi_t} [\delta\rho_t]}\,,
	\end{align}
	where $\pi_t := \Phi_t(\pi)$ and $\delta\rho_t := \Phi_t(\delta\rho)$. Using the integral expression in Eq.~\eqref{cf:eq:thm2HGexpIntegral}, we can rewrite the Fisher information as:
	\begin{align}
		\mathcal{F}_{f,t}: = 2\,{\rm Re}\int_0^1 \de N_g (s)\; \Tr{\delta\rho_t\, (\LL_{\pi_t}+s\,\RR_{\pi_t})^{-1}[\delta\rho_t]}\,,
	\end{align}
	where it is useful to highlight the fact that this quantity is real. This expression is particularly convenient for calculations, due to the simple form that the derivative of $(\LL_{\pi_t}+s\,\RR_{\pi_t})^{-1}$ takes. In fact, this is given by:
	\begin{align}\label{appM:eq:3}
		\frac{\de}{\dt}(\LL_{\pi_t}+s\,\RR_{\pi_t})^{-1} = - (\LL_{\pi_t}+s\,\RR_{\pi_t})^{-1}(\LL_{\dot\pi_t}+s\,\RR_{\dot\pi_t})(\LL_{\pi_t}+s\,\RR_{\pi_t})^{-1}\,,
	\end{align}
	where $\dot \pi_t$ is simply the derivative of the state. This expression can be proved by noticing that $\frac{\de}{\dt}\LL_{\pi_t} = \LL_{\dot{\pi}_t}$ (and similarly for $\RR_{\pi_t}$) and by taking the derivative of:
	\begin{align}
		\frac{\de}{\dt}\norbra{(\LL_{\pi_t}+s\,\RR_{\pi_t})(\LL_{\pi_t}+s\,\RR_{\pi_t})^{-1}} &= (\LL_{\dot{\pi}_t}+s\,\RR_{\dot{\pi}_t})(\LL_{\pi_t}+s\,\RR_{\pi_t})^{-1} +(\LL_{\pi_t}+s\,\RR_{\pi_t})\frac{\de}{\dt}(\LL_{\pi_t}+s\,\RR_{\pi_t})^{-1}=\\
		&= \frac{\de}{\dt} \,\idO = 0\,,
	\end{align}
	which directly implies Eq.~\eqref{appM:eq:3}. Given this technical tool, we can start analysing the evolution of $\mathcal{F}_{f,t}$ under the dynamics generated by the Lindbladian:
	\begin{align}
		\lind_t[\rho] = -i[H_t,\rho] + \sum_{\alpha}^{d^2} \;\lambda_\alpha(t)\,\norbra{A_\alpha(t) \,\rho\, A_\alpha(t)^\dagger - \frac{1}{2}\{A_\alpha(t)^\dagger \,A_\alpha(t), \rho\}}\,.
	\end{align}
	Notice that, since $\mathcal{F}_{f,t}$ is invariant under unitary transformations, there is no contribution coming from the commutator in the previous equation. Moreover, since the derivative is linear, it decomposes into a sum of the form:
	\begin{align}
		\mathcal{F}'_{f,t} = \sum_\alpha\; \lambda_\alpha(t)\, \,\mathcal{I}^f_\alpha(t)\,,
	\end{align}
	where each current $\mathcal{I}^f_\alpha(t)$ only contains the corresponding jump operator $A_\alpha(t)$, together with its adjoint. For this reason, without loss of generality, we consider here Lindblad operators generated by a single jump operator. In order to shorten the notation we also assume that the jump operator, denoted by $A$, is time independent, again without loss of generality. We start by rewriting the derivative of the Fisher information as:
	\begin{align}
		\mathcal{F}'_{f,t}=2\,{\rm Re}\int_0^1 \de N_g (s)\; \bigg (2\, &\Tr{\delta\dot\rho_t (\LL_{\pi_t}+s\,\RR_{\pi_t})^{-1}[\delta\rho_t]} + \Tr{\delta\rho_t \norbra{\frac{\de}{\dt}(\LL_{\pi_t}+s\,\RR_{\pi_t})^{-1} }[\delta\rho_t]}\bigg )\,.\label{appM:eq:9}
	\end{align}
	The second term in the integral  can be expanded as:
	\begin{align}
		\Tr{\delta\rho_t \norbra{ \frac{\de}{\dt}(\LL_{\pi_t}+s\,\RR_{\pi_t})^{-1} }[\delta\rho_t]} &=- \Tr{\delta\rho_t (\LL_{\pi_t}+s\,\RR_{\pi_t})^{-1}(\LL_{\dot\pi_t}+s\,\RR_{\dot\pi_t})(\LL_{\pi_t}+s\,\RR_{\pi_t})^{-1}[\delta\rho_t]} = \\
		&= - \Tr{B_{s}(t)^\dagger \,\dot\pi_t \,B_{s}(t)} - s\, \Tr{B_{s}(t)^\dagger  B_{s}(t)\,\dot\pi_t} =\\
		&= - \Tr{B_{s}(t)^\dagger \,\lind({\pi_t}) \,B_{s}(t)} - s\, \Tr{B_{s}(t)^\dagger  B_{s}(t)\,\lind({\pi_t}) }\,,\label{cf:eq:B8}
	\end{align}
	where we introduced the notation $B_{s}(t) := (\LL_{\pi_t}+s\,\RR_{\pi_t})^{-1}[\delta\rho_t]$.  On the other hand, the first term in Eq.~\eqref{appM:eq:9} simply gives:
	\begin{align}
		2\, \Tr{\delta\dot\rho_t (\LL_{\pi_t}+s\,\RR_{\pi_t})^{-1}[\delta\rho_t]} =
		&= 2 \,\Tr{\lind(\delta\rho_t) B_{s}(t)} = 2\, \Tr{\lind((\RR_{\pi_t}+s\,\LL_{\pi_t})[ B_{s}(t)^\dagger])  B_{s}(t)}=\\
		&=2 \,\Tr{\lind( B_{s}(t)^\dagger \pi_t) B_{s}(t)} + 2s\, \Tr{\lind(\pi_t  B_{s}(t)^\dagger ) B_{s}(t)}\,,\label{cf:eq:B10}
	\end{align}
	where in the first line we have multiplied and divided by $(\RR_{\pi_t}+s\,\LL_{\pi_t})$ to obtain $ B_{s}(t)^\dagger $. We can now proceed in summing up Eq.~\eqref{cf:eq:B8} and Eq.~\eqref{cf:eq:B10}. Due to the number of terms that will appear, though, we first consider the first traces in both equations, and then the second ones. Hence, summing the first term in  Eq.~\eqref{cf:eq:B8} and the first of Eq.~\eqref{cf:eq:B10}, and explicitly expanding the Lindbladian, we obtain:
	\begin{align}
		&2 \,\Tr{A\, B_{s}(t)^\dagger \pi_tA^\dagger  B_{s}(t)} - \cancel{\Tr{ B_{s}(t)^\dagger \pi_tA^\dagger A\,  B_{s}(t)}}-  \Tr{A^\dagger A\, B_{s}(t)^\dagger \pi_t B_{s}(t)} +\nonumber\\
		&- \Tr{ B_{s}(t)^\dagger A\,\pi_t A^\dagger\,  B_{s}(t)} +\cancel{\frac{1}{2} \Tr{ B_{s}(t)^\dagger\pi_t A^\dagger A\,  B_{s}(t)}}+\cancel{\frac{1}{2} \Tr{ B_{s}(t)^\dagger A^\dagger A\,\pi_t   B_{s}(t)}}\,,
	\end{align}
	where the first three terms come from Eq.~\eqref{cf:eq:B10}, and the last three from Eq.~\eqref{cf:eq:B8}.
	Notice that the second term corresponding to Eq.~\eqref{cf:eq:B10} cancels with the last two terms coming from Eq.~\eqref{cf:eq:B8}, since we can take the Hermitian conjugate of the last trace without affecting the result (thanks to the real part in Eq.~\eqref{appM:eq:9}).  The remaining terms can be further simplified to give:
	\begin{align}
		2&\Tr{\pi_tA^\dagger  B_{s}(t)A\, B_{s}(t)^\dagger } -  \Tr{ \pi_t B_{s}(t)A^\dagger A\, B_{s}(t)^\dagger}- \Tr{\pi_t A^\dagger\,  B_{s}(t) B_{s}(t)^\dagger A} =\\
		&\qquad= \Tr{\pi_t\,[A^\dagger,  B_{s}(t)]\,A B_{s}(t)^\dagger } + \Tr{\pi_t\,A^\dagger  B_{s}(t)\,[A, B_{s}(t)^\dagger] }=\\
		&\qquad= -\Tr{\pi_t\, B_{s}(t)A^\dagger\,[A,  B_{s}(t)^\dagger]} + \Tr{\pi_t\,A^\dagger  B_{s}(t)\,[A, B_{s}(t)^\dagger] } =\\
		&\qquad=\Tr{\pi_t\,[A^\dagger,  B_{s}(t)][A, B_{s}(t)^\dagger] } =\\
		&\qquad= -  \Tr{\pi_t\,[A,  B_{s}(t)^\dagger]^\dagger[A, B_{s}(t)^\dagger] }\,, 
	\end{align}
	where in the first line we used the cyclicity of the trace to put in evidence $\pi_t$, then in the third line we took the complex conjugate of the trace (again exploiting the fact that only the real part contributes to the Fisher information), and used the identity $[X,Y]^\dagger = [Y^\dagger,X^\dagger] =-[X^\dagger,Y^\dagger] $.
	
	We can now pass to examine the last two terms in Eq.~\eqref{cf:eq:B8} and Eq.~\eqref{cf:eq:B10}. Since the steps are completely analogous to the ones in the previous derivation, we present them all together. Carrying them out gives:
	\begin{align}
		&2s\, \Tr{A\pi_t  B_{s}(t)^\dagger A^\dagger  B_{s}(t)} - \cancel{s\Tr{A^\dagger A\pi_t  B_{s}(t)^\dagger  B_{s}(t)}}
		- s\,\Tr{\pi_t  B_{s}(t)^\dagger A^\dagger A\, B_{s}(t)} +\\
		&- s\, \Tr{ B_{s}(t)^\dagger   B_{s}(t)A\pi_tA^\dagger}+\cancel{\frac {s}{2}\, \Tr{ B_{s}(t)^\dagger   B_{s}(t)A^\dagger A\pi_t}}+\cancel{\frac {s}{2}\, \Tr{ B_{s}(t)^\dagger   B_{s}(t)\pi_tA^\dagger A}} = \\
		&= s\, \Tr{\pi_t \,[ B_{s}(t)^\dagger, A^\dagger ] \,B_{s}(t)A}+s\, \Tr{\pi_t\,  B_{s}(t)^\dagger A^\dagger\, [ B_{s}(t),A]} =\\
		&= - s\, \Tr{\pi_t \,A^\dagger  B_{s}(t)^\dagger\, [ B_{s}(t), A ]} +s\, \Tr{\pi_t \, B_{s}(t)^\dagger A^\dagger\, [ B_{s}(t),A]}=\\ 
		&= s\, \Tr{\pi_t\, [ B_{s}(t)^\dagger, A^\dagger] [ B_{s}(t),A]} =\\
		&= -s\, \Tr{\pi_t\, [ B_{s}(t), A]^\dagger [ B_{s}(t),A]}=\\
		& = -s\, \Tr{\pi_t \,[A, B_{s}(t)]^\dagger [A, B_{s}(t)]} \,.
	\end{align}
	This concludes the proof of Thm.~\ref{cf:thm:FisherFlow}. In fact, it is straightforward to generalise to the case of many time dependent jump operators, which gives:
	\begin{align}
		&\mathcal{F}'_{f,t} = \sum_\alpha\; \lambda_\alpha(t)\, \,\mathcal{I}^f_\alpha(t)= \\
		&=-2\, \sum_\alpha\; \lambda_\alpha(t)\,\int_0^1 \de N_g (s)\; \bigg(\Tr{\pi_t\,[A_\alpha(t),  B_{s}(t)^\dagger]^\dagger[A_\alpha(t), B_{s}(t)^\dagger] } +s\, \Tr{\pi_t\, [A_\alpha(t), B_{s}(t)]^\dagger [A_\alpha(t), B_{s}(t)]} \bigg)\,,
	\end{align}
	where we can drop the real part, as both terms in the integral are positive definite.

	\section{Derivations from Sec.~\ref{detailedBalance}}\label{app:sec:DB}
	
	In this appendix we first prove Thm.~\ref{cf:thm:eqDB}, and then proceed to derive the structural characterisation of Fisher detailed balanced Lindbladians presented in Eq.~\eqref{cf:eq:structAS}. 
	
	Due to the amount of different notions of adjoints used in the following, we remind the reader about the notation. There are three different scalar product used, namely the Hilbert-Schmidt one, $K_\pi^o$ and $K_{f,\pi}$, which induce the following adjoints:
	\begin{align}
		\text{Hilbert}&\text{-Schmidt:}  \qquad \;\;\Tr{A X(B)}=\Tr{X^\dagger(A) B}\,;\label{eq:appHS}\\
		&K_\pi^o: \qquad\qquad \;\;\,\Tr{A\,\cJ_\pi[OB]} = \Tr{\widetilde{O}^o(A) \cJ_\pi[B]}\,;\label{eq:appAlickiProduct}\\
		&K_{f,\pi}: \qquad \qquad\Tr{A \,\J_f^{-1}\big|_\pi[\mathcal{O} B]} = \Tr{(\widetilde{\mathcal{O}}_f A) \J_f^{-1}\big|_\pi[ B]}\,,\label{eq:appASProduct}
	\end{align}
	where $O$ and $\mathcal{O}$ are a superoperator on the space of observables or on the state space, respectively, while $X$ is a generic bounded operator. We also remind the reader that we use the notation $\cJ_\pi:= \RR_\pi$. Then, the adjoint with respect to $K_\pi^o$ or $K_{f,\pi}$ are related to the Hilbert-Schmidt one by the relation:
	\begin{align}
		\widetilde{O}^o = \cJ_\pi^{-1}\circ O^\dagger\circ\cJ_\pi\,;\qquad\qquad\qquad
		\widetilde{\mathcal{O}}_f  =  \J_f\big|_{\pi}\circ \mathcal{O}^\dagger\circ \J_f^{-1}\big|_{\pi}\,,
	\end{align}
	as it can be verified directly from the definition. In this context, self-adjointness with respect to $K^o_\pi$ is equivalent to the condition $\cJ_\pi\circ O= O^\dagger\circ\cJ_\pi$, while for the Fisher scalar product it can be expressed by the equality $\mathcal{O}\circ\J_f\big|_{\pi} = \J_f\big|_{\pi}\circ \mathcal{O}^\dagger$. Thanks to this characterisation we can prove the following useful result:
	\begin{lemma}\label{appM:lem:equivalence_when_commuting} 
		Suppose $\mathcal{O}$ and $O$ are adjoint of each other, $\mathcal{O}^\dagger \equiv O$. Then, if $[\mathcal{O},\LL_\pi\RR_\pi^{-1}]=0$, the two conditions of self-adjointness and skew-self-adjointness with respect to $K_{f,\pi}$ and $K_\pi^o$ coincide.
	\end{lemma}
	\begin{proof}
		First of all, it should be noticed that the adjoint of the modular operator takes the form $(\LL_\pi\RR_\pi^{-1})^\dagger = \LL_\pi^{-1}\RR_\pi = (\LL_\pi\RR_\pi^{-1})^{-1}$, while for the right multiplication operator it holds that $(\RR_\pi)^\dagger = \LL_\pi$. Using these two properties we can show that the (skew-)self-adjointness with respect to $K_{f,\pi}$ is equivalent to the corresponding notion for  $K_\pi^o$. In fact, the following relations are equivalent:
		\begin{align}
			&\widetilde{\mathcal{O}}_f = \pm \mathcal{O} \iff
			(\mathcal{O} \J_f\big|_{\pi})^\dagger = \pm\,\mathcal{O}\J_f\big|_{\pi}\; \iff\;f((\LL_\pi\RR_\pi^{-1})^{-1}) \LL_\pi O = \pm\,\mathcal{O} \RR_\pi f(\LL_\pi\RR_\pi^{-1})   \; \iff\\
			&\iff\; \cancel{f(\LL_\pi\RR_\pi^{-1}) }\RR_\pi O = \pm\,\cancel{f(\LL_\pi\RR_\pi^{-1})} O^\dagger  \RR_\pi \iff\; \widetilde{O}^o = \pm O \,,
		\end{align}
		where in the first line we used the definition $\mathcal{O}^\dagger \equiv O$, and in the second line we exploited the property $f(x) = x f(x^{-1})$, together with the commutation between $\mathcal{O}$ and $\LL_\pi\RR_\pi^{-1}$ to push $f(\LL_\pi\RR_\pi^{-1})$ to the left of both equations. 
	\end{proof} 
	
	This lemma is particularly useful because it allows to reduce the question about the equivalence of Def.~\ref{cf:def:alicki} and Def.~\ref{cf:def:AS} to the decision about the commutation of the Lindbladian with the modular operator $\LL_\pi\RR_\pi^{-1}$.
	
	\subsection{Proof of Theorem~\ref{cf:thm:eqDB}}
	
	The aim of this section is to prove Thm.~\ref{cf:thm:eqDB}, which we repeat here for convenience:
	\begin{theorem*}\emph{
			The following conditions are equivalent:
			\begin{enumerate}
				\item  the generator of the dynamics in the Heisenberg picture $\lind^\dagger$ satisfies the adjointness relations in Def.~\ref{cf:def:alicki};\label{appM:cond:al}
				\item the Lindbladian $\lind$ satisfies the structural characterisation in Def.~\ref{cf:def:breuer}.\label{appM:cond:breuer}
				\setcounter{counterN}{\value{enumi}}
			\end{enumerate}
			These conditions imply the condition:
			\begin{enumerate}
				\setcounter{enumi}{\value{counterN}}
				\item the generator of the dynamics in the Schroedinger picture $\lind$ satisfies the adjointness relations in Def.~\ref{cf:def:AS}.\label{appM:cond:AS}
			\end{enumerate}
			Moreover, if the Hamiltonian $H$ is non-degenerate the three conditions are equivalent.}
	\end{theorem*}
	
	First, it should be noticed that the equivalence between condition~\ref{appM:cond:al} and~\ref{appM:cond:breuer} was already proved by Alicki in~\cite{alicki1976detailed}, so we postpone the proof to App.~\ref{appM:sec:structuralCh} where we characterise the Lindbladians satisfying condition~\ref{appM:cond:AS}. Then, if $H$ is non-degenerate, this provides a proof of the structural definition of condition~\ref{appM:cond:breuer}.  
	
	\begin{proof}
		First it should be noticed that if condition~\ref{appM:cond:breuer} is satisfied, the Lindbladian commutes with the modular operator $\LL_\pi\RR_\pi^{-1}$. In fact, starting from the characterisation:
		\begin{align}
			\lind(\rho)=-i[H,\rho]+\sum_{\omega, i} \;\lambda_{i}^\omega \left(A_i^{\omega}\,\rho \,(A_i^{\omega})^\dagger -\frac{1}{2}\{(A_i^{\omega})^\dagger A_i^{\omega},\rho \} \right)\,,
		\end{align}
		it is a matter of straightforward calculations to verify that:
		\begin{align}
			\LL_\pi\RR_\pi^{-1}(\lind(\rho))&=-i\pi[H,\rho]\pi^{-1}+\sum_{\omega, i} \;\lambda_{i}^\omega \left((\pi\,A_i^{\omega}\pi^{-1})\,\pi\,\rho \,\pi^{-1}(\pi^{-1}A_i^{\omega}\,\pi)^\dagger -\frac{1}{2}\pi\{(A_i^{\omega})^\dagger A_i^{\omega},\rho \}\pi^{-1} \right)=\\
			&= -i[H,\pi\rho\pi^{-1}]+\sum_{\omega, i} \;\lambda_{i}^\omega \cancel{e^{\omega}}\cancel{e^{-\omega}}\left(A_i^{\omega}\,\pi\,\rho \,\pi^{-1}(A_i^{\omega})^\dagger -\frac{1}{2}\{(A_i^{\omega})^\dagger A_i^{\omega},\pi\,\rho\pi^{-1} \} \right)=\\
			&=\lind(\LL_\pi\RR_\pi^{-1}(\rho))
		\end{align}
		where we used the condition $[H,\pi]=0$, together with $\LL_\pi\RR_\pi^{-1}(A_i^\omega) = e^\omega \, A_i^{\omega}$ and $(A_i^{\omega})^\dagger = A_i^{-\omega}$. Since condition~\ref{appM:cond:breuer} is equivalent to condition~\ref{appM:cond:al}, and $[\mathcal{L},\LL_\pi\RR_\pi^{-1}]=0$, thanks to Lemma~\ref{appM:lem:equivalence_when_commuting} this means that $\lind$ has the same self-adjointness properties with respect to $K_\pi^o$ and $K_{f,\pi}$. This proves the forward implication.
		
		Let us prove the reverse, namely that condition~\ref{appM:cond:AS} is equivalent to condition~\ref{appM:cond:al} for non-degenerate Hamiltonians. Let us first focus on the unitary part. Then, using $\mathcal{U}^\dagger = -\mathcal{U}$ we can rewrite the skew-adjointness condition as:
		\begin{align}
			\mathcal{U} \circ\J_f\big|_{\pi} = -\J_f\big|_{\pi}\circ \mathcal{U}^\dagger \qquad \iff\qquad\mathcal{U} \circ\J_f\big|_{\pi} = \J_f\big|_{\pi}\circ \mathcal{U}\,.
		\end{align}
		Then, applying the two operators in the last equation to the identity, we can verify that:
		\begin{align}
			\mathcal{U} \circ\J_f\big|_{\pi}(\id) = \J_f\big|_{\pi}\circ \mathcal{U}(\id) &\implies \mathcal{U}(\pi) = -i\,\J_f\big|_{\pi}([H,\id]) \implies [H,\pi] = 0\,.\label{appM:eq:commU}
		\end{align}
		Notice that this result is generic, i.e., no assumptions on the spectrum of $H$ need to be made. This directly implies the commutation $[\mathcal{U},\LL_\pi\RR_\pi^{-1}]=0$, so again thanks to Lemma~\ref{appM:lem:equivalence_when_commuting} we have that  for the unitary part condition~\ref{appM:cond:AS} and~\ref{appM:cond:al} are always equivalent.
		
		Let us now focus on Eq.~\eqref{appM:eq:commU}. Since $H$ and $\pi$ commute, we can find a common set of eigenvectors, which we denote by $\{\ket{\alpha}\}$. Moreover, the basis $\{\ket{\alpha}\bra{\beta}\}$ gives a set of common eigenvectors to $\mathcal{U}$ and $\LL_\pi\RR_\pi^{-1}$, as it can be verified by direct calculation:
		\begin{align}
			\mathcal{U}(\ket{\alpha}\bra{\beta})=-i(H_\alpha-H_\beta)\ket{\alpha}\bra{\beta}\,; \qquad\qquad\qquad \LL_\pi\RR_\pi^{-1}(\ket{\alpha}\bra{\beta})=\frac{\pi_\alpha}{\pi_\beta}\ket{\alpha}\bra{\beta} \, .
		\end{align}
		Both superoperators have constant eigenvalues for all eigenvectors of the form $\ket{\alpha}\bra{\alpha}$. Under the assumption of continuity under small perturbations, we can also assume that each eigenvector such that $\alpha\neq\beta$ has a different eigenvalue (non-degenerate gap condition). These two observations together then imply that any superoperator commuting with $\mathcal{U}$ needs to commute with $\LL_\pi\RR_\pi^{-1}$ as well. In the main text we showed how normality of $\lind$ implies $[\mathcal{U},\lind_\mathcal{D}]=0$. Then, thanks
		to the considerations above, we also have that:
		\begin{align}
			[\lind_\mathcal{D},\LL_\pi\RR_\pi^{-1}]=0\ .\label{appM:eq:106}
		\end{align}
		Thus, thanks once again to Lemma~\ref{appM:lem:equivalence_when_commuting}, we have that $(\widetilde{\lind_\mathcal{D}})_f=\lind_\mathcal{D}$ implies $(\widetilde{\lind_\mathcal{D}^\dagger})^o=\lind_\mathcal{D}^\dagger $. This concludes the proof.
	\end{proof}
	
	It should be noticed that the non-degeneracy of the spectrum is needed to prove the commutation relation in Eq.~\eqref{appM:eq:106}. The same equivalence can be proved if $\LL_\pi\RR_\pi^{-1}$ has some functional dependence on $\mathcal{U}$. Take for example the thermal scenario, i.e., $\pi\propto \exp[-\beta H]$. Then the modular operator takes the form $\LL_\pi\RR_\pi^{-1}=\exp[-i\beta \mathcal{U}]$. Due to normality of the Lindbladian (which implies $[\mathcal{U},\lind_\mathcal{D}]=0$), we directly obtain Eq.~\eqref{appM:eq:106}, proving the equivalence without further assumptions on the spectrum of $\pi$ or $H$.
	
	\subsection{Def.~\ref{cf:def:AS} is weaker in general}
	
	Whereas the constraints coming from Def.~\ref{cf:def:alicki} imply the ones in Def.~\ref{cf:def:AS}, the reverse does not hold in general. In fact, this is connected with the commutation between $\lind_\mathcal{D}$ and $\LL_\pi\RR_\pi^{-1}$. Whereas in the first definition of detailed balance these two operators always commute, this is not the case for the Fisher detailed balance dissipators. This leads to a less constrained evolution of the coherences, as it will be shown in the following.
	
	First, as it was discussed in the proof in the previous section, if $\{\ket{\alpha}\}$ is an eigenbasis for $\pi$, then $\{\ket{\alpha}\bra{\beta}\}$ are eigenvectors for $\LL_\pi\RR_\pi^{-1}$, with eigenvalues:
	\begin{align}
		\LL_\pi\RR_\pi^{-1}(\ket{\alpha}\bra{\beta})=\frac{\pi_\alpha}{\pi_\beta}\ket{\alpha}\bra{\beta} \, .
	\end{align}
	Since the steady state $\pi$ is always assumed to be full rank, proving that $\lind_\mathcal{D}$ commutes with $\LL_\pi\RR_\pi^{-1}$ is equivalent to requiring that the matrix elements
	\begin{align}
		(\lind_\mathcal{D})_{\delta|\beta}^{\gamma|\alpha}\equiv\bra{\gamma}\lind_\mathcal{D}(\ketbra{\alpha}{\beta})\ket{\delta}
	\end{align}
	satisfy the following condition
	\begin{align}
		[\lind_\mathcal{D},\Phi_\pi]=0 \quad\iff\quad 	(\lind_\mathcal{D})_{\delta|\beta}^{\gamma|\alpha}\left(\frac{\pi_\alpha}{\pi_\beta} -\frac{\pi_\gamma}{\pi_\delta}\right)=0\,.
	\end{align}
	Equivalently, this means that matrix elements of $\mathcal{D}$ can be nonzero only if:
	\begin{align}
		\label{appM:eq:D_comm_coordinates}
		(\lind_\mathcal{D})_{\delta|\beta}^{\gamma|\alpha}\neq 0 \quad\implies\quad \frac{\pi_\alpha}{\pi_\beta} =\frac{\pi_\gamma}{\pi_\delta}\,.
	\end{align}
	
	In the following we show that a slightly more general condition follows from the requirement that $(\widetilde{\lind_\mathcal{D}})_f  = \lind_\mathcal{D}$ for all standard monotone functions $f$. Indeed, this condition can be written in coordinates as:
	\begin{align}
		\lind_\mathcal{D}&\circ\J_f\big|_{\pi} = \J_f\big|_{\pi}\circ \lind_\mathcal{D}^\dagger \quad\iff\quad (\lind_\mathcal{D})^{\gamma|\alpha}_{\delta|\beta}\, f\norbra{\frac{\pi_\alpha}{\pi_\beta}}\pi_\beta  = (\lind_\mathcal{D})_{\alpha|\gamma}^{\beta|\delta}\,f\norbra{\frac{\pi_\gamma}{\pi_\delta}}\pi_\delta\,.
		\label{appM:eq:coordinateExpression}
	\end{align}
	At this point it is useful to introduce the notation $e^{-\omega_1} := \pi_\alpha/\pi_\beta$ and $e^{-\omega_2} := \pi_\gamma/\pi_\delta$.
	Grouping the functional dependence on one side of the equation we obtain:
	\begin{align}
		\frac{(\lind_\mathcal{D})^{\gamma|\alpha}_{\delta|\beta}\;\pi_\beta}{(\lind_\mathcal{D})_{\alpha|\gamma}^{\beta|\delta}\;\pi_\delta}  = \frac{f\norbra{e^{-\omega_2}}}{f\norbra{e^{-\omega_1}}}\,.\label{appM:eq:coordinateExpression2}
	\end{align}
	It should be noticed that the left hand side of the equation does not depend on the function $f$, so the coordinates of $\lind_\mathcal{D}$ are zero unless $\omega_1=\pm\omega_2$ (notice that one cannot rule out the case $\omega_1=-\omega_2$, since this follows from the symmetry of standard monotone functions $f(x) = x\,f(x^{-1})$).
	
	Then, the only non-zero elements of  a Fisher self-adjoint $\lind_\mathcal{D}$ are the ones for which either of the conditions:
	\begin{align}
		\label{appM:eq:D_prop_or}
		(\lind_\mathcal{D})^{\gamma|\alpha}_{\delta|\beta}\neq 0 \implies
		\left(\frac{\pi_\alpha}{\pi_\beta}=\frac{\pi_\gamma}{\pi_\delta}\right) 
		\lor  
		\left(\frac{\pi_\alpha}{\pi_\beta}=\frac{\pi_\delta}{\pi_\gamma}\right)\, ,
	\end{align}
	are satisfied.
	Comparing this result with Eq.~\eqref{appM:eq:D_comm_coordinates} directly shows that in general $\lind_\mathcal{D}$ does not commute with $\LL_\pi\RR_\pi^{-1}$, so Def.~\ref{cf:def:AS} is weaker than Def.~\ref{cf:def:alicki} (at the end of the section we present an explicit example showing this). For this reason, it is interesting to explore which constraints Eq.~\eqref{appM:eq:D_prop_or} imposes on the Lindbladian. First, it should be noticed that since $\lind_\mathcal{D}$ is adjoint preserving, its coordinates satisfy 
	\begin{align}
		\label{appM:eq:D_prop_dagger}
		\lind_\mathcal{D}(A^\dagger)=\lind_\mathcal{D}(A)^\dagger 
		\quad \iff \quad
		(\lind_\mathcal{D})^{\gamma|\alpha}_{\delta|\beta}=\overline{(\lind_\mathcal{D})_{\gamma|\alpha}^{\delta|\beta}} \,.
	\end{align}
	Then, combining Eq.~\eqref{appM:eq:coordinateExpression2}, Eq.~\eqref{appM:eq:D_prop_or}, and Eq.~\eqref{appM:eq:D_prop_dagger}, we can see that:
	\begin{itemize}
		\item populations and coherences do not mix, and the populations on the diagonal satisfy the classical detailed balance condition. In fact, from Eq.~\eqref{appM:eq:D_prop_or} we see that from $\alpha=\beta$ it follows that $\gamma=\delta$ (assuming that $\pi$ is non-degenerate), and the transition probabilities are related by the standard detailed balance condition:
		\begin{align}
			(\lind_\mathcal{D})^{\gamma|\alpha}_{\gamma|\alpha}\,\pi_\alpha=(\lind_\mathcal{D})^{\alpha|\gamma}_{\alpha|\gamma}\,\pi_\gamma\,.
		\end{align}
	\end{itemize}
	Moreover, the dynamics of the coherences can be split in two cases:
	\begin{itemize}
		\item the one for which  $\frac{\pi_\alpha}{\pi_\beta}=\frac{\pi_\gamma}{\pi_\delta}$, implying the following relation: 
		\begin{align}
			\label{appM:eq:old_case}
			(\lind_\mathcal{D})^{\gamma|\alpha}_{\delta|\beta}\,\pi_\beta =
			(\lind_\mathcal{D})^{\beta|\delta}_{\alpha|\gamma}\,\pi_\delta
			=\overline{(\lind_\mathcal{D})_{\beta|\delta}^{\alpha|\gamma}}\,\pi_\delta\, .
		\end{align}
		This property is satisfied also in the Alicki's definition of detailed balanced generator. 
		\item the additional transitions between coherences, corresponding to the case $\frac{\pi_\alpha}{\pi_\beta}=\frac{\pi_\delta}{\pi_\gamma}$.
		In this case the rate are given by:
		\begin{align}
			\label{appM:eq:new_case}
			(\lind_\mathcal{D})^{\gamma|\alpha}_{\delta|\beta}\,\pi_\beta =
			(\lind_\mathcal{D})^{\beta|\delta}_{\alpha|\gamma}\,\pi_\gamma
			=\overline{(\lind_\mathcal{D})_{\beta|\delta}^{\alpha|\gamma}}\,\pi_\gamma\, .
		\end{align}
		These rates are the only novelty compared with the ones coming from Def.~\ref{cf:def:alicki}, and are the ones that prevent $\lind_\DD$ from commuting with $\LL_\pi\RR_\pi^{-1}$.
	\end{itemize}
	
	In order to justify the preference for Def.~\ref{cf:def:AS} as the quantum generalisation of detailed balance, we argue here that this last case is still physically sensible. Consider indeed two coherences terms $\ket{\alpha}\bra{\beta}$ and   $\ket{\gamma}\bra{\delta}$ such that $\frac{\pi_\alpha}{\pi_\beta}=\frac{\pi_\gamma}{\pi_\delta}$. Then, from Eq.~\eqref{appM:eq:old_case} it follows that the ratio between the currents induced between the two coherences is given by:
	\begin{align}
		\frac{\big|(\lind_\mathcal{D})^{\gamma|\alpha}_{\delta|\beta}\big|}{\big|{(\lind_\mathcal{D})_{\beta|\delta}^{\alpha|\gamma}}\big|}=
		\frac{\pi_\delta}{\pi_\beta}=\frac{\pi_\gamma}{\pi_\alpha}\, .\label{appM:eq:118}
	\end{align}
	The additional freedom given by  Eq.~\eqref{appM:eq:new_case} corresponds to  the possibility of the matrix element $(\lind_\mathcal{D})^{\delta|\alpha}_{\gamma|\beta}$ to be non-zero. It should be noticed, though, that the current between the two coherences is consistent with Eq.~\eqref{appM:eq:118}:
	\begin{align}
		\frac{\big|(\lind_\mathcal{D})^{\delta|\alpha}_{\gamma|\beta}\big|}{\big|({\lind}_\mathcal{D})^{\alpha|\delta}_{\beta|\gamma}\big|}=
		\frac{\pi_\delta}{\pi_\beta}=\frac{\pi_\gamma}{\pi_\alpha}\, .\label{appM:eq:123c}
	\end{align}
	Thus, the difference between the detailed balance condition in Def.~\ref{cf:def:alicki} and the Fisher one (i.e., Def.~\ref{cf:def:AS}), is that the first allows the coherences $\ket{\alpha}\bra{\beta}$ and   $\ket{\gamma}\bra{\delta}$ to communicate but prohibits interactions between $\ket{\alpha}\bra{\beta}$ and   $\ket{\delta}\bra{\gamma}$, while the latter allows the dynamics to connect  both off-diagonal elements, while still keeping the ratio between the two currents detailed balanced. Since there is no clear argument to disregard this second set of transitions, it seems that Def.~\ref{cf:def:AS} should be preferred. 
	
	Finally, before moving on to the characterisation of the Lindbladian operators satisfying Def.~\ref{cf:def:AS} we present here an example of $\lind$ which is detailed balance in the Fisher sense, but not according to Def.~\ref{cf:def:alicki}. Consider a two-levels system equilibrating to the state
	\begin{align}
		\pi_\beta=\frac{\ketbra{0}{0}+e^{-\beta}\ketbra{1}{1}}{1+e^{-\beta}}\, .
	\end{align}
	We consider the Hamiltonian of the system to be completely degenerate, i.e., $H\propto \id$, which implies $\mathcal{U}=0$. Then, it follows from an explicit calculation that the Lindbladian  
	\begin{align}
		\mathcal{L}(\rho)=A\rho A^\dagger -\frac{1}{2}\{\rho,A^\dagger A\}\,,
	\end{align}
	with the jump operator given by $ A=\ketbra{0}{1} + \sqrt{e^{-\beta}}\ketbra{1}{0}$ satisfies $\widetilde{\lind}_f=\lind$, but not $\widetilde{\lind}^o=\lind$.
	
	\subsection{Structural characterisation of Def.~\ref{cf:def:AS}}\label{appM:sec:structuralCh}
	
	This section is devoted to the derivation of the structural form of Lindbladians satisfying the detailed balance condition in Def.~\ref{cf:def:AS}. Restricting the derivation to the case in which $\lind$ commutes with $\LL_\pi\RR^{-1}_\pi$ also proves the equivalence between 
	Def.~\ref{cf:def:alicki} and Def.~\ref{cf:def:breuer}.
	
	With the hindsight of the previous section it is useful to expand the Lindbladian operator in terms of the following eigenbasis:
	\begin{align}
		F^\omega_\alpha := \curbra{\ketbra{\gamma}{\alpha}\bigg|\; \frac{\pi_\gamma}{\pi_\alpha} = e^{\omega}}\,.\label{appM:eq:a122}
	\end{align}
	It is also useful to introduce the function $\beta_\alpha(\omega) := \curbra{\beta|\, \pi_\beta = \pi_\alpha e^{\omega}}$, namely a function that returns the index $\beta$ such that $\frac{\pi_\beta}{\pi_\alpha} = e^\omega$. In order to keep the notation clear we also define $\gamma_\alpha(\omega)$ and $\delta_\alpha(\omega)$ exactly in the same way. The elements of the eigenbasis in Eq.~\eqref{appM:eq:a122} have the property that:
	\begin{align}
		&(F^\omega_\alpha)^\dagger = F^{-\omega}_{\gamma_\alpha(\omega)}\,;\qquad\qquad\pi F^\omega_\alpha = e^{\omega} \,F^\omega_\alpha \,\pi\,.
	\end{align}
	Generically, one can express the action of the dissipator as:
	\begin{align}
		\lind_\mathcal{D}(\rho) :&= \sum_{\substack{\alpha, \beta,\\\gamma,\delta}}\; (\lind_\mathcal{D})^{\gamma|\alpha}_{\delta|\beta}\; \ketbra{\gamma}{\alpha}\, \rho \,(\ketbra{\delta}{\beta})^\dagger =\sum_{\substack{\alpha, \omega,\\\omega_1,\omega_2}}\; (\lind_\mathcal{D})^{\gamma_\alpha(\omega)|\alpha}_{\delta_\alpha(\omega+ \omega_2)|\beta_\alpha(\omega_1)}\; F^{\omega}_\alpha\, \rho \,(F^{\omega-\omega_1+\omega_2}_{\beta_\alpha(\omega_1)})^\dagger\,,
	\end{align}
	where we implicitly defined $ \pi_\gamma/\pi_\alpha=:e^{\omega}$, $ \pi_\beta/\pi_\alpha =: e^{\omega_1}$ and $\pi_\delta/\pi_\gamma =: e^{\omega_2}$. This expression is particularly useful because it allows for a straightforward  application of the constraints in Eq.~\eqref{appM:eq:D_prop_or}. In fact, we have that:
	\begin{align}
		\frac{\pi_\alpha}{\pi_\beta}=\frac{\pi_\gamma}{\pi_\delta}\qquad&\iff\qquad \omega_1 = \omega_2\,;\\
		\frac{\pi_\alpha}{\pi_\beta}=\frac{\pi_\delta}{\pi_\gamma}\qquad&\iff\qquad \omega_1 = -\omega_2 \,.
	\end{align}
	Hence, the sum above can be restricted to the case $\omega_1 = \pm \omega_2$, giving:
	\begin{align}
		\lind_\mathcal{D}(\rho) =&\sum_{\substack{\alpha,\omega\\\omega_1}}\; (\lind_\mathcal{D})^{\gamma_\alpha(\omega)|\alpha}_{\delta_\alpha(\omega+ \omega_1)|\beta_\alpha(\omega_1)}\; F^{\omega}_\alpha\, \rho \,(F^{\omega}_{\beta_\alpha(\omega_1)})^\dagger + \sum_{\substack{\alpha, \omega,\\\omega_1\neq 0}}\; (\lind_\mathcal{D})^{\gamma_\alpha(\omega)|\alpha}_{\delta_\alpha( \omega-\omega_1)|\beta_\alpha(\omega_1)}\; F^{\omega}_\alpha\, \rho \,(F^{\omega-2\omega_1}_{\beta_\alpha(\omega_1)})^\dagger\,.\label{eq:dissipatorSplit}
	\end{align}
	It should be noticed that the case $\omega_1 =0 $ is included in the first sum, so that we have to impose the constraint $\omega_1 \neq0 $ in the second sum.
	In this way $\lind_\DD$ naturally splits in two parts, $\lind_{\DD_1}$ (i.e., the operator in the first sum) corresponding to the component of the dissipator commuting with $\LL_\pi\RR_\pi^{-1}$, and $\lind_{\DD_2}$. Interestingly, the condition $\widetilde\lind_\DD^o = \lind_\DD$ implies $\lind_\DD \equiv\lind_{\DD_1}$, so characterising the latter provides the structural form of  Lindbladians detailed balance according to Def.~\ref{cf:def:alicki}. 
	
	Whereas $\lind_{\DD_1}$ directly commutes with $\LL_\pi\RR_\pi^{-1}$, we can apply a transformation to $\lind_{\DD_2}$ to make it commuting. In particular, it should be noticed that $\lind_{\DD_2}$ transforms under the transposition superoperator $\Theta$ as:
	\begin{align}
		[\Theta\,\lind_{\DD_2}](\rho) &= \sum_{\substack{\alpha, \omega,\\\omega_1\neq 0}}\; (\lind_\mathcal{D})^{\gamma_\alpha(\omega)|\alpha}_{\delta_\alpha( \omega-\omega_1)|\beta_\alpha(\omega_1)}\; \norbra{F^{\omega}_\alpha\, \rho \,(F^{\omega-2\omega_1}_{\beta_\alpha(\omega_1)})^\dagger}^T =\\
		&= \sum_{\substack{\alpha, \omega,\\\omega_1\neq 0}}\; (\lind_\mathcal{D})^{\gamma_\alpha(\omega)|\alpha}_{\delta_\alpha( \omega-\omega_1)|\beta_\alpha(\omega_1)}\; \ketbra{\pi_\alpha e^{\omega-\omega_1}}{\pi_\alpha} \rho \ketbra{\pi_\alpha e^{\omega_1}}{\pi_\alpha e^{\omega}} = \\
		&=\sum_{\substack{\alpha, \omega,\\\omega_1\neq 0}}\; (\lind_\mathcal{D})^{\gamma_\alpha(\omega)|\alpha}_{\delta_\alpha( \omega-\omega_1)|\beta_\alpha(\omega_1)}\; F^{\omega-\omega_1}_\alpha\, \rho \,(F^{\omega-\omega_1}_{\beta_\alpha(\omega_1)})^\dagger=\sum_{\substack{\alpha, \omega,\\\omega_1\neq 0}}\; (\lind_\mathcal{D})^{\gamma_\alpha(\omega+\omega_1)|\alpha}_{\delta_\alpha( \omega)|\beta_\alpha(\omega_1)}\; F^{\omega}_\alpha\, \rho \,(F^{\omega}_{\beta_\alpha(\omega_1)})^\dagger\,,\label{appMeq:d2Transpose}
	\end{align}
	where in the second line we used the abuse of notation $\ket{\pi_\alpha}$ for $\ket{\alpha}$, and in the last line we made an implicit change of variables. From  Eq.~\eqref{appMeq:d2Transpose} we can see that $\Theta\lind_{\DD_2}$ takes a form completely analogous to $\lind_{\DD_1}$. For this reason, it directly follows that:
	\begin{align}
		[\Theta\lind_{\DD_2}, \LL_\pi\RR_\pi^{-1}] = 0\,.
	\end{align}
	This condition allows to lift the characterisation for $\lind_{\DD_1}$ to $\lind_{\DD_2}$, as it will be shown in the following.
	
	We begin by studying $\lind_{\DD_1}$. Notice that the coordinates of $\lind_\DD$ are related by:
	\begin{align}
		(\lind_\mathcal{D})^{\gamma_\alpha(\omega)|\alpha}_{\delta_\alpha(\omega+ \omega_1)|\beta_\alpha(\omega_1)} =e^\omega\, (\lind_\mathcal{D})_{\alpha|\gamma_\alpha(\omega)}^{\beta_\alpha(\omega_1)|\delta_\alpha(\omega+ \omega_1)}\,, \label{appM:eq:appRateRelations}
	\end{align}
	as it can be verified from Eq.~\eqref{appM:eq:coordinateExpression2}. Then, we can proceed with the standard approach to diagonalise $\lind_{\DD_1}$~\cite{alicki1976detailed, carlen2017gradient}. To this end, it is useful to introduce a new basis of operators, given by $\curbra{X^\omega_m} = \curbra{\Delta_i}_{1\leq i\leq d}\cup \curbra{F^\omega_\alpha}_{\omega\neq0}$, where $\curbra{\Delta_i}_{1\leq i\leq d}$ is an orthonormal basis for the diagonal matrices, and $\Delta_1 = \id/\sqrt{d}$. Then, we can rewrite $\lind_{\DD_1}$ in this basis as:
	\begin{align}
		\lind_{\DD_1}(\rho) &= \sum_{\substack{\alpha,\beta}}\; (\lind_\mathcal{D})^{\gamma_\alpha(0)|\alpha}_{\delta_\beta( 0)|\beta}\; F^{0}_\alpha\, \rho \,(F^{0}_{\beta})^\dagger+\sum_{\substack{\alpha,\beta\\\omega\neq 0}}\; (\lind_\mathcal{D})^{\gamma_\alpha(\omega)|\alpha}_{\delta_\beta(\omega)|\beta}\; F^{\omega}_\alpha\, \rho \,(F^{\omega}_{\beta})^\dagger =\\
		&=\sum_{\substack{i,j}}\; D_{i,j}^0\; \Delta_i\, \rho \,(\Delta_j)^\dagger+ \sum_{\substack{\alpha, \beta,\\\omega\neq 0}} \;D_{\alpha,\beta}^\omega\; F^{\omega}_\alpha\, \rho \,(F^{\omega}_{\beta})^\dagger\,,
	\end{align} 
	where we introduced the coefficients $D_{\alpha,\beta}^\omega:=(\lind_\mathcal{D})^{\gamma_\alpha(\omega)|\alpha}_{\delta_\beta(\omega)|\beta}$ and
	\begin{align}
		D^0_{i,j} = \sum_{\alpha,\beta}\;(\lind_\mathcal{D})^{\alpha|\alpha}_{\beta|\beta}\, \bar{U}_{\alpha,i} \,U_{j,\beta}\,,
	\end{align}
	where $U$ is the unitary defined by $\Delta_i := \sum_\alpha U_{i,\alpha} F_\alpha$. In order to make the dissipator explicitly trace preserving, we highlight the terms containing the identity operator. This leads to the form:
	\begin{align}
		\lind_{\mathcal{D}_1}(\rho) =\curbra{K,\rho} +\sum_{\substack{i,j \neq (1,1)\\ \omega}} D^\omega_{i,j}\; X^\omega_i\, \rho \,(X_j^\omega)^\dagger\,,
	\end{align}
	where $K := \frac{D^0_{1,1}}{2 d}\id + \frac{1}{\sqrt{d}}\sum_i D^0_{i,1}\Delta_i$, and there is no additional term thanks to the reality of $D^0_{i,j}$. It should be noticed that $\lind_{\mathcal{D}_2}$ only gives contributions out of the diagonal, so one can impose trace preservation on $\lind_{\mathcal{D}_1}$ alone, i.e., $\lind_{\mathcal{D}_1}^\dagger(\id) = 0$. This directly implies that $K = -\frac{1}{2}\norbra{\sum_{\substack{i,j \neq (1,1)}} D^\omega_{i,j}(X_j^\omega)^\dagger X_i^\omega}$,  so the dissipator can be rewritten in Lindblad form: 
	\begin{align}
		\lind_{\mathcal{D}_1}(\rho) = \sum_{\substack{i,j \neq (1,1)\\ \omega}} D^\omega_{i,j}\; \norbra{X_i^\omega\, \rho \,(X_j^\omega)^\dagger -\frac{1}{2}\curbra{(X_j^\omega)^\dagger X_i^\omega,\rho}}\,.\label{appM:eq:141}
	\end{align}
	The property of the dissipator of being adjoint preserving implies that the matrix $D^\omega_{i,j}$ is Hermitian. Then, there exists a unitary matrix $V$, such that $D^\omega_{i,j} = \sum_{m,n}V^\omega_{i,m}(\lambda^\omega_m \delta^{m}_n)(V^\omega)_{n,j}^\dagger $. We can then define the jump operators as $A_m^\omega := \sum_{i}X^\omega_i V^\omega_{i,m}$. This allows us to recast Eq.~\eqref{appM:eq:141} in the form:
	\begin{align}\label{appM:eq:142}
		\lind_{\mathcal{D}_1}(\rho)  &= \sum_{\substack{i,j \neq (1,1)\\\omega}}\; \lambda^\omega_m \; V^\omega_{i,m}(V^\omega)_{m,j}^\dagger \norbra{X^\omega_i\, \rho \,(X^\omega_j)^\dagger -\frac{1}{2}\curbra{(X^\omega_j)^\dagger X^\omega_i,\rho}} = \\
		&=\sum_{m,\omega} \;\lambda^\omega_m \;\norbra{A^\omega_m\, \rho \,(A^\omega_m)^\dagger -\frac{1}{2}\curbra{(A^\omega_m)^\dagger A^\omega_m,\rho}}\,.
	\end{align}
	We can now characterise the properties of the jump operators and of the rates in the previous equation. First, it should be noticed that since $X_i^\omega$ are eigenoperators of the  modular operator $\LL_\pi\RR_\pi^{-1}$, the same holds for $A_m^\omega$, as the unitary does not mix $X_i^\omega$s  with different $\omega$s. Moreover, Eq.~\eqref{appM:eq:appRateRelations} implies that $D^\omega_{\alpha,\beta} = e^{\omega}\,D^{-\omega}_{\delta_\beta(\omega), \gamma_\alpha(\omega)}$, where we used the same indices as in the equation. This relation shows that $D^{-\omega}_{\delta_\beta(\omega), \gamma_\alpha(\omega)}$ can be diagonalised as $e^{-\omega}\lambda^{\omega}_m\delta_n^m =\sum_{\alpha,\beta}(V^\omega)_{m,\alpha}^\dagger D^{-\omega}_{\delta_\beta(\omega), \gamma_\alpha(\omega)}(V^\omega)_{\beta,n}$. This allows to deduce the following two facts: first, the spectrum of $D^\omega_{\alpha,\beta}$ satisfies $\lambda_{i}^\omega=e^{\omega}\,\lambda_{i}^{-\omega}$; second, since $(X^\omega_{\alpha})^\dagger = X^{-\omega}_{\gamma_\alpha(\omega)}$, it also holds that:
	\begin{align}
		(A^\omega_i)^\dagger = \sum_{\alpha}\;X^{-\omega}_{\gamma_\alpha(\omega)} (V^\omega_{\alpha,i})^\dagger = \sum_{\alpha}\;X^{-\omega}_{\gamma_\alpha(\omega)} V^{-\omega}_{\gamma_\alpha(\omega),i} = A^{-\omega}_i\,.
	\end{align}
	Thus, the dissipator $\lind_{\mathcal{D}_1}$ satisfies the same conditions of Def.~\ref{cf:def:breuer}, namely:
	\begin{enumerate}
		\item $(A_i^{\omega})^\dagger = A_i^{-\omega}$;
		\item $\pi \,A_i^{\omega}\, \pi^{-1}=e^{\omega}\, A_i^{\omega}$ ;
		\item $\lambda_{i}^\omega=e^{\omega}\,\lambda_{i}^{-\omega}$.
	\end{enumerate}
	Since $\lind_{\DD_1}$ is the only component of the dissipator if one uses Def.~\ref{cf:def:alicki}, this proves the equivalence between this notion of detailed balance and the structural characterisation in Def.~\ref{cf:def:breuer}.
	
	We can now pass to characterise $\lind_{\mathcal{D}_2}$. Thanks to Eq.~\eqref{appMeq:d2Transpose} we can rewrite it as:
	\begin{align}
		\lind_{\mathcal{D}_2}(\rho)&=\sum_{\substack{\alpha, \omega,\\\omega_1\neq 0}}\; (\lind_\mathcal{D})^{\gamma_\alpha(\omega+\omega_1)|\alpha}_{\delta_\alpha( \omega)|\beta_\alpha(\omega_1)}\;(F^{\omega}_\alpha\, \rho \,(F^{\omega}_{\beta_\alpha(\omega_1)})^\dagger)^T=\sum_{\substack{\alpha\neq\beta,\\\omega}}\; (\lind_\mathcal{D})^{\gamma_\beta(\omega)|\alpha}_{\delta_\alpha( \omega)|\beta}\;F^{\omega}_{\beta}\, \rho^T \,(F^{\omega}_\alpha)^\dagger =\\
		&=\sum_{\substack{\alpha\neq\beta,\\\omega}}\; (\lind_\mathcal{D})^{\delta_\beta(\omega)|\alpha}_{\gamma_\alpha( \omega)|\beta}\;F^{\omega}_{\beta}\, \rho^T \,(F^{\omega}_\alpha)^\dagger  \,,\label{appM:eq:144}
	\end{align}
	where we eliminated the dependence on $\omega_1$ by enforcing the constraint $\alpha\neq\beta$. Finally, in the last line we exchanged the dummy indexes $\gamma$ and $\delta$ to highlight the analogy with the other part of the Lindbladian. Indeed, define the matrix $\widetilde{D_{\alpha,\beta}^\omega}:= (\lind_\mathcal{D})^{\delta_\beta(\omega)|\alpha}_{\gamma_\alpha( \omega)|\beta}$ for any $\alpha\neq\beta$ and zero on the diagonal. It is interesting to compare it to the off-diagonal elements of $D_{\alpha,\beta}^\omega = (\lind_\mathcal{D})^{\gamma_\alpha(\omega)|\alpha}_{\delta_\beta( \omega)|\beta}$: as it can be seen, the two are related by an exchange $\gamma_\alpha(\omega)\leftrightarrow\delta_\beta(\omega)$. Moreover, thanks to Eq.~\eqref{appM:eq:123c} it also holds that:
	\begin{align}
		\widetilde{D_{\alpha,\beta}^\omega}  = e^\omega\,\widetilde{D}_{\delta_\beta(\omega),\gamma_\alpha(\omega)}^{-\omega}\,,
	\end{align}
	which shows the analogy with $D_{\alpha,\beta}^\omega$  even further.  Using Eq.~\eqref{appM:eq:D_prop_dagger} it also follows that $\widetilde{D_{\alpha,\beta}^\omega}$ is Hermitian, so there exists a unitary matrix $W$ such that  $\widetilde{D_{\alpha,\beta}^\omega} =\sum_{m,n} W^\omega_{\alpha,m}(\mu^\omega_m \delta^{m}_n)(W^\omega)_{n,\beta}^\dagger $. We can then define the jump operators $B_m^\omega := \sum_{\alpha}(W^\omega_{m,\alpha})^\dagger F^\omega_\alpha $. This allows to rewrite $\lind_{\mathcal{D}_2}$ as:
	\begin{align}
		\lind_{\mathcal{D}_2}(\rho)&= \sum_{\substack{\alpha\neq\beta,\\\omega}}\; W^\omega_{\alpha,m}(\mu^\omega_m \delta^{m}_n)(W^\omega)_{n,\beta}^\dagger \;F^{\omega}_{\beta}\, \rho^T \,(F^{\omega}_\alpha)^\dagger  =\sum_{m}\;\mu^\omega_m\,B_m^\omega \,\rho^T (B_m^\omega )^\dagger\,.
	\end{align}
	In analogy with the previous case it also holds that $\mu^\omega_i = e^\omega\mu^{-\omega}_i$ and $(B_m^\omega)^\dagger = B_m^{-\omega}$, together with the fact that $B_m^\omega$ is an eigenoperator of $\LL_\pi \RR_\pi^{-1}$. Still, there is one crucial difference with $\lind_{\mathcal{D}_1}$: since all the diagonal elements of $\widetilde{D_{\alpha,\beta}^\omega}$ are zero, this matrix is traceless, meaning that the sum of the eigenvalues is also zero. Hence, whereas $\lambda_i^\omega\geq0$ for all $i$ and $\omega$, we have the extra constraint $\sum_i \mu_i^\omega = 0 $, implying the negativity of some of the $\mu_i^\omega$. Putting everything together, we finally obtain the characterisation in Eq.~\eqref{cf:eq:structAS}.
	
	\section{Derivations from Sec.~\ref{sec:FisherOperatorProperties}}\label{box:fisherComputations}
	
	We report here the explicit computation that were omitted in Sec.~\ref{sec:FisherOperatorProperties}. We start by discussing how the defining properties of the standard monotone functions are reflected on the Fisher information operators. First, the property that $f(1)=1$ implies that $\J_f\big|_{\pi}$ reduces to the multiplication by $\pi$ for commuting operators, as it was shown in Eq.~\eqref{eq:classicalFisher}. We repeat the derivation here for completeness: it should be noticed that the modular operator acts trivially on the commutant of $\pi$, that is $\LL_\pi\RR_\pi^{-1}[A] = A$  for $[\pi,A] = 0$. For this reason $f(\LL_\pi\RR_\pi^{-1})[A]  = f(\idO)[A] = A$, where in the last step we used the fact that $f(1)=1$. Hence, if $[\pi,A] = 0$, then $\J_f\big|_{\pi}[A] = \RR_\pi [A] = A\pi$.
	
	We can now pass to the second condition that standard monotone functions should satisfy, namely ${f(x) = x f(x^{-1})}$. This enforces $\J_f\big|_{\pi}$ to be adjoint preserving. Indeed, one has that, if $A  = A^\dagger$, it follows that:
	\begin{align}
		(\J_f\big|_{\pi} [A])^\dagger &=  (\RR_\pi f\boldsymbol{(}\LL_{\pi }\RR_{\pi  }^{-1}\boldsymbol{)} [A])^\dagger = \LL_\pi f\boldsymbol{(}\RR_{\pi }\LL_{\pi  }^{-1}\boldsymbol{)} [A^\dagger]
		=\cancel{\LL_\pi} \,\cancel{\LL_{\pi  }^{-1}}\RR_{\pi } f\boldsymbol{(}\LL_{\pi  }\RR_{\pi }^{-1}\boldsymbol{)} [A]=\J_f\big|_{\pi} [A]\,,
	\end{align}
	where we used the fact that $(\LL_\pi[A])^\dagger = \RR_\pi[A^\dagger]$ (and similarly for $\RR_\pi$), the commutation $[\LL_\pi, \RR_\pi] = 0$ and the property ${f(x) = x f(x^{-1})}$.
	
	Finally, from the monotonicity of $f$ it follows that (see Thm.~\ref{cf:thm:Petz}):
	\begin{align}
		\Tr{\Phi(A) \, \J_f^{-1}\big|_{\Phi(\pi)} [ \Phi(A)]} \leq \Tr{A \, \J_f^{-1}\big|_{\pi} [ A]}\,,\label{cf:eq:A30}
	\end{align}
	for every self-adjoint operator $A$ and CPTP map $\Phi$. By polarisation this means that $\Phi^\dagger\norbra{\J_f^{-1}\big|_{\Phi(\pi)}}\Phi\leq \J_f^{-1}\big|_{\pi}$, which is the first of the two conditions in Eq.~\ref{eq:monotonicityPropertiesPolar}. By multiplying the two sides of the inequality by $\J_f^{1/2}\big|_{\pi}$ on the right and on the left we obtain:
	\begin{align}
		\J_f^{1/2}\big|_{\pi}\Phi^\dagger\norbra{\J_f^{-1}\big|_{\Phi(\pi)}}\Phi\J_f^{1/2}\big|_{\pi} = \norbra{\J_f^{1/2}\big|_{\pi}\Phi^\dagger\J_f^{-1/2}\big|_{\Phi(\pi)}}\norbra{\J_f^{1/2}\big|_{\pi}\Phi^\dagger\J_f^{-1/2}\big|_{\Phi(\pi)}}^\dagger\leq \idO\,.
	\end{align}
	This inequality is of the form $AA^\dagger \leq \idO$, for $A = \J_f^{1/2}\big|_{\pi}\Phi^\dagger\J_f^{-1/2}\big|_{\Phi(\pi)}$. But then it follows that also  $A^\dagger A \leq \idO$, which expands to:
	\begin{align}
		\J_f^{-1/2}\big|_{\Phi(\pi)}\Phi\,\J_f\big|_{\pi}\Phi^\dagger\J_f^{-1/2}\big|_{\Phi(\pi)}\leq \idO\,.
	\end{align}
	Multiplying both sides on the left and on the right by $\J_f^{1/2}\big|_{\Phi(\pi)}$ we finally obtain $\Phi\norbra{\J_f\big|_{\pi}}\Phi^\dagger\;\leq\; {\J_f\big|_{\Phi(\pi)}}$, which proves the second condition in Eq.~\ref{eq:monotonicityPropertiesPolar}.
	
	After having solved the question about the properties of $\J_f\big|_{\pi}$ we present here the explicit derivation of Eq.~\eqref{eq:157} and Eq.~\eqref{eq:158}. First it should be noticed that we can rewrite $\J_{f_\lambda} \big|_{\pi}$ as:
	\begin{align}
		\J_{f_\lambda} \big|_{\pi}[A] &= \RR_\pi \,f_\lambda\norbraB{\LL_\pi\RR_\pi^{-1}}[A] =\norbra{\frac{1+\lambda}{2}} \norbra{\frac{\LL_\pi\RR_\pi^{-1}}{(\LL_\pi\RR_\pi^{-1}+\lambda)} +\frac{\LL_\pi\RR_\pi^{-1}}{(1+\lambda\LL_\pi\RR_\pi^{-1})}  }[ A\pi]=\\
		&=\norbra{\frac{1+\lambda}{2}} \norbra{\frac{\LL_\pi\cancel{\RR_\pi^{-1}}}{\cancel{\RR_\pi^{-1}}(\LL_\pi+\lambda\RR_\pi)} +\frac{\LL_\pi\cancel{\RR_\pi^{-1}}}{\cancel{\RR_\pi^{-1}}(\RR_\pi+\lambda\LL_\pi)}  }[ A\pi]=\\
		&= \norbra{\frac{1+\lambda}{2}} \norbra{(\LL_\pi+\lambda\RR_\pi)^{-1} + (\RR_\pi+\lambda\LL_\pi)^{-1}}[\pi A\pi] \,,
	\end{align}
	where we used the commutation between $\LL_\pi$ and $\RR_\pi$ to simplify the manipulations. Integrating over $\lambda$ then gives Eq.~\eqref{eq:157}. A similar manipulation can be carried out for $\J_{f}^{-1} \big|_{\pi}$, which gives:
	\begin{align}
		\J_{f}^{-1} \big|_{\pi}[A] &= \J_{Tf} \big|_{\pi^{-1}}[A]=\int_0^1\de\mu_{f}(\lambda)\; \norbra{\frac{1+\lambda}{2}} \norbra{(\LL_\pi^{-1}+\lambda\RR_\pi^{-1})^{-1} + (\RR_\pi^{-1}+\lambda\LL_\pi^{-1})^{-1}}[\pi^{-1} A\pi^{-1}]= \\
		&=\int_0^1\de\mu_{f}(\lambda)\;\norbra{\frac{1+\lambda}{2}} \norbra{\frac{\LL_\pi^{-1}\RR_\pi^{-1}}{(\LL_\pi^{-1}+\lambda\RR_\pi^{-1})} +\frac{\LL_\pi^{-1}\RR_\pi^{-1}}{(\RR_\pi^{-1}+\lambda\LL_\pi^{-1})}  }[ A]=\\
		&= \int_0^1\de\mu_{f}(\lambda)\; \norbra{\frac{1+\lambda}{2}} \norbra{(\LL_\pi+\lambda\RR_\pi)^{-1} + (\RR_\pi+\lambda\LL_\pi)^{-1}}[ A]\,,
	\end{align}
	which proves Eq.~\eqref{eq:158}.
	
	\section{Ordering of symmetrised contrast functions}\label{app:orderContrast}
	
	In this section we prove that symmetrised contrast functions satisfy the ordering:
	\begin{align}
		H^{\rm symm}_{Lf_B} ( \rho || \sigma)\leq H_g^{\rm symm} ( \rho || \sigma) \leq H_{Lf_H}^{\rm symm} ( \rho || \sigma)\,.\label{eq:E1}
	\end{align}
	The two main ingredients we use are the integral expression in Eq.~\eqref{cf:eq:AppIntegralExpression}, namely:
	\begin{align}
		H^{\rm symm}_g(\rho || \sigma)  &=  \frac{1}{2}\int_0^1 \de N_g(s)\;\Tr{(\rho-\sigma)\norbra{(\LL_\sigma + s \RR_\rho)^{-1}+(\LL_\rho + s \RR_\sigma)^{-1}}[(\rho-\sigma)]}\label{eq:E2}
	\end{align}
	and the following Lemma:
	\begin{lemma}\label{lemmino}
		For any two full-rank states $\rho$ and $\sigma$ and any $s\in[0,1]$ the following chain of operator inequalities hold:
		\begin{align}
			2\,(\LL_\sigma+\RR_\rho)^{-1}\leq \norbra{\frac{1+s}{2}}\norbra{(\LL_\sigma+s\RR_\rho)^{-1}+(\RR_\rho+s\LL_\sigma)^{-1}}\leq \frac{1}{2}\,(\LL_\sigma^{-1} + \RR_\rho^{-1})\,.\label{eq:E3}
		\end{align}
	\end{lemma}
	\begin{proof}
		The three operators in Eq.~\eqref{eq:E3} are all diagonal in the basis given by $\{\ketbra{\sigma_j}{\rho_i}\}$, as it can be readily verified from the equalities:
		\begin{align}
			2\,(\LL_\sigma+\RR_\rho)^{-1}[\ketbra{\sigma_j}{\rho_i}]&=\frac{2}{\sigma_j+\rho_i}\ketbra{\sigma_j}{\rho_i}\,;\\
			\norbra{\frac{1+s}{2}}\norbra{(\LL_\sigma+s\RR_\rho)^{-1}+(\RR_\rho+s\LL_\sigma)^{-1}}[\ketbra{\sigma_j}{\rho_i}]&=\norbra{\frac{1+s}{2}}\norbra{\frac{1}{\sigma_j+s\,\rho_i}+\frac{1}{s\,\sigma_j+\rho_i}}\ketbra{\sigma_j}{\rho_i}\,;\\
			{\frac{1}{2}}\,(\LL_\sigma^{-1} + \RR_\rho^{-1})[\ketbra{\sigma_j}{\rho_i}]&={\frac{1}{2}}\,\norbra{\frac{1}{\sigma_j}+\frac{1}{\rho_i}}\ketbra{\sigma_j}{\rho_i}\,.
		\end{align}
		Then the claim follows from the chain of inequalities:
		\begin{align}
			\frac{2}{1+x}\leq \norbra{\frac{1+s}{2}}\norbra{\frac{1}{1+s\,x} +  \frac{1}{s+x} }\leq \frac{1}{2}+\frac{1}{2x}\,,
		\end{align}
		which holds for any positive $x$ and $s\in[0,1]$. Moreover, the function in the middle is monotonically decreasing in $s$.
	\end{proof}
	We are now ready to prove the claim. Indeed, starting from Eq.~\eqref{eq:E2}, and using Eq.~\eqref{eq:cf:102o} we have:
	\begin{align}
		H^{\rm symm}_g(\rho || \sigma)  &=  \frac{1}{2}\int_0^1 \de N_g(s)\;\Tr{(\rho-\sigma)\norbra{(\LL_\sigma + s \RR_\rho)^{-1}+(\LL_\rho + s \RR_\sigma)^{-1}}[(\rho-\sigma)]} =\\
		&= \frac{1}{2}\int_0^1\de\mu_{f}(s)\; \norbra{\frac{1+s}{2}}\Tr{(\rho-\sigma)\norbra{(\LL_\sigma + s \RR_\rho)^{-1}+(\RR_\rho + s \LL_\sigma)^{-1}}[(\rho-\sigma)]}\,,
	\end{align}
	where in the second line we took the adjoint of the second term in the trace. Since $\de\mu_{f}(s)$ is a probability distribution and thanks to Lemma~\ref{lemmino}, any contrast function is smaller than the one corresponding to $\delta\mu_{f_H}(s) = \delta(s)$ and larger than the one corresponding to $\delta\mu_{f_B}(s) = \delta(s-1)$. Hence, for any symmetrised contrast function Eq.~\eqref{eq:E1} holds.
	
	\section{Integral expression of $(\LL_\sigma+\RR_\rho)^{-1}$}\label{app:inverseAnticomm}
	
	Consider the equation:
	\begin{align}
		(\LL_\sigma+\RR_\rho)[A]= \sigma A+ A\rho = X\,.\label{f1}
	\end{align}
	One can implicitly solve for $A$ as $A= (\LL_\sigma+\RR_\rho)^{-1}[X]$. In this section, we prove that this solution can be rewritten as:
	\begin{align}
		A = \int_0^{\infty} \dt \; e^{-t\sigma}\,X\,e^{-t\rho}\,.\label{cf:eq:148x}
	\end{align}
	If this were the case, it follows from  Eq.~\eqref{f1} that multiplying the right hand side of the equation by 
	$(\LL_\sigma+\RR_\rho)$ would give back $X$. Indeed, this can be straightforwardly verified as:
	\begin{align}
		(\LL_\sigma+\RR_\rho)\int_0^{\infty} \dt \; e^{-t\sigma}\,X\,e^{-t\rho} &= \int_0^{\infty} \dt \; \norbra{\sigma\,e^{-t\sigma}\,X\,e^{-t\rho}+e^{-t\sigma}\,X\,e^{-t\rho}\,\rho} =\\
		&= -\int_0^{\infty} \dt \; \norbra{\frac{\de}{\dt}\,e^{-t\sigma}\,X\,e^{-t\rho}} = \norbra{e^{-t\sigma}\,X\,e^{-t\rho}}\bigg|_{t=\infty}^{t=0}= X \,,
	\end{align}
	which proves the claim.

	\bibliography{Bib.bib}
\end{document}